\documentclass[a4paper,12pt]{extarticle}

\usepackage[tmargin=5.8cm]{geometry}

\usepackage{chngcntr}
\counterwithin{paragraph}{subsubsection} 

\usepackage{cancel}
\usepackage{graphicx}
\usepackage{amscd}
\usepackage{pb-diagram}

\usepackage{newclude}

\usepackage [english]{ babel } 
\usepackage{float} 
\usepackage[all]{xy} 
\usepackage{amsmath,amsfonts,amssymb}
\usepackage{color}
\usepackage[latin1]{inputenc}
\usepackage{setspace}

\usepackage{pstricks}
\usepackage{pst-node}
\usepackage{amstext} 
\usepackage{pstricks-add}

\usepackage[breaklinks=true,urlcolor=blue,citecolor=blue,linkcolor=blue]{hyperref}
\usepackage{tocloft}

\newcommand{\listdefinitionsname}{\Large{List of Definitions}}
\newlistof{mydefinitions}{def}{\listdefinitionsname}
\newcommand{\mydefinitions}[1]{%
\addcontentsline{def}{mydefinitions}{\protect\numberline{\thedefinition}#1}\par}

\newcommand{\ws}{\textcolor{white}{e}}

\newtheorem{theorem}{Theorem}

\newtheorem{claim}[theorem]{Claim}

\newtheorem{corollary}[theorem]{Corollary}

\newtheorem{definition}{Definition}

\newtheorem{lemma}[theorem]{Lemma}

\newtheorem{proposition}[theorem]{Proposition}
\newtheorem{remark}[theorem]{Remark}

\newenvironment{proof}[1][Proof]{\textbf{#1.} }{\ \rule{0.5em}{0.5em}}

\newcommand{\enter} {\vskip 0.3cm}

\newcommand{\A}{\mathbb A}

\newcommand{\N}{\mathbb N}
\newcommand{\Z}{\mathbb Z}
\newcommand{\C}{\mathbb C}

\newcommand{\klk}{,\ldots,}
\newcommand{\ol}{\overline}
\newcommand{\MM}{\ol{\mathcal{M}}}

\newcommand{\M}{{\mathcal{M}}}

\newcommand{\wt}{\widetilde}

\DeclareMathAlphabet{\mathpzc}{OT1}{pzc}{m}{it} 

\begin{document}


\begin{center}
\textbf{\LARGE{Software Engineering and Complexity in Effective Algebraic Geometry
\footnote{
Research partially supported by the following Argentinian, Belgian and Spanish grants: CONICET PIP 2461/01, UBACYT 20020100100945, PICT--2010--0525, FWO G.0344.05, MTM2010-16051. 
}}}
\end{center}

\begin{center}
\Large{Joos Heintz\footnote{Departamento de Computaci\'on, Universidad de Buenos Aires and CONICET, Ciudad Universitaria, Pab. I, 1428 Buenos Aires, Argentina, and Departamento de Matem\'aticas, Estad\'istica y Computaci\'on, Facultad de Ciencias, Universidad de Cantabria, Avda. de los Castros s/n, 39005 Santander, Spain. joos@dc.uba.ar \& joos.heintz@unican.es}, Bart Kuijpers\footnote{Database and Theoretical Computer Science Research Group, Hasselt University, Agoralaan, Gebouw D, 3590 Diepenbeek, Belgium. bart.kuijpers@uhasselt.be}, Andr\'es Rojas Paredes\footnote{Departamento de Computaci\'on, Universidad de Buenos Aires, Ciudad Universitaria, Pab. I, 1428 Buenos Aires, Argentina. arojas@dc.uba.ar}\\}
\end{center}

\begin{center}
\textit{Dedicated to the memory of Jacques Morgenstern\\
whose ideas inspired this work}
\end{center}

\begin{center}
\today 
\end{center}

\enter
\begin{abstract}
One may represent polynomials not only by their coefficients but also by arithmetic circuits which evaluate them. This idea allowed in the last fifteen years considerable complexity progress in effective polynomial equation solving. We present a circuit based computation model which captures all known symbolic elimination algorithms in effective Algebraic Geometry and exhibit a class of simple elimination problems which require exponential size circuits to be solved in this model. This implies that the known, circuit based elimination algorithms are already optimal.
\end{abstract}

\begin{quote}
\small{\textit{
Keywords: Robust parameterized arithmetic circuit, isoparametric routine, branching parsimonious algorithm, flat family of zero dimensional elimination problems.\\
MSC: 68Q05, 68Q17, 68Q60, 68N30, 14E99, 14Q99}}
\end{quote}


\section[Introduction]{\large{Introduction}}

We introduce and motivate a new computation model which is well--adapted to scientific computing in effective Algebraic Geometry and especially to elimination theory. This model is based on the symbolic manipulation of arithmetic circuits which evaluate rational functions.

Most algorithms in effective Algebraic Geometry may be formulated as routines which operate on polynomials and rational functions and this suggests the representation of these mathematical entities by arithmetic circuits. Thus, we shall consider arithmetic circuits as objects (in the sense of object oriented programming) which become mapped into abstract data types consisting of polynomials and rational functions. Inputs and outputs of the routines of our computation model will therefore be arithmetic circuits. These circuits will appear as \emph{parameterized} in the sense that they depend on two distinct ingredients, namely on piecewise rational functions called ``basic parameters'' and indeterminates, called ``input variables''. On the other hand, so called ``elementary routines'' will constitute the basic building block of our computation model. They will be branching--free and therefore our parameterized arithmetic circuits will be branching--free too.

However, in effective Algebraic Geometry, divisions are sometimes unavoidable and divisions may lead to branchings. Nevertheless, in typical situations, they may be replaced by limit processes. In order to capture this situation, we introduce in Section \ref{sec:Model-circuits} the notion of a \emph{robust} parameterized arithmetic circuit.
 
An important issue will be the concept of \emph{well behavedness} of routines, under certain modifications of the input circuits. This concept will appear in Section \ref{sec:Model-discussion} in different disguises, called well behavedness under \emph{restrictions} and \emph{reductions}, (output) \emph{isoparametricity} and \emph{coalescence}.

All these technical notions have in common that they allow to formulate algorithmical restrictions on the routines of our computation model which are motivated by specific quality attributes of programs in Software Engineering.

In this sense we establish first in Section \ref{sec:Model-discussion-simplified} a branching--free variant of our computation model which operates on robust parameterized arithmetic circuits. The algorithms captured by this variant are the elementary routines mentioned before.

In order to capture the whole spectrum of really existing elimination algorithms in Algebraic Geometry, we extend our computation model in Section \ref{sec:Model-discussion-programs and algorithms} admitting some limited branchings. The resulting algorithms are called \emph{branching parsimonious}. Moreover we introduce the concept of a \emph{procedure} as a branching parsimonious algorithm with a particular architecture. Procedures are well suited to discuss computational issues in effective elimination theory. In Section \ref{sec:Model-Applications}, we apply our computation model to this task.

Before we are going to enter into the details of these applications we look closer to our computation model. The basic construction method of elementary routines is recursion. Mimicking the directed acyclic graph structure of a given robust parameterized arithmetic circuit we compose the graphs of previously fixed computations in order to obtain a new parameterized arithmetic circuit. To guarantee that the resulting circuit is again robust, we require that these compositions behave well under restrictions. Moreover, we require that at corresponding nodes the intermediate results of the two circuits become directly linked by a continuous, piecewise rational map and call this requirement isoparametricity. This means that our computation model includes also \emph{formal specifications}. We develop this aspect in Sections \ref{sec: A specification language for circuits} and \ref{Isoparametricity and program specification}.

In Section \ref{sec:Model-Applications} we use our computation model to show that already very elementary elimination problems require exponential time for their solution (see Theorem \ref{def: main theorem model}, Proposition \ref{proposition A} and Theorem \ref{proposition C D} below).

In particular, we exhibit in Section \ref{Arithmetization techniques for Boolean circuits} a family of parameterized boolean circuits whose (standard) arithmetizations represent an elimination problem which requires exponential time to be solved in our model (see Theorem \ref{theorem 3} below).

As a major outcome of this paper we exhibit in Section \ref{independent of the model} an infinite family of parameter dependent elimination polynomials which require essentially division--free, \emph{robust} parameterized arithmetic circuits of exponential size for their evaluation, whereas the circuit size of the corresponding input problems grows only polynomially. We observe that essentially division--free, robust parameterized arithmetic circuits for elimination polynomials capture the intuitive meaning of an algorithmic solution with few equations and branchings of the underlying elimination problem. 

The proof of this result, which is absolutely new in his kind, is astonishly elementary and simple.

In Section \ref{divisions and blow ups} we arrive at the conclusion that our method to show lower complexity bounds consists of counting how many steps are necessary to decompose a given rational map into a sequence of ``simple'' blow ups and a polynomial map.

Finally in Section \ref{Comments on complexity models for geometric elimination} we establish a link between our computation model and our lower bound results with other complexity views in geometric elimination theory. In this context we discuss the BSS--model of \cite{BSS89} and the view of interactive protocols.

Our computation model and complexity results are based on the concept of a \emph{geometrically robust constructible map}. This concept was introduced in \cite{GHMS09} and we develop it further in Section \ref{sec: geometry}, which is devoted to the algebraic geometric underpinning of the present paper.

The relevance of the lower complexity bounds of this paper for elimination problems depends on the ``naturalness'' of the computation model. Therefore we emphasize throughout this article the arguments which justify our computation model. Of course, these arguments cannot be entirely of mathematical nature. In this paper they are borrowed from Software Engineering which constitutes a discipline which analyzes and qualifies practical programming issues. In these terms we show that a circuit based algorithm which solves most elementary parametric elimination problems and which is programmed under the application of the most common rules of Software Engineering, can never be efficient.

This paper is based on the idea to represent, e.g., in the case of elimination algorithms, polynomials by arithmetic circuits, which are not considered as algorithms but as data structures. This has sometimes certain advantages. For example, the generic $n\times m$ determinant has as polynomial in the matrix entries $n!$ terms but it can be evaluated by division and branching--free arithmetic circuit of size $O(n^5)$ (\cite{Mul87}). More generally, an arbitrary elimination polynomial of degree $\delta$ and $m$ variables has always a circuit representation which is essentially of order $\delta^{O(1)}$, whereas the representation by its coefficients may become of order $\delta^{\Omega(n)}$ (see \cite{Giusti1} and \cite{Giusti2} for details).

The idea of the circuit representation of polynomials was introduced in Theoretical Computer Science at the beginning of the eighties by the first author in collaboration with Malte Sieveking (Frankfurt/Main). However, the first publications on this subject treated only the case of the elimination of a single variable (see e.g. \cite{HeSie81}, \cite{Ka88} and \cite{FIK86}).

In the case of the simultaneous elimination of several variables, substantial progress was made in the nineties by the first author in collaboration with Marc Giusti (Paris) and Jacques Morgenstern (Nice). A series of paper give account of this development (see e.g. \cite{GH93}, \cite{HeintzMorgenstern93}, \cite{Giusti1} and \cite{Giusti2}). In particular the work with Jacques Morgenstern lead to new views which influenced this paper.


\section[Concepts and tools from Algebraic Geometry]{\large{Concepts and tools from Algebraic Geometry}}\label{sec: geometry}

In this section, we use freely standard notions and notations from Commutative Algebra and Algebraic Geometry. These can be found for example in \cite{Lang93}, \cite{ZaSa60}, \cite{Kunz85} and \cite{Shafarevich94}. In Sections~\ref{sec: geometry-basic-maps} and~\ref{sec: geometry-robust}, we introduce the notions and definitions which constitute our fundamental tool for the modelling of elimination problems and algorithms. Most of these notions and their definitions are taken from \cite{GHMS09}.

\subsection{Basic notions and notations} \label{sec: geometry-basic}

For any $n\in\N$, we denote by $\A^n:=\A^n(\C)$ the $n$--dimensional affine space $\C^n$ equipped with its respective Zariski and Euclidean topologies over $\C$. In algebraic geometry, the Euclidean topology of $\A^n$ is also called the {\em strong topology}. We shall use this terminology only exceptionally. 

Let $X_1,\ldots,X_n$ be indeterminates over $\C$ and let $X:=(X_1,\ldots,X_n)$. We denote by 
$\C[X]$ the ring of polynomials in the variables $X$ with complex coefficients.

Let $V$ be a closed affine subvariety of $\A^n$. As usual, we write $\dim V$ for the dimension of the variety $V$. Let $C_1,\dots,C_s$ be the irreducible components of $V$. For $1\leq j\leq s$ we define the degree of $C_j$ as the number of points which arise when we intersect $C_j$ with $\dim C_j$ many generic affine hyperplanes of $\A^n$. Observe that this number is a well--determined positive integer which we denote by $\deg C_j$. The \emph{(geometric) degree} $\deg V$ of $V$ is defined by $\deg V:= \sum_{1\leq j\leq s} \deg C_j$. This notion of degree satisfies the so called B\'ezout Inequality. Namely, for another closed affine subvariety $W$ of $\A^n$ we have $\deg V\cap W \leq \deg V \cdot \deg W.$ 

For details we refer to \cite{HeintzTesis83}, where the notion of geometric degree was introduced and the B\'ezout Inequality was proved for the first time (other references are \cite{Fu84} and \cite{Vo84}).    

For $f_1,\dots,f_s,g \in \C[X]$ we  shall use the notation $\{f_1=0,\dots,f_s=0\}$ 
in order to denote the closed affine subvariety $V$ of $\A^n$ defined by $f_1,\dots,f_s$ 
and the notation $\{f_1=0,\dots,f_s=0,g\not=0\}$ in order to denote the Zariski open subset $V_g$ of $V$ defined by the intersection of $V$ with the complement of $\{g=0\}$. %
Observe that $V_g$ is a locally closed affine subvariety of $\A^n$ whose coordinate ring
is the localization ${\C[V]}_g$ of $\C[V]$.

We denote by $I(V):= \{f \in \C[X]: f(x)=0$ for any $x \in V \}$ the ideal of definition of $V$ in $\C[X]$ and by $\C[V]:= \{\varphi:V \to \C \ws; \textrm{ there exists\ } f \in \C[X] \textrm{\ with\ } {\varphi}(x)=f(x) \textrm{\ for any\ } x \in V\}$ its coordinate ring. Observe that $\C[V]$ is isomorphic to the quotient $\C$--algebra $\C[V]:=\C[X]/I(V)$. If $V$ is irreducible, then $\C[V]$ is zero--divisor free and we denote by $\C(V)$ the field formed by the rational functions of $V$ with maximal domain ($\C(V)$ is called the rational function field of $V$). Observe that $\C(V)$ is isomorphic to the fraction field of the integral domain $\C[V]$.

In the general situation where $V$ is an arbitrary closed affine
subvariety of $\A^n$, the notion of a rational function of $V$ has
also a precise meaning. The only point to underline is that the domain, say $U$, of a rational function of $V$ has to be a maximal Zariski open and dense subset of $V$ to which the given rational function can be extended. In particular, $U$ has a nonempty intersection with any of the irreducible components of $V$.

As in the case where $V$ is irreducible, we denote by $\C(V)$ the $\C$--algebra formed by the rational functions of $V$. In algebraic terms, $\C(V)$ is the total quotient ring of $\C[V]$ and is isomorphic to the direct product of the rational function fields of the irreducible components of $V$.

Let be given a partial map $\phi: V \dashrightarrow W$, where $V$ and $W$ are closed subvarieties of some affine spaces $\A^n$ and $\A^m$, and let $\phi_1\klk\phi_m$ be the components of $\phi$. With these notations we have the following definitions:

\begin{definition}[Polynomial map] \label{def: pol. map}
The map $\phi$ is called a morphism of affine varieties or just polynomial map if the complex valued functions $\phi_1\klk\phi_m$ belong to $\C[V]$. Thus, in particular, $\phi$ is a total map.
\end{definition}
\mydefinitions{\label{def: pol. map} Polynomial map}
 
\begin{definition}[Rational map] \label{def: rational map}
We call $\phi$ a rational map of $V$ to $W$, if the domain $U$ of $\phi$ is a Zariski open and dense subset of $V$ and $\phi_1\klk\phi_m$ are the restrictions of suitable rational functions of $V$ to $U$. 
\end{definition}
\mydefinitions{\label{def: rational map} Rational map}

Observe that our definition of a rational map differs from the usual one in Algebraic Geometry, since we do not require that the domain $U$ of $\phi$ is maximal. Hence, in the case $m:=1$, our concepts of rational function and rational map do not coincide (see also \cite{GHMS09}).


\subsubsection{Constructible sets and constructible maps}
\label{sec: geometry-basic-constuctible}
Let $\mathcal{M}$ be a subset of some affine space $\A^n$ and, for a given nonnegative integer $m$, let $\phi:\mathcal{M}\dashrightarrow \A^m$ be a partial map. 

\begin{definition}[Constructible set] \label{def: constr. set}
We call the set $\mathcal{M}$ {\em constructible} if $\mathcal{M}$ is definable by a Boolean combination of polynomial equations.
\end{definition} 
\mydefinitions{\label{def: constr. set} Constructible set}

A basic fact we shall use in the sequel is that if
$\mathcal{M}$ is constructible, then its Zariski closure is equal to
its Euclidean closure (see, e.g., \cite{Mumford88}, Chapter I, \S 10, Corollary
1). In the same vein we have the following definition.

\begin{definition}[Constructible map] \label{def: constr. map}
We call the partial map $\phi$ {\em constructible} if the graph of $\phi$ is constructible as a subset of the affine space $\A^n\times \A^m$.
\end{definition} 
\mydefinitions{\label{def: constr. map} Constructible map}

We say that $\phi$ is {\em polynomial} if $\phi$ is the restriction of a morphism of affine varieties $\A^n\to \A^m$ to a constructible subset $\M$ of $\A^n$ and hence a total
map from $\M$ to $\A^m$. Furthermore, we call $\phi$ a {\em
rational} map of $\mathcal{M}$ if the domain $U$ of $\phi$ is
contained in $\mathcal{M}$ and $\phi$ is the restriction to
$\mathcal{M}$ of a rational map of the Zariski closure
$\ol{\mathcal{M}}$ of $\mathcal{M}$. In this case $U$ is a Zariski
open and dense subset of $\mathcal{M}$.

Since the elementary, i.e., first--order theory of algebraically closed fields with constants in $\C$ admits quantifier elimination, constructibility means just elementary definability. In particular, $\phi$ is constructible implies that the domain and the image of $\phi$ are constructible subsets of $\A^n$ and $\A^m$, respectively.

\begin{remark}
\label{remark constructible}
A partial map $\phi:\M\dashrightarrow\A^m$ is constructible if and only if it is piecewise rational. If $\phi$ is constructible there exists a Zariski open and dense subset $U$ of $\M$ such that the restriction $\phi|_U$ of $\phi$ to $U$ is a rational map of $\M$ (and of $\MM$).    
\end{remark}

For details we refer to \cite{GHMS09}, Lemma 1.

\subsection{Weakly continuous, strongly continuous, topologically robust and hereditary maps}
\label{sec: geometry-basic-maps}

We are now going to present the notions of a weakly continuous, a strongly continuous, a topologically robust, a geometrically robust and a hereditary map of
a constructible set $\mathcal{M}$. These five notions will
constitute our fundamental tool for the modelling of elimination problems and algorithms.

\begin{definition}\label{def: maps}
Let $\mathcal{M}$ be a constructible subset of $\A^n$ and let
$\phi:\mathcal{M}\to\A^m$ be a (total) constructible map. We
consider the following four conditions:

\begin{enumerate}
\item[(i)]  there exists  a Zariski open and dense
subset $U$ of $\mathcal{M}$ such that the restriction $\phi|_{U}$
of $\phi$ to $U$ is a rational map of $\mathcal{M}$ and the graph
of $\phi$ is contained in the Zariski closure of the graph of
$\phi|_U$ in $\mathcal{M}\times \A^m$;

\item[(ii)] $\phi$ is continuous with respect to the Euclidean, i.e., strong, topologies of $\mathcal{M}$ and $\A^m$;

\item [(iii)] for any sequence $(x_k)_{k\in \N}$ of points of $\M$ which
converges in the Euclidean topology to a point of $\M$, the sequence
$(\phi(x_k))_{k\in \N}$ is bounded;

\item [(iv)] for any constructible subset $\mathcal{N}$ of $\mathcal{M}$ the restriction $\phi |_{\mathcal{N}}:\mathcal{N}\to\A^m$ is an extension of a rational map of $\mathcal{N}$ and the graph of $\phi |_{\mathcal{N}}$ is contained in the Zariski closure of the graph of this rational map in $\mathcal{N}\times \A^m$.
\end{enumerate}

We call the map $\phi$
\begin{itemize}
 \item {\bf weakly continuous} if $\phi$ satisfies condition $(i)$,
 \item {\bf strongly continuous} if $\phi$ satisfies condition $(ii)$,
 \item {\bf  topologically robust} if $\phi$ satisfies conditions
$(i)$ and $(iii)$,
 \item {\bf hereditary} if $\phi$ satisfies condition $(iv)$.
\end{itemize}
\end{definition}
\mydefinitions{\label{def: maps} Weakly and strongly continuous, topologically robust and hereditary maps}

In all these cases we shall refer to $\mathcal{M}$ as the domain of definition of $\phi$ or we shall say that $\phi$ is defined on $\mathcal{M}$. 

\begin{lemma}
\label{remark strongly continuous}
\textit{(\cite{GHMS09}, Lemma 4)} A strongly continuous constructible map is always weakly continuous, topologically robust and hereditary.
\end{lemma}

In Section~\ref{sec: geometry-robust}, we shall establish an algebraic condition, namely geometric robustness, which implies hereditarity.

\subsection{The concept of robustness for constructible maps} \label{sec: geometry-robust}

In this Section we introduce the algebraic--geometric tools we shall use in Section \ref{sec: Singularities and models of computation} and \ref{sec:Model-Applications} for the mathematical modelling of algorithms which solve parameterized computational problems. The main issue of this section will be the notion of a \emph{geometrically robust constructible map} which captures simultaneously the concepts of topological robustness and hereditarity introduced in Section \ref{sec: geometry-basic-maps}

We first characterize in algebraic terms the concept of topological robustness (Theorem \ref{th: equiv robus} below). In Section \ref{sec: Singularities and models of computation} we shall interpret topological robustness as the informal concept of \emph{coalescence} (we call it informal because distinct authors introduce it differently, following the context). For example in Interpolation Theory coalescence refers to certain types of ``convergence'' of problems and algorithms (see \cite{BlCa97}, \cite{DeRo92}, \cite{Olver06} and \cite{GHMS09} for details). In this paper coalescence will be the algorithmic counterpart of topological robustness.

Finally, we introduce the notion of a geometrically robust constructible map and show that such maps are always hereditary. In particular they are topologically robust and give rise to coalescent algorithms.

\subsubsection[An algebraic characterization of the notion of topological robustness]{An algebraic characterization of the notion of topological\\  robustness}
\label{subsec: robustness}

In this subsection, we present an algebraic--geometric result of \cite{GHMS09} which will be relevant in Sections \ref{subsec: geometrical robustness}, \ref{sec: Singularities and models of computation} and \ref{sec:Model-Applications}.

For the moment let us fix a constructible subset $\mathcal M$ of the affine space $\A ^n$ and a (total) constructible map $\phi:\mathcal{M}\rightarrow \A^m$ with components $\phi_1,\ldots,\phi_m$. 

We consider now the Zariski closure $\ol{\mathcal M}$ of the
constructible subset $\mathcal M$ of $\A^n$. Observe that
$\ol{\mathcal M}$ is a closed affine subvariety of $\A ^n$ and
that we may interpret $\C(\ol{\mathcal M})$ as a $\C[\ol{\mathcal{M}}]$--module (or $\C[\ol{\mathcal{M}}]$--algebra). 

Fix now an arbitrary point $x$ of $\ol{\mathcal M}$. By $\mathfrak{M}_x$ we denote the maximal ideal
of coordinate functions of $\C[\ol{\mathcal M}]$ which vanish at
the point $x$. By $\C[\ol{\mathcal M}]_{\mathfrak{M}_x}$ we denote the local
$\C$--algebra of the variety $\MM$ at the point $x$, i.e., the
localization of $\C[\MM]$ at the maximal ideal $\mathfrak{M}_x$. By $\C(\MM )_{\mathfrak{M}_x}$ we denote the localization of the $\C[\MM]$--module $\C(\MM )$ at $\mathfrak{M}_x$.

Following Remark \ref{remark constructible}, we may interpret $\phi_1,\ldots,\phi_m$ as rational functions of the
affine variety $\MM$ and therefore as elements of the total fraction ring $\C(\MM)$ of $\C[\MM]$. Thus $\C[\MM][\phi_1,\ldots,\phi_m]$ and
$\C[\MM]_{\mathfrak{M}_x}[\phi_1,\ldots,\phi_m]$ are
$\C$--subalgebras of $\C(\MM)$ and $\C(\MM)_{\mathfrak{M}_x}$
which contain $\C[\MM]$ and $\C[\MM]_{\mathfrak{M}_x}$,
respectively.

With these notations we are able to formulate the following
statement which establishes the bridge to an algebraic understanding of the notion of topological robustness.

\begin{theorem}(\cite{GHMS09}, Corollary 11)
\label{th: equiv robus}
Let notations and assumptions be as before and suppose that the constructible map $\phi:\M\to\A^m$ is weakly continuous. Then $\phi$ is topologically robust if and only if for any point $x$ of $\M$ the $\C$--algebra $\C[\MM]_{\mathfrak{M}_x}[\phi_1,\dots,\phi_m]$ is a finite $\C[\MM]_{\mathfrak{M}_x}$--module.
\end{theorem}

The only if part of Theorem \ref{th: equiv robus} is an almost immediate consequence of \cite{CaGiHeMaPa03}, Lemma 3, which in its turn is based on a non--elementary and deep result from Algebraic Geometry, namely Zariski's Main Theorem (see, e.g., \cite{Iversen73}, \S IV.2). 

Let $\phi:\mathcal{M}\to\A^m$ be a topologically robust constructible map and let $u$ be an arbitrary point of $\mathcal{M}$. From Theorem \ref{th: equiv robus} one deduces easily that for all sequences $(u_k)_{k\in\N}$ of points $u_k\in\mathcal{M}$ which converge to $u$, the sequences $(\phi(u_k))_{k\in\N}$ have only finitely many accumulation points.

\subsubsection{The notion of geometrical robustness}
\label{subsec: geometrical robustness}

The main mathematical tool of Section \ref{sec: Singularities and models of computation} of this paper is the notion of geometrical robustness we are going to introduce now. We shall use the same notations as in Section \ref{subsec: robustness}.

\begin{definition}
\label{def: geometrically robust map}
Let $\mathcal{M}$ be a constructible subset of a suitable affine space and let $\phi:\mathcal{M}\to\A^m$ be a (total) constructible map with components $\phi_1,\dots,\phi_m$. According to Remark \ref{remark constructible} we may interpret $\phi_1,\dots,\phi_m$ as rational functions of $\ol{\mathcal{M}}$. We call $\phi$ geometrically robust if for any point $x\in\mathcal{M}$ the following two conditions are satisfied:
\begin{itemize}
	\item[(i)] $\C[\ol{\mathcal{M}}]_{\mathfrak{M}_x}[\phi_1,\dots,\phi_m]$ is a finite $\C[\ol{\mathcal{M}}]_{\mathfrak{M}_x}$--module.	 
	\item[(ii)] $\C[\ol{\mathcal{M}}]_{\mathfrak{M}_x}[\phi_1,\dots,\phi_m]$ is a local $\C[\ol{\mathcal{M}}]_{\mathfrak{M}_x}$--algebra whose maximal ideal is generated by $\mathfrak{M}_x$ and $\phi_1-\phi_1(x),\dots,\phi_m-\phi_m(x)$.
\end{itemize}
\end{definition} 

Observe that the notion of a geometrically robust map makes also sense when $\C$ is replaced by an arbitrary algebraically closed field (of any characteristic). In view of Theorem \ref{th: equiv robus} the same is true for the notion of a topologically robust map. 

In this paper we shall restrict our attention to the algebraically closed field $\C$. In this particular case we have the following characterization of geometrically robust constructible maps.

\begin{theorem}
\label{theorem 2}
Let assumptions and notations be as before. Then the constructible map $\phi:\mathcal{M}\to\A^m$ is geometrically robust if and only if $\phi$ is strongly continuous.
\end{theorem}

\begin{proof}
Suppose that the constructible map $\phi$ is geometrically robust. We are first going to show that $\phi$ is weakly continuous.

By Remark \ref{remark constructible} there exists a Zariski open and dense subset $U$ of $\M$ such that the restriction map $\phi|_U$ is rational. Let $Y_1,\dots,Y_m$ be new indeterminates, $Y:=(Y_1,\dots,Y_m)$ and suppose that the affine ambient space of $\M$ has dimension $n$. Observe that any $(n+m)$--variate polynomial over $\C$ which vanishes on the graph of the rational map $\phi|_U$ gives rise to a polynomial $A\in\C[\MM][Y]$ with $A[\phi_1,\dots,\phi_m]=0$.

Let $x$ be an arbitrary point of $\M$ and consider $A$ as an element of $\C[\MM]_{\mathfrak{M}_x}[Y]$. Denote by $A(x,\phi(x))$ the value of $A$ at $(x,\phi(x))$. Then condition $(ii)$ of Definition \ref{def: geometrically robust map} implies that $A[\phi_1,\dots,\phi_m]-A(x,\phi(x))$ belongs to the maximal ideal of $\C[\MM]_{\mathfrak{M}_x}[\phi_1,\dots,\phi_m]$. From $A[\phi_1,\dots,\phi_m]=0$ we deduce now $A(x,\phi(x))=0$.

Since the choice of $x\in\M$ was arbitrary, we conclude that $A$ vanishes on the graph of $\phi$. This implies that the graph of $\phi$ is contained in the Zariski closure of the graph of $\phi|_U$. Hence $\phi$ is weakly continuous.

\enter
Let be given an arbitrary point $x\in\mathcal{M}$ and a sequence $(x_k)_{k\in\N}$, $x_k\in\mathcal{M}$, which converges to $x$ in the strong topology of $\mathcal{M}$. We are now going to show that the sequence $(\phi(x_k))_{k\in\N}$ converges to $\phi(x)$.

Since $\phi$ is weakly continuous, we deduce from condition $(i)$ of Definition \ref{def: geometrically robust map} and Theorem \ref{th: equiv robus} that the sequence $(\phi(x_n))_{k\in\N}$ contains at least one accumulation point, say $a=(a_1,\dots,a_m)$, which belongs to $\A^m$. Let $\mathfrak{a}$ be the ideal of all polynomials $A\in\C[\mathcal{M}]_{\mathfrak{M}_x}[Y]$ that vanish at the point $(x,a)\in\A^n\times\A^m$. Without loss of generality we may assume that the sequence $(\phi(x_k))_{k\in\N}$ converges to $a$. Let $\wt{\mathfrak{a}}:=\{ A(\phi); A\in\mathfrak{a} \}$ be the image of the ideal $\mathfrak{a}$ under the surjective $\C[\ol{\mathcal{M}}]_{\mathfrak{M}_x}$--algebra homomorphism $\C[\ol{\mathcal{M}}]_{\mathfrak{M}_x}[Y] \to \C[\ol{\mathcal{M}}]_{\mathfrak{M}_x}[\phi_1,\dots,\phi_m]$ which maps $Y_1,\dots,Y_m$ onto $\phi_1,\dots,\phi_m$. Observe that $\wt{\mathfrak{a}}$ is an ideal of $\C[\ol{\mathcal{M}}]_{\mathfrak{M}_x}[\phi_1,\dots,\phi_m]$.

We are now going to show the following statement.

\begin{claim}
\label{intermediate claim}
The ideal $\wt{\mathfrak{a}}$ is proper.
\end{claim}
\paragraph*{Proof of the claim.} Suppose that the ideal $\wt{\mathfrak{a}}$ is not proper. Then there exists a polynomial $A=\sum_{j_1,\dots,j_m} a_{j_1 \dots j_m} Y_1^{j_1}\dots Y_m^{j_m}$ of $\mathfrak{a}$, with $a_{j_1 \dots j_m}\in\C[\ol{\mathcal{M}}]_{\mathfrak{M}_x}$, which satisfies the condition $\sum_{j_1,\dots,j_m} a_{j_1 \dots j_m} \phi_1^{j_1}\dots \phi_m^{j_m}=A(\phi)=1$. Since for any $m$--tuple of indices $j_1,\dots,j_m$ the rational function $a_{j_1 \dots j_m}$ of $\ol{\mathcal{M}}$ is defined at $x$ and the sequence $(x_k)_{k\in\N}$ converges to $x$, we may assume without loss of generality that $a_{j_1 \dots j_m}$ is defined at $x_k$ for any $k\in\N$ and that $(a_{j_1 \dots j_m}(x_k))_{k\in\N}$ converges to $a_{j_1 \dots j_m}(x)$. We may therefore write $A^{(x')}:= \sum a_{j_1 \dots j_m}(x') Y_1^{j_1}\dots Y_m^{j_m} \in\C[Y]$ for $x':=x$ or $x':=x_k$, $k\in\N$. From $A\in\mathfrak{a}$ we deduce $A^{(x)}(a)=0$. By assumption $(\phi(x_k))_{k\in\N}$ converges to $a$. Hence the sequence of complex numbers $(A^{(x_k)}(\phi(x_k)))_{k\in\N}$ converges to $A^{(x)}(a)=0$. On the other hand $A(\phi)=1$ and the weak continuity of $\phi$ imply $A^{(x_k)}(\phi(x_k))=1$ for any $k\in \N$. This contradiction proves our claim.

\enter
From condition $(ii)$ of Definition \ref{def: geometrically robust map} we deduce that the $\C[\ol{\mathcal{M}}]_{\mathfrak{M}_x}$--algebra\\
$\C[\ol{\mathcal{M}}]_{\mathfrak{M}_x}[\phi_1,\dots,\phi_m]$ contains a single maximal ideal, say $\mathfrak{M}$, and that $\mathfrak{M}$ is generated by $\mathfrak{M}_x$ and $\phi_1-\phi_1(x),\dots,\phi_m-\phi_m(x)$.

Since by Claim \ref{intermediate claim} the ideal $\wt{\mathfrak{a}}$ is proper, $\wt{\mathfrak{a}}$ must be contained in $\mathfrak{M}$. Observe that the polynomials    
$Y_1-a_1,\dots,Y_m-a_m$ belong to $\mathfrak{a}$. Hence $\phi_1-a_1,\dots,\phi_m-a_m$ belong to $\wt{\mathfrak{a}}$ and therefore also to $\mathfrak{M}$.  Since $\mathfrak{M}$ is proper, this is only possible if $a_1=\phi_1(x),\dots,a_m=\phi_m(x)$ holds.

Thus the sequence $(\phi(x_k))_{k\in\N}$ converges to $\phi(x)$. 

\ws\\
Suppose now that the constructible map $\phi$ is strongly continuous. From Lemma \ref{remark strongly continuous} we deduce that $\phi$ is topologically robust. Theorem \ref{th: equiv robus} implies now that $\phi$ satisfies condition $(i)$ of Definition \ref{def: geometrically robust map} at any point of $\mathcal{M}$.   

Let $x$ be an arbitrary point of $\mathcal{M}$. We have to show that $\phi$ satisfies at $x$ condition $(ii)$ of Definition \ref{def: geometrically robust map}. 

Since the graph of $\phi$ is constructible, its strong and Zariski closures in $\mathcal{M}\times\A^m$ coincide. Moreover, since $\phi$ is by assumption strongly continuous, its graph is closed with respect to the strong topology of $\mathcal{M}\times\A^m$ and therefore also with respect to the Zariski topology. Let $\mathfrak{a}$ be an arbitrary maximal ideal of the $\C[\ol{\mathcal{M}}]_{\mathfrak{M}_x}$--algebra $\C[\ol{\mathcal{M}}]_{\mathfrak{M}_x}[\phi_1,\dots,\phi_m]$. Then there exists a point $a=(a_1,\dots,a_m)$ of $\A^m$ such that $\mathfrak{a}$ is generated by $\mathfrak{M}_x$ and $\phi_1 - a_1, \dots , \phi_m - a_m$. Thus $(x,a)\in\mathcal{M}\times\A^m$ belongs to the Zariski closure of the graph of $\phi$ in $\mathcal{M}\times\A^m$ and therefore to the graph of $\phi$ itself. This implies $a=\phi(x)$. With other words, $\mathfrak{a}$ is generated by $\mathfrak{M}_x$ and $\phi_1 - \phi_1(x), \dots , \phi_m - \phi_m(x)$. There is exactly one ideal of $\C[\ol{\mathcal{M}}]_{\mathfrak{M}_x}[\varphi_1,\dots,\varphi_m]$ which satisfies this condition. Therefore the $\C[\ol{\mathcal{M}}]_{\mathfrak{M}_x}$--algebra $\C[\ol{\mathcal{M}}]_{\mathfrak{M}_x}[\varphi_1,\dots,\varphi_m]$ is local and condition $(ii)$ and Definition \ref{def: geometrically robust map} is satisfied at the point $x\in\mathcal{M}$.
\end{proof}

Theorem \ref{theorem 2} implies immediately the following result which will be fundamental in the sequel.

\begin{corollary}
\label{proposition 1}
Geometrically robust constructible maps are weakly continuous, hereditary and even topologically robust. If we restrict a geometrically robust constructible map to a constructible subset of its domain of definition we obtain again a geometrically robust map. Moreover the composition and the cartesian product of two geometrically robust constructible maps are geometrically robust. The geometrically robust constructible functions form a commutative $\C$--algebra which contains the polynomial functions.
\end{corollary}

Notice that Corollary \ref{proposition 1} remains mutatis mutandis true if the notions of geometrical and topological robustness are applied to constructible maps defined over an \emph{arbitrary} algebraically closed field $k$. 

Adapting the corresponding proofs to this more general situation, one sees that weak continuity and hereditarity of geometrically robust constructible maps with \emph{irreducible} domains of definition follows from \cite{GHMS09}, Proposition 16, Theorem 17 and Corollary 18. These results imply also that restrictions of such maps to irreducible constructible subsets of their domains of definition are again geometrically robust. From this one deduces immediately the same statements for the case of arbitrary domains of definition. Topological robustness is a direct consequence of Definition \ref{def: geometrically robust map}. Closedness under composition follows from the transitivity law for integral dependence. One infers from Definition \ref{def: geometrically robust map} closedness under cartesian products and that the geometrically robust constructible functions form a commutative $k$--algebra which contains the polynomial functions.

Theorem \ref{theorem 2} is new. It gives a topological motivation for the rather abstract, algebraic notion of geometrical robustness. The reader not acquainted with commutative algebra may just identify the concept of geometrical robustness with that of strong continuity of constructible maps.

The origin of the concept of a geometrically robust map can be found, implicitly, in \cite{GH01}. It was introduced explicitly for constructible maps with irreducible domains of definition in \cite{GHMS09}, where it is used to analyze the complexity character of multivariate Hermite--Lagrange interpolation.


\section[A software architecture based model for computations with parameterized arithmetic circuits]
{\large{A software architecture based model for computations with  parameterized arithmetic circuits}}
\label{sec: Singularities and models of computation}
 
\subsection{Parameterized arithmetic circuits and their semantics}
\label{sec:Model-circuits}

The routines of our computation model, which will be introduced in Section \ref{sec:Model-discussion}, operate with circuits representing parameter dependent rational functions. They will behave well under restrictions. In this spirit, the objects of our abstract data types will be parameter dependent multivariate rational functions over $\C$, the concrete objects of our classes will be parameterized arithmetic circuits and our abstraction function will associate circuits with rational functions. In what follows, $\C$ may always be replaced, mutatis mutandis, by an arbitrary algebraically closed field (of any characteristic).

Let us fix natural numbers $n$ and $r$, indeterminates $X_1,\dots ,X_n$ and a non--empty constructible subset $\mathcal{M}$ of $\A^r$. By $\pi_1,\dots ,\pi_r$ we denote the restrictions to $\mathcal{M}$ of the canonical projections $\A^r\to\A^1$.

A \emph{(by $\mathcal{M}$) parameterized arithmetic circuit $\beta$} (with \emph{basic parameters $\pi_1,\dots ,\pi_r$} and \emph{inputs $X_1,\dots ,X_n$}) is a labelled directed acyclic graph (labelled DAG) satisfying the following conditions:\\
each node of indegree zero is labelled by a scalar from $\C$, a basic parameter $\pi_1,\dots ,\pi_r$ or a input variable $X_1,\dots ,X_n$. Following the case, we shall refer to the \emph{scalar, basic parameter} and (standard) \emph{input} nodes of $\beta$. All other nodes of $\beta$ have indegree two and are called \emph{internal}. They are labelled by arithmetic operations (addition, subtraction, multiplication, division). A \emph{parameter} node of $\beta$ depends only on scalar and basic parameter nodes, but not on any input node of $\beta$ (here ``dependence'' refers to the existence of a connecting path). An addition or multiplication node whose two ingoing edges depend on an input is called \emph{essential}. The same terminology is applied to division nodes whose second argument depends on an input. Moreover, at least one circuit node becomes labelled as output. Without loss of generality we may suppose that all nodes of outdegree zero are outputs of $\beta$.	

\enter
We consider $\beta$ as a syntactical object which we wish to equip with a certain semantics. In principle there exists a canonical evaluation procedure of $\beta$ assigning to each node a rational function of $\mathcal{M}\times \A^n$ which, in case of a parameter node, may also be interpreted as a rational function of $\mathcal{M}$. In either situation we call such a rational function an \emph{intermediate result} of $\beta$. 

The evaluation procedure may fail if we divide at some node an intermediate result by another one which vanishes on a Zariski dense subset of a whole irreducible component of $\mathcal{M}\times \A^n$. If this occurs, we call the labelled DAG $\beta$ \emph{inconsistent}, otherwise \emph{consistent}. From~\cite{CaGiHeMaPa03}, Corollary 2 (compare also~\cite{HeSc82}, Theorem 4.4 and~\cite{GH01}, Lemma 3) one deduces easily that testing whether an intermediate result of $\beta$ vanishes on a Zariski dense subset of a whole irreducible component of $\mathcal{M}\times \A^n$ can efficiently be reduced to the same task for circuit represented rational functions of $\mathcal{M}$ (the procedure is of non--uniform deterministic or alternatively of uniform probabilistic nature).

Mutatis mutandis the same is true for identity checking between intermediate results of $\beta$. If $\mathcal{M}$ is irreducible, both tasks boil down to an identity--to--zero test on $\mathcal{M}$. In case that $\mathcal{M}$ is not Zariski dense in $\A^r$, this issue presents a major open problem in modern Theoretical Computer Science (see \cite{SaxenaSurvey09} and \cite{shpilka} for details).

If nothing else is said, we shall from now on assume that $\beta$ is a consistent parameterized arithmetic circuit. The intermediate results associated with output nodes will be called \emph{final results} of $\beta$. 

We call an intermediate result associated with a parameter node a \emph{parameter} of $\beta$ and interpret it generally as a rational function of $\mathcal{M}$. A parameter associated with a node which has an outgoing edge into a node which depends on some input of $\beta$ is called \emph{essential}. In the sequel we shall refer to the constructible set $\mathcal{M}$ as the \emph{parameter domain} of $\beta$.

We consider $\beta$ as a syntactic object which represents the final results of $\beta$, i.e., the rational functions of $\mathcal{M}\times\A^n$ assigned to its output nodes. In this way becomes introduced an abstraction function which associates $\beta$ with these rational functions. This abstraction function assigns therefore to $\beta$ a rational map $\mathcal{M}\times\A^n \dashrightarrow \A^q$, where $q$ is the number of output nodes of $\beta$. On its turn, this rational map may also be understood as a (by $\mathcal{M}$) parameterized family of rational maps $\A^n \dashrightarrow \A^q$.

Now we suppose that the parameterized arithmetic circuit $\beta$ has been equipped with an additional structure, linked to the semantics of $\beta$. We assume that for each node $\rho$ of $\beta$ there is given a \emph{total} constructible map $\mathcal{M}\times \A^n \to \A^1$ which extends the intermediate result associated with $\rho$. In this way, if $\beta$ has $K$ nodes, we obtain a total constructible map $\Omega:\mathcal{M}\times \A^n \to \A^K$ which extends the rational map $\mathcal{M}\times \A^n \dashrightarrow \A^K$ given by the labels at the indegree zero nodes and the intermediate results of $\beta$.
    
\begin{definition}[Robust circuit] \label{def: robust circuit}
Let notations and assumptions be as before. The pair $(\beta,\Omega)$ is called a robust parameterized arithmetic circuit if the constructible map $\Omega$ is geometrically robust. 
\end{definition}
\mydefinitions{\label{def: robust circuit} Robust circuit}

We shall make the following two observations to this definition.

We state our first observation. Suppose that $(\beta,\Omega)$ is robust. This means that the constructible map $\Omega:\mathcal{M}\times\A^n \to \A^K$ is geometrically and hence also topologically robust and hereditary. Moreover, 
 the above rational map $\mathcal{M}\times \A^n \dashrightarrow \A^K$ can be extended to at most one geometrically robust constructible map $\Omega:\mathcal{M}\times \A^n \to \A^K$. Therefore we shall apply from now on the term ``robust'' also to the circuit $\beta$.

Let us now state our second observation. We may consider the parameterized circuit $\beta$ as a program which solves the problem to evaluate, for any sufficiently generic parameter instance $u\in\mathcal{M}$, the rational map $\A^n \dashrightarrow \A^q$ which we obtain by specializing to the point $u$ the first argument of the rational map $\mathcal{M}\times \A^n \dashrightarrow \A^q$ defined by the final results of $\beta$. In this sense, the ``computational problem'' solved by $\beta$ is given by the final results of $\beta$.
	
Being robust becomes now an architectural requirement for the circuit $\beta$ and for its output. Robust parameterized arithmetic circuits may be restricted as follows:

Let $\mathcal{N}$ be a constructible subset of $\mathcal{M}$ and suppose that $(\beta,\Omega)$ is robust. Then Corollary \ref{proposition 1} implies that the restriction $\Omega\vert_{\mathcal{N}\times\A^n}$ of the constructible map $\Omega$ to $\mathcal{N}\times\A^n$ is still a geometrically robust constructible map.

This implies that $(\beta,\Omega)$ induces a by $\mathcal{N}$ parameterized arithmetical circuit $\beta_{\mathcal{N}}$. Observe that $\beta_{\mathcal{N}}$ may become inconsistent. If $\beta_{\mathcal{N}}$ is consistent then $(\beta_{\mathcal{N}}, \Omega\vert_{\mathcal{N}\times\A^n})$ is robust. The nodes where the evaluation of $\beta_{\mathcal{N}}$ fails correspond to divisions of zero by zero which may be replaced by so called approximative algorithms having unique limits (see Section \ref{sec:Model-discussion-simplified}). These limits are given by the map $\Omega\vert_{\mathcal{N}\times\A^n}$. We call $(\beta_{\mathcal{N}}, \Omega\vert_{\mathcal{N}\times\A^n})$, or simply $\beta_{\mathcal{N}}$, the \emph{restriction} of $(\beta,\Omega)$ or $\beta$ to $\mathcal{N}$. 

We say that the parameterized arithmetic circuit $\beta$ is \emph{totally division--free} if any division node of $\beta$ corresponds to a division by a non--zero complex scalar.

We call $\beta$ \emph{essentially division--free} if only parameter nodes are labelled by divisions. Thus the property of $\beta$ being totally division--free implies that $\beta$ is essentially division--free, but not vice versa. Moreover, if $\beta$ is totally division-free, the rational map given by the intermediate results of $\beta$ is polynomial and therefore a geometrically robust constructible map. Thus, any by $\mathcal{M}$ parameterized, totally division--free circuit is in a natural way robust.

In the sequel, we shall need the following elementary fact. 

\begin{lemma}
\label{lemma intermediate results}
Let notations and assumptions be as before and suppose that the parameterized arithmetic circuit $\beta$ is robust. Then all intermediate results of $\beta$ are polynomials in $X_1,\dots,X_n$ over the $\C$--algebra of geometrically robust constructible functions defined on $\mathcal{M}$.
\end{lemma}
\begin{proof}
Without loss of generality we may assume that $\M$ is irreducible. Let $\rho$ be a node of $\beta$ which computes the intermediate result $G_{\rho}:\M\times\A^n\to\A^1$. Definition 6 $(i)$ and the irreducibility of $\mathcal{M}$ imply that $G_{\rho}$ is a polynomial of $\C(\MM)[X_1,\dots,X_n]$. Observe that any $x\in\A^n$ induces a geometrically robust constructible map $\M\to\A^1$ whose value at the point $u\in\M$ is $G_{\rho}(u,x)$. Using interpolation at suitable points of $\A^n$, we see that the coefficients of the polynomial $G_{\rho}$ are geometrically robust constructible functions with domain of definition $\M$.
\end{proof}

\enter
The statement of this lemma should not lead to confusions with the notion of an essentially division--free parameterized circuit. We say just that the intermediate results of $\beta$ are polynomials in $X_1,\dots , X_n$ and do not restrict the type of arithmetic operations contained in $\beta$ (as we did defining the notion of an essentially division--free parameterized circuit). 

Whether a division of a polynomial by one of its factors may always be substituted efficiently by additions and multiplications is an important issue in Theoretical Computer Science (compare \cite{Strassen73}).

To our parameterized arithmetic circuit $\beta$ we may associate different complexity measures and models. In this paper we shall mainly be concerned with \emph{sequential computing time}, measured by the \emph{size} of $\beta$. Here we refer with ``size'' to the number of internal nodes of $\beta$ which count for the given complexity measure. Our basic complexity measure is the \emph{non--scalar} one (also called \emph{Ostrowski measure}) over the ground field $\C$. This means that we count, at unit costs, only essential multiplications and divisions (involving basic parameters or input variables in both arguments in the case of a multiplication and in the second argument in the case of a division), whereas $\C$--linear operations are free (see \cite{Burgisser97} for details). 

\subsubsection{Operations with robust parameterized arithmetic circuits}

\paragraph{The operation join}\ws\ws

Let $\gamma_1$ and $\gamma_2$ be two robust parameterized arithmetic circuits with parameter domain $\M$ and suppose that there is given a one--to--one correspondence $\lambda$ which identifies the output nodes of $\gamma_1$ with the input nodes of $\gamma_2$ (thus they must have the same number). Using this identification we may now join the circuit $\gamma_1$ with the circuit $\gamma_2$ in order to obtain a new parameterized arithmetic circuit $\gamma_2*_{\lambda}\gamma_1$ with parameter domain $\M$. The circuit $\gamma_2*_{\lambda}\gamma_1$ has the same input nodes as $\gamma_1$ and the same output nodes as $\gamma_2$ and one deduces easily from Lemma \ref{lemma intermediate results} and Corollary \ref{proposition 1} that the circuit $\gamma_2*_{\lambda}\gamma_1$ is robust and represents a composition of the rational maps defined by $\gamma_1$ and $\gamma_2$, if $\gamma_2*_{\lambda}\gamma_1$ is consistent. The (consistent) circuit $\gamma_2*_{\lambda}\gamma_1$ is called the (consistent) \emph{join} of $\gamma_1$ with $\gamma_2$.

Observe that the final results of a given robust parameterized arithmetic circuit may constitute a vector of parameters. The join of such a circuit with another robust parameterized arithmetic circuit, say $\beta$, is again a robust parameterized arithmetic circuit which is called an \emph{evaluation} of $\beta$. Hence, mutatis mutandis, the notion of join of two routines includes also the case of circuit evaluation.

\paragraph{The operations reduction and broadcasting}\ws\ws

We describe now how, based on its semantics, a given parameterized arithmetic circuit $\beta$ with parameter domain $\M$ may be rewritten as a new circuit over $\M$ which computes the same final results as $\beta$.

The resulting two rewriting procedures, called \emph{reduction} and \emph{broadcasting}, will neither be unique, nor generally confluent. To help understanding, the reader may suppose that there is given an (efficient) algorithm which allows identity checking between intermediate results of $\beta$. However, we shall not make explicit reference to this assumption. We are now going to explain the first rewriting procedure.

Suppose that the parameterized arithmetic circuit $\beta$ computes at two different nodes, say $\rho$ and $\rho'$, the same intermediate result. Assume first that $\rho$ neither depends on $\rho'$, nor $\rho'$ on $\rho$. Then we may erase $\rho'$ and its two ingoing edges (if $\rho'$ is an internal node) and draw an outgoing edge from $\rho$ to any other node of $\beta$ which is reached by an outgoing edge of $\rho'$. If $\rho'$ is an output node, we label $\rho$ also as output node. Observe that in this manner a possible indexing of the output nodes of $\beta$ may become changed but not the final results of $\beta$ themselves.

Suppose now that $\rho'$ depends on $\rho$. Since the DAG $\beta$ is acyclic, $\rho$ does not depend on $\rho'$. We may now proceed in the same way as before, erasing the node $\rho'$.

Let $\beta'$ be the parameterized arithmetic circuit obtained, as described before, by erasing the node $\rho'$. Then we call $\beta'$ a \emph{reduction} of $\beta$ and call the way we obtained $\beta'$ from $\beta$ a \emph{reduction step}. A \emph{reduction procedure} is a sequence of successive reduction steps.

One sees now easily that a reduction procedure applied to $\beta$ produces a new parameterized arithmetic circuit $\beta^*$ (also called a \emph{reduction} of $\beta$) with the same basic parameter and input nodes, which computes the same final results as $\beta$ (although their possible indexing may be changed). Moreover, if $\beta$ is a robust parameterized circuit, then $\beta^*$ is robust too. Observe also that in the case of robust parameterized circuits our reduction commutes with restriction. 

Now we introduce the second rewriting procedure.

Let assumptions and notations be as before and let be given a set $P$ of nodes of $\beta$ and a robust parameterized arithmetic circuit $\gamma$ with parameter domain $\mathcal{M}$ and $\# P$ input nodes, namely for each $\rho\in P$ one which becomes labelled by a new input variable $Y_{\rho}$. We obtain a new parameterized arithmetic circuit, denoted by $\gamma*_P\beta$, when we join $\gamma$ with $\beta$, replacing for each $\rho\in P$ the input node of $\gamma$, which is labelled by the variable $Y_{\rho}$, by the node $\rho$ of $\beta$. The output nodes of $\beta$ constitute also the output nodes of $\gamma*_P\beta$. Thus $\beta$ and $\gamma*_P\beta$ compute the same final results. Observe that $\gamma*_P\beta$ is robust if it is consistent. We call the circuit $\gamma*_P\beta$ and all its reductions \emph{broadcastings} of $\beta$. Thus broadcasting a robust parameterized arithmetic circuit means rewriting it using only valid polynomial identities.

If we consider arithmetic circuits as computer programs, then reduction and broadcasting represent a kind of program transformations.

\subsubsection{A specification language for circuits}
\label{sec: A specification language for circuits}
Computer programs (or ``programmable algorithms'') written in high level languages are not the same thing as just ``algorithms'' in Complexity Theory. Whereas in the uniform view algorithms become implemented by suitable machine models and in the non--uniform view by devices like circuits; specifications and correctness proofs are not treated by the general theory, but only, if necessary, outside of it in a case--by--case ad--hoc manner. The meaning of ``algorithm'' in Complexity Theory is therefore of syntactic nature. 

On the other hand, computer programs, as well as their subroutines (modules) include specifications and correctness proofs, typically written in languages organized by a hierarchy of different abstraction levels. In this sense \emph{programmable algorithms} become equipped with semantics. This is probably the main difference between Complexity Theory and Software Engineering.

In this paper, we are only interested in algorithms which in some sense are programmable. The routines of our computation model will operate on parameterized arithmetic circuits (see Section \ref{sec:Model-discussion}). Therefore we are now going to fix a (many--sorted) first--order specification language $\mathcal{L}$ for these circuits.

The language $\mathcal{L}$ will include the following non--logical symbols:
\begin{enumerate}
	\item[-] $0,1,+,-,\times$, and a constant for each complex number,
	\item[-] variables
							$$n_1,\dots,n_s\dots $$
							$$\alpha^{(1)},\dots,\alpha^{(t)}\dots $$
							$$\beta_1,\dots,\beta_k\dots $$ 
							$$\rho_1,\dots,\rho_l\dots $$ 
							$$\mathcal{M}_1,\dots,\mathcal{M}_k\dots $$
							$$U^{(1)},\dots,U^{(m)}\dots $$ 
							$$X^{(1)},\dots,X^{(h)}\dots $$ 
							$$Y^{(1)},\dots,Y^{(q)}\dots $$ 
							
to denote non--negative integers and vectors of them, robust parameterized arithmetic circuits, their nodes, their parameter domains, their parameter instances, their input variable vectors and instances of input variable vectors in suitable affine spaces, 
	\item[-] suitable binary predicate symbols to express relations like ``$\rho$ is a node of the circuit $\beta$'', ``multiplication is the label of the node $\rho$ of the circuit $\beta$'', ``$\mathcal{M}$ is the parameter domain of the circuit $\beta$'', ``$U$ is a parameter instance of the circuit $\beta$'', ``$r$ is a non--negative integer and the vector length of $U$ is $r$'', ``$X$ is the input variable vector of the circuit $\beta$'' and ``$n$ is a non--negative integer and the vector length of $X$ is $n$'',
	\item[-] a ternary predicate symbol to express ``$\rho_1$ and $\rho_2$ are nodes of the circuit $\beta$ and there is an edge of $\beta$ from $\rho_1$ to $\rho_2$'',
	\item[-] binary function symbols to express ``$U$ is a parameter instance, $k$ is a natural number and $U_k$ is the $k$--th entry of $U$'' and ``$X$ is an input variable vector, $n$ is a natural number and $X_n$ is the $n$--th entry of $X$'' and ``$Y$ is a variable vector instance, $n$ is a natural number and $Y_n$ is the $n$--th entry of $Y$'',
	\item[-] a unary function and a binary predicate symbol to express ``the set of output nodes of the circuit $\beta$'' and ``$\rho$ is an output node of the circuit $\beta$''
	\item[-] a quaternary function symbol $G_{\rho}(\beta;U;X)$ to express ``$\rho$ is a node of the circuit $\beta$, $U$ is a parameter instance and $X$ is the input variable vector of $\beta$ and $G_{\rho}(\beta;U;X)$ is the intermediate result of $\beta$ at the node $\rho$ and the parameter instance $U$'',
	
	\item[-] a predicate symbol for equality for any of the sorts just introduced. 
\end{enumerate}

For the treatment of non--negative integers we add the Presburger arithmetic to our first--order specification language $\mathcal{L}$. 

At our convenience we may add new function and predicate symbols and variable sorts to $\mathcal{L}$. Typical examples are for $\beta$ a circuit, $U$ a parameter instance, $X$ the input variable vector and $\rho,\rho_1,\dots,\rho_m$ nodes of $\beta$:\\
``degree of $G_{\rho}(\beta;U;X)$'' and ``the vector lengths of $X$ and $Y$ are equal (say $n$) and $Y$ is a point of the closed subvariety of $\A^n$ defined by the polynomials $G_{\rho_1}($$\beta;U;X),$ $\dots,$ $G_{\rho_m}(\beta;U;X)$''.

In the same spirit, we may increase the expressive power of $\mathcal{L}$ in order to be able to express for a robust parameterized circuit $\beta$ with irreducible parameter domain, $U$ a parameter instance, $X$ the input variable vector, $\rho$ a node of $\beta$ and $\alpha$ a vector of non--negative integers of the same length as $X$ (say $n$), ``the coefficient of the monomial $X^{\alpha}$ occurring in the polynomial $G_{\rho}(\beta;U;X)$'' (recall Lemma \ref{lemma intermediate results}). Here we denote for $X:=(X_1,\dots,X_n)$ and $\alpha:=(\alpha_1,\dots,\alpha_n)$ by $X^{\alpha}$ the monomial $X^{\alpha}:=X_1^{\alpha_1},\dots,X_n^{\alpha_n}$.

The semantics of the specification language $\mathcal{L}$ is determined by the universe of all robust parameterized arithmetic circuits, where we interpret all variables, function symbols and predicates as explained before. We call this universe the \emph{standard model} of $\mathcal{L}$. The set of all closed formulas of $\mathcal{L}$ which are true in this model form the \emph{elementary theory} of $\mathcal{L}$. 

\subsection{Generic computations}
\label{sec:Model-generic computations}

In the sequel, we shall use ordinary arithmetic circuits over $\C$ as \emph{generic computations} \cite{Burgisser97} (also called \emph{computation schemes} in \cite{Hei89}). The indegree zero nodes of these arithmetic circuits are labelled by scalars and parameter and input variables. 

The aim is to represent different parameterized arithmetic circuits of similar size and appearance by different specializations (i.e., instantiations) of the parameter variables in one and the same generic computation. For a suitable specialization of the parameter variables, the original parameterized arithmetic circuit may then be recovered by an appropriate reduction process applied to the specialized generic computation.

This alternative view of parameterized arithmetic circuits will be fundamental for the design of routines of the branching--free computation model we are going to describe in Section \ref{sec:Model-discussion-simplified}. The routines of our computation model will operate on robust parameterized arithmetic circuits and their basic ingredients will be subroutines which calculate parameter instances of suitable, by the model previously fixed, generic computations. These generic computations will be organized in finitely many families which will only depend on a constant number of discrete parameters. These discrete families constitute the basic building block of our model for branching--free computation.

\enter
We shall now exemplify these abstract considerations in the concrete situation of the given parameterized arithmetic circuit $\beta$. Mutatis mutandis we shall follow the exposition of \cite{Krick96}, Section 2. Let $l,L_0,\dots,L_{l+1}$ with $L_0\geq r+n+1$ and $L_{l+1}\geq q$ be given natural numbers. Without loss of generality we may suppose that the non--scalar depth of $\beta$ is positive and at most $l$, and that $\beta$ has an oblivious levelled structure of $l+2$ levels of width at most $L_0,\dots,L_{l+1}$. Let $U_1,\dots,U_r$ be new indeterminates (they will play the role of a set of ``special'' parameter variables which will only be instantiated by $\pi_1,\dots,\pi_r$).

We shall need the following indexed families of ``scalar'' parameter variables (which will only be instantiated by complex numbers):

\begin{enumerate}
	\item[-] for $n+r<j\leq L_0$ the indeterminate $V_j$;
	\item[-] for $1\leq i\leq l$, $1\leq j\leq L_i$, $0\leq h\leq i$, $1\leq k\leq L_h$, the indeterminates $A^{(h,k)}_{i,j}$, $B^{(h,k)}_{i,j}$ and $S_{i,j}$, $T_{i,j}$;
	\item[-] for $1\leq j\leq L_{l+1}$, $1\leq k \leq L_l$ the indeterminate $C^k_j$.
\end{enumerate}

We consider now the following function $Q$ which assigns to every pair $(i,j)$, $1\leq i\leq l$, $1\leq j\leq L_i$ and $(l+1,j)$, $1\leq j\leq L_{l+1}$ the rational expressions defined below:
$$Q_{0,1}:=U_1,\dots,Q_{0,r}:=U_r,$$
$$Q_{0,r+1}:=X_1,\dots,Q_{0,r+n}:=X_n,$$
$$Q_{0,r+n+1}:=V_{r+n+1},\dots, Q_{0,L_0}:=V_{L_0}.$$

For $1\leq i\leq l$ and $1\leq j\leq L_i$ the value $Q_{i,j}$ of the function $Q$ is recursively defined by
$$Q_{i,j}:=S_{i,j}(\sum_{\stackrel{0\leq h<i}{1\leq k\leq L_h}} A^{(h,k)}_{i,j} Q_{h,k}\ws. \sum_{\stackrel{0\leq k'<i}{1\leq k'\leq L_{h'}}} B^{(h',k')}_{i,j} Q_{h',k'} ) \ws+$$
$$\textcolor{white}{Q_{i,j}=}T_{i,j}(\sum_{\stackrel{0\leq h<i}{1\leq k\leq L_h}} A^{(h,k)}_{i,j} Q_{h,k}\ws/ \sum_{\stackrel{0\leq h'<i}{1\leq k'\leq L_{h'}}} B^{(h',k')}_{i,j} Q_{h',k'} ).$$

Finally, for $(l+1,j)$, $1\leq j\leq L_{l+1}$ we define $Q_{(l+1,j)}:=\sum_{1\leq k\leq L_l}C^k_j Q_{l,k}$.
\enter

We interpret the function $Q$ as a (consistent) ordinary arithmetic circuit, say $\Gamma$, over $\Z$ (and hence over $\C$) whose indegree zero nodes are labelled by the ``standard'' input variables $X_1,\dots,X_n$, the special parameter variables $U_1,\dots,U_r$ and the scalar parameter variables just introduced.

We consider first the result of instantiating the scalar parameter variables contained in $\Gamma$ by complex numbers. We call such an instantiation a \emph{specialization} of $\Gamma$. It is determined by a point in a suitable affine space. Not all possible specializations are \emph{consistent}, giving rise to an assignment of a rational function of $\C(U_1,\dots,U_r,X_1,\dots,X_n)$ to each node of $\Gamma$ as intermediate result.

We call the specializations which produce a failing assignment \emph{inconsistent}. If in the context of a given specialization of the scalar parameter variables of $\Gamma$ we instantiate for each index pair $(i,j)$, $1\leq i\leq l$, $1\leq j\leq L_i$ the variables $S_{i,j}$ and $T_{i,j}$ by two different values from $\left\{0,1\right\}$, the labelled directed acyclic graph $\Gamma$ becomes an ordinary arithmetic circuit over $\C$ of non--scalar depth at most $l$ and non--scalar size at most $L_1+\dots+L_l$ with the inputs $U_1,\dots,U_r,X_1,\dots,X_n$.

We may now find a suitable specialization of the circuit $\Gamma$ into a new circuit $\Gamma'$ over $\C$ such that the following condition is satisfied:\\
the (by $\mathcal{M}$) parameterized circuit obtained from $\Gamma'$ by replacing the special parameter variables $U_1,\dots,U_r$ by $\pi_1, \dots, \pi_r$, is consistent and can be reduced to the circuit $\beta$.	

We may consider the circuit $\Gamma$ as a generic computation which allows to recover $\beta$ by means of a suitable specialization of its scalar and special parameter variables into complex numbers and basic parameters $\pi_1, \dots, \pi_r$ and by means of circuit reductions. Moreover, any by $\mathcal{M}$ parameterized, consistent arithmetic circuit of non--scalar depth at most $l$, with inputs $X_1,\dots,X_n$ and $q$ outputs, which has an oblivious level structure with $l+2$ levels of width at most $L_0,\dots,L_{l+1}$, may be recovered from $\Gamma$ by suitable specializations and reductions (see \cite{Burgisser97}, Chapter 9 for more details on generic computations).

\subsection{A model for branching--free computation.}
\label{sec:Model-discussion}

\subsubsection{Requirements to be satisfied by our branching--free computation model. Informal discussion.}
\label{sec:Model-discussion-requirements}

We are now going to introduce a model of branching--free computation with parameterized arithmetic circuits. We shall first require that the routines of this computation model should be well behaved under restrictions of the inputs. We discuss this issue first informally.

Suppose for the moment that our branching--free computation model is already established. Then its routines transform a given robust parameterized arithmetic (input) circuit into another parameterized (output) circuit such that both circuits have the same parameter domain. Well behavedness under restrictions will be a property of circuit transformation that guarantees that the output circuit is still robust. In particular, we wish that the following requirement is satisfied.

Let $\mathcal{A}$ be a routine of our branching--free computation model and consider the previously introduced parameterized circuit $\beta$. Let $\mathcal{N}$ be a constructible subset of $\mathcal{M}$ and suppose that $\beta$ is an admissible input for the routine $\mathcal{A}$. Then $\mathcal{A}$ produces on input $\beta$ a parameterized arithmetic output circuit with parameter domain $\mathcal{M}$ which we denote by $\mathcal{A}(\beta)$. In order to formulate for the routine $\mathcal{A}$ our requirement, we must be able to restrict $\beta$ and $\mathcal{A}(\beta)$ to $\mathcal{N}$. Thus $\beta$ and $\mathcal{A}(\beta)$ should be \emph{robust}, $\beta_{\mathcal{N}}$ should be a consistent admissible input circuit for $\mathcal{A}$ and $\mathcal{A}(\beta_{\mathcal{N}})$ should be consistent too.

Our architectural requirement on the routine $\mathcal{A}$ may now be formulated as follows:

\begin{quote}
\emph{The parameterized arithmetic circuit $\mathcal{A}(\beta_{\mathcal{N}})$ can be recovered from $\mathcal{A}(\beta)$ by restriction to $\mathcal{N}$ and circuit reduction.}	
\end{quote}

Routines which are well behaved under restrictions will automatically satisfy this requirement.

\enter

The routine $\mathcal{A}$ performs with the parameterized arithmetic circuit $\beta$ a transformation whose crucial feature is that only nodes which depend on the inputs $X_1,\dots,X_n$ of $\beta$ become modified, whereas parameter nodes remain substantially preserved. This needs an explicitation. 

Suppose that $\beta$ has $t$ essential parameter nodes. Then the essential parameters (intermediate results) of $\beta$ associated with these nodes define a geometrically robust constructible map $\theta:\mathcal{M}\to\A^t$. The image $\mathcal{T}$ of $\theta$ is a constructible subset of $\A^t$. We require now that, as far as $\mathcal{A}$ performs arithmetic operations with parameters of $\beta$, $\mathcal{A}$ does it only with essential ones, and that all essential parameters of $\mathcal{A}(\beta)$ are obtained in this way. Further we require that there exists a geometrically robust constructible map $\nu$ defined on $\mathcal{T}$ (e.g., a polynomial map) such that the results of these arithmetic operations occur as entries of the composition map $\nu\circ \theta$. From Corollary \ref{proposition 1} we deduce that $\nu\circ\theta$ is a geometrically robust constructible map.  

\enter
Our basic construction method of routines will be recursion. A routine of our computation model which can be obtained in this way is called \emph{recursive}.

Suppose now that $\mathcal{A}$ is a recursive routine of our computation model. Then $\mathcal{A}$ should be organized in such a way that for each internal node $\rho$ of $\beta$, which depends on at least one input, there exists a set of nodes of $\mathcal{A}(\beta)$, also denoted by $\rho$, with the following property:
\\
the elements of the set $\rho$ of nodes of $\mathcal{A}(\beta)$ represent the outcome of the action of $\mathcal{A}$ at the node $\rho$ of $\beta$. 	

We fix now a node $\rho$ of $\beta$ which depends on at least one input. Let $G_{\rho}$ be the intermediate result associated with the node $\rho$ of $\beta$ and let $F_{\rho}$ be a vector whose entries are the intermediate results of $\mathcal{A}(\beta)$ at the nodes contained in the set $\rho$ of nodes of $\mathcal{A}(\beta)$. Thus $F_{\rho}$ is a vector of rational functions in a suitable tuple of (standard) variables, say $X'$. 

Recall that by assumption $\beta$ and $\mathcal{A}(\beta)$ are robust parameterized arithmetic circuits with parameter domain $\mathcal{M}$. Therefore we deduce from Lemma \ref{lemma intermediate results} that $G_{\rho}$ and the entries of $F_{\rho}$ are in fact polynomials in $X_1,\dots,X_n$ and $X'$, respectively, and that their coefficients are geometrically robust functions defined on $\mathcal{M}$. 

As part of our second and main requirement of our computation model we demand now that $\mathcal{A}$ satisfies at the node $\rho$ of $\beta$ the following isoparametricity condition:
\textit{\begin{enumerate}
	\item[(i)] for any two parameter instances $u_1$ and $u_2$ of $\mathcal{M}$ the assumption 
$$G_{\rho}(u_1,X_1,\dots,X_n)=G_{\rho}(u_2,X_1,\dots,X_n)$$
implies 
$$F_{\rho}(u_1,X')=F_{\rho}(u_2,X').$$
\end{enumerate}}

Let $\theta_{\rho}$ be the coefficient vector of $G_{\rho}$ and observe that $\theta_{\rho}$ is a geometrically robust constructible map defined on $\mathcal{M}$, whose image, say $\mathcal{T}_{\rho}$, is an irreducible constructible subset of a suitable affine space.

Since the first--order theory of the algebraically closed field $\C$ admits quantifier elimination, one concludes easily that condition $(i)$ is satisfied if and only if there exists a constructible map $\sigma_{\rho}$ defined on $\mathcal{T}_{\rho}$ such that the composition map $\sigma_{\rho}\circ\theta_{\rho} $ (which is also constructible) represents the coefficient vector of (all entries of) $F_{\rho}$.

In the sequel we shall need that the dependence $\sigma_{\rho}$ of the coefficient vector of $F_{\rho}$ on the coefficient vector of $G_{\rho}$ is in some stronger sense uniform (and not just constructible). Therefore we include the following condition in our requirement:
\textit{\begin{enumerate}
	\item[(ii)] the constructible map $\sigma_{\rho}$ is geometrically robust.
\end{enumerate}}

The map $\sigma_{\rho}$ is uniquely determined by condition $(i)$. Moreover, the map $\sigma_{\rho}$ depends on the (combinatorial) labelled DAG structure of $\beta$ below the node $\rho$, but not directly on the basic parameters $\pi_1,\dots,\pi_r$. This is the essence of the isoparametric nature of conditions $(i)$ and $(ii)$. We shall therefore require that our recursive routine is \emph{isoparametric} in this sense, i.e., that $\mathcal{A}$ satisfies conditions $(i)$ and $(ii)$ at any internal node $\rho$ of $\beta$ which depends at least on one input.

Observe that the geometrically robust constructible map $\sigma_{\rho}$ (which depends on $\beta$ as well as on $\rho$) is not an artifact, but emerges naturally from the recursive construction of a circuit semantic within the paradigm of object--oriented programming. To explain this, let notations and assumptions be as before and suppose that $\mathcal{A}$ is a isoparametric recursive routine of our model and that we apply $\mathcal{A}$ to the robust parameterized arithmetic circuit $\beta$. Let $\rho$ again be a node of $\beta$ which depends at least on one input. Let $u$ be a parameter instance of $\mathcal{M}$ and denote by $\beta^{(u)}$, $G_{\rho}^{(u)}$, $\mathcal{A}(\beta)^{(u)}$ and $F_{\rho}^{(u)}$ the instantiations of $\beta$, $G_{\rho}$, $\mathcal{A}(\beta)$ and $F_{\rho}$ at $u$ (observe that the intermediate results of $\beta^{(u)}$ and $\mathcal{A}(\beta)^{(u)}$ are well defined although we do not require that these circuits are consistent). Then the intermediate results of $\mathcal{A}(\beta)^{(u)}$ contained in $F_{\rho}^{(u)}$ depend only on the intermediate result $G_{\rho}^{(u)}$ of $\beta^{(u)}$ and not on the parameter instance $u$ itself. In this spirit we may consider the sets $\Gamma_{\rho}:=\{ G_{\rho}^{(u)} \ws ;\ws u\in\mathcal{M} \}$ and $\Phi_{\rho}:=\{ F_{\rho}^{(u)} \ws ; \ws u\in\mathcal{M} \}$ as abstract data types and $\beta$ and $\mathcal{A}(\beta)$ as syntactic descriptions of two abstraction functions which associate to any concrete object $u\in\mathcal{M}$ the abstract objects $G_{\rho}^{(u)}$ and $F_{\rho}^{(u)}$, respectively. The identity map $id_{\mathcal{M}}:\mathcal{M}\to\mathcal{M}$ induces now an \emph{abstract function} \cite{Meyer00} from $\Gamma_{\rho}$ to $\Phi_{\rho}$, namely $\sigma_{\rho}:\Gamma_{\rho}\to\Phi_{\rho}$. In this terminology, $id_{\mathcal{M}}$ is just an implementation of $\sigma_{\rho}$. If we now consider that each recursive step of the routine $\mathcal{A}$ on input $\beta$ has to be realized by some routine of the object--oriented programming paradigm, we arrive to a situation which requires the existence of a geometrically robust constructible map $\sigma_{\rho}:\Gamma_{\rho}\to\Phi_{\rho}$ as above.

We may interpret the map $\sigma_{\rho}:\Gamma_{\rho}\to\Phi_{\rho}$ also as an ingredient of a specification of the recursive routine $\mathcal{A}$. The map $\sigma_{\rho}$ may be thought as an operational specification which determines $F_{\rho}$ in function of $G_{\rho}$. A weaker specification would be a descriptive one which relates $G_{\rho}$ and $F_{\rho}$ without determining $F_{\rho}$ from $G_{\rho}$ completely.

\enter
In order to motivate the requirement that the recursive routine $\mathcal{A}$ should be isoparametric, we shall consider the following condition for recursive routines which we call \emph{well behavedness under reductions}. 

We only outline here this condition and leave the details until Section \ref{sec:Model-discussion-simplified}.

Suppose now that we apply a reduction procedure to the robust parameterized input circuit $\beta$ producing thus another robust, by $\mathcal{M}$ parameterized circuit $\beta^*$ which computes the same final results as $\beta$. Then the reduced circuit $\beta^*$ should also be an admissible input for the routine $\mathcal{A}$. We call the recursive routine $\mathcal{A}$ \emph{well behaved under reductions} if on input $\beta$ it is possible to extend the given reduction procedure to the output circuit $\mathcal{A}(\beta)$ in such a way, that the extended reduction procedure, applied to $\mathcal{A}(\beta)$, reproduces the circuit $\mathcal{A}(\beta^*)$. 

Obviously well behavedness under reductions limits the structure of $\mathcal{A}(\beta)$. Later, in Section \ref{sec:Model-discussion-simplified}, we shall see that, cum grano salis, any recursive routine, which is well behaved under restrictions and reductions, is necessarily isoparametric. Since well behavedness under restrictions and reductions are very natural quality attributes for routines which transform robust parameterized arithmetic circuits, the weaker requirement, namely that recursive routines should be isoparametric, turns out to be well motivated.

In Section \ref{sec:Model-discussion-simplified}, we shall formally introduce our branching--free computation model. We postpone for then the precise definitions of the notions of well behavedness under restrictions and reductions.

There exists a second reason to limit the recursive routines of our branching--free computation model to isoparametric ones. Isoparametric recursive routines have considerable advantages for program specification and verification by means of Hoare Logics (see \cite{Apt81}). We shall come back to this issue in Section \ref{Isoparametricity and program specification}. 

\subsubsection{The branching--free computation model}
\label{sec:Model-discussion-simplified}

The computation model we are going to introduce in this and the next subsection will be comprehensive enough to capture the essence of all known circuit based elimination algorithms in effective Algebraic Geometry and, mutatis mutandis, also of all other (linear algebra and truncated rewriting) elimination procedures (see Sections \ref{sec:Model-discussion-programs and algorithms}, \ref{sec:Model-Applications}, \cite{Mora03}, \cite{Mora05}, and the references cited therein, and for truncated rewriting methods especially \cite{Dickens}). The only algorithm from symbolic arithmetic circuit manipulation which will escape from our model is the Baur--Strassen gradient computation \cite{Burgisser97}, Chapter 7.2.

In the sequel we shall distinguish sharply between the notions of input variable and parameter and the corresponding categories of circuit nodes.

Input variables, called ``standard'', will occur in parameterized arithmetic circuits and generic computations. The input variables of generic computations will appear subdivided in three sorts, namely as ``parameter'', ``argument'' and ``standard'' input variables.

The branching--free computation model we are going to introduce in this subsection will assume different \emph{shapes}, each shape being determined by a finite number of a priori given \emph{discrete} (i.e., by tuples of natural numbers indexed) families of generic computations. The labels of the inputs of the ordinary arithmetic circuits which represent these generic computations will become subdivided into \emph{parameter}, \emph{argument} and \emph{standard} input variables. We shall use the letters like $U,U',U'',\dots$ and $W,W',W''$ to denote vectors of parameters, $Y,Y',Y'',\dots$ and $Z,Z',Z''$ to denote vectors of argument and $X,X',X'',\dots$ to denote vectors of standard input variables (see Section \ref{sec:Model-generic computations}). 

We shall not write down explicitly the indexations of our generic computations by tuples of natural numbers. Generic computations will simply be distinguished by subscripts and superscripts, if necessary.

Ordinary arithmetic circuits of the form
\begin{center}
\begin{tabular}{l l l}
$R_{X_1}(W_{1};X^{(1)})$, 		& $R_{X_2}(W_{2};X^{(2)})$, 		& $\dots$ \\ 
$R_{X_1}'(W_{1'};X^{(1')})$, 	& $R_{X_2}'(W_{2'};X^{(2')})$, 	& $\dots$ \\ 
$\dots$ 											& $\dots$ 											& $\dots$ \\ 
\end{tabular}
\end{center}
represent a first type of a discrete family of generic computations (for each variable $X_1,X_2,\dots,X_n$ we suppose to have at least one generic computation). Other types of families of generic computations are of the form
\begin{center}
\begin{tabular}{l l l l}
$R_+(W;U,Y;X)$, 		& $R_+'(W';U',Y';X')$, 		& $R_+''(W'';U'',Y'';X'')$	 & $\dots$ \\ 
$R_{._/}(W;U,Y;X)$, & $R_{._/}'(W';U',Y';X')$, & $R_{._/}''(W'';U'',Y'';X'')$	 &$\dots$ \\ 
$R_{add}(W;Y,Z;X)$, & $R_{add}'(W';Y',Z';X')$, & $R_{add}''(W'';Y'',Z'';X'')$	 &$\dots$ \\ 
$R_{mult}(W;Y,Z;X)$,& $R_{mult}'(W';Y',Z';X')$,& $R_{mult}''(W'';Y'',Z'';X'')$	 & $\dots$ \\ 
\end{tabular}
\end{center}
and
\begin{center}
\begin{tabular}{l l l l}
$R_{div}(W;Y,Z;X)$, & $R_{div}'(W';Y',Z';X')$, & $R_{div}''(W'';Y'',Z'';X'')$	& $\dots$. \\ 
\end{tabular}
\end{center}
Here the subscripts refer to addition of, and multiplication or division by a parameter (or scalar) and to essential addition, multiplication and division. A final type of families of generic computations is of the form
$$R(W;Y;X),\ws R'(W';Y';X'),\ws R''(W'';Y'';X''),\dots$$ 
We recall from Section \ref{sec:Model-discussion-requirements} that the objects handled by the routines of any shape of our computation model will always be robust parameterized arithmetic circuits. The inputs of these circuits will only consist of standard variables.

From now on we have in mind a previously fixed shape when we refer to the branching--free computation model we are going to introduce. We start with a given finite set of discrete families of generic computations which constitute a shape as described before.


\paragraph{The notion of well behavedness under restrictions}\ws\ws

A fundamental issue is how we recursively transform a given input circuit into another one with the same parameter domain. During such a transformation we make an iterative use of previously fixed generic computations. On their turn these determine the corresponding \emph{recursive routine} of our branching--free computation model.

We consider again our input circuit $\beta$. We suppose that we have already chosen for each node $\rho$, which depends at least on one of the input variables $X_1,\dots,X_n$, a generic computation 
$$R^{(\rho)}_{X_i}(W_{\rho};X^{(\rho)}),$$
$$R^{(\rho)}_{+}(W_{\rho};U_{\rho},Y_{\rho};X^{(\rho)}),$$
$$R^{(\rho)}_{._/}(W_{\rho};U_{\rho},Y_{\rho};X^{(\rho)}),$$
$$R^{(\rho)}_{add}(W_{\rho};Y_{\rho},Z_{\rho};X^{(\rho)}),$$
$$R^{(\rho)}_{mult}(W_{\rho};Y_{\rho},Z_{\rho};X^{(\rho)}),$$ 
$$R^{(\rho)}_{div}(W_{\rho};Y_{\rho},Z_{\rho};X^{(\rho)}),$$

and that this choice was made according to the label of $\rho$, namely $X_i, 1\leq i\leq n$, or addition of, or multiplication or division by an essential parameter, or essential addition, multiplication or division. Here we suppose that $U_{\rho}$ is a single variable, whereas $W_{\rho},Y_{\rho},Z_{\rho}$ and $X^{(\rho)}$ may be arbitrary vectors of variables.

Furthermore, we suppose that we have already precomputed for each node $\rho$ of $\beta$, which depends at least on one input, a vector $w_{\rho}$ of geometrically robust constructible functions defined on $\mathcal{M}$. If $\rho$ is an input node we assume that $w_{\rho}$ is a vector of complex numbers. Moreover, we assume that the length of $w_{\rho}$ equals the length of the variable vector $W_{\rho}$. We call the entries of $w_{\rho}$ the \emph{parameters at the node $\rho$} of the routine $\mathcal{A}$ applied to the input circuit $\beta$. 

We are now going to develop the routine $\mathcal{A}$ step by step. The routine $\mathcal{A}$ takes over all computations of $\beta$ which involve only parameter nodes, without modifying them.

Consider an arbitrary internal node $\rho$ of $\beta$ which depends at least on one input. The node $\rho$ has two ingoing edges which come from two other nodes of $\beta$, say $\rho_1$ and $\rho_2$. Suppose that the routine $\mathcal{A}$, on input $\beta$, has already computed two results, namely $F_{\rho_1}$ and $F_{\rho_2}$, corresponding to the nodes $\rho_1$ and $\rho_2$. Suppose inductively that these results are vectors of polynomials depending on those standard input variables that occur in the vectors of the form $X^{(\rho')}$, where $\rho'$ is any predecessor node of $\rho$. Furthermore, we assume that the coefficients of these polynomials constitute the entries of a geometrically robust, constructible map defined on $\mathcal{M}$. Finally we suppose that the lengths of the vectors $F_{\rho_1}$ and $Y_{\rho}$ (or $U_{\rho}$) and $F_{\rho_2}$ and $Z_{\rho}$ coincide. 

The parameter vector $w_{\rho}$ of the routine $\mathcal{A}$ forms a geometrically robust, constructible map defined on $\mathcal{M}$, whose image we denote by $\mathcal{K}_{\rho}$. Observe that $\mathcal{K}_{\rho}$ is a constructible subset of the affine space of the same dimension as the length of the vectors $w_{\rho}$ and $W_{\rho}$. Denote by $\kappa_{\rho}$ the vector of the restrictions to $\mathcal{K}_{\rho}$ of the canonical projections of this affine space. We consider $\mathcal{K}_{\rho}$ as a new parameter domain with basic parameters $\kappa_{\rho}$. For the sake of simplicity we suppose that the node $\rho$ is labelled by a multiplication. Thus the corresponding generic computation has the form
\begin{equation}
\label{(1)}
R^{(\rho)}_{._/}(W_{\rho};U_{\rho},Y_{\rho};X^{(\rho)})	
\end{equation}
or
\begin{equation}
\label{(2)}
R^{(\rho)}_{mult}(W_{\rho};Y_{\rho},Z_{\rho};X^{(\rho)}).	
\end{equation}

Let the specialized generic computations 
$$R_{._/}^{(\rho)}(\kappa_{\rho},U_{\rho},Y_{\rho},X^{(\rho)}) \text{\ws and \ws} R_{mult}^{(\rho)}(\kappa_{\rho},Y_{\rho},Z_{\rho},X^{(\rho)})$$
be the by $\mathcal{K}_{\rho}$ parameterized arithmetic circuits obtained by substituting in the generic computations \eqref{(1)} and \eqref{(2)} for the vector of parameter variables $W_{\rho}$ the basic parameters $\kappa_{\rho}$. At the node $\rho$ we shall now make the following requirement on the routine $\mathcal{A}$ applied to the input circuit $\beta$:
\textit{\begin{enumerate}
	\item[(A)]  The by $\mathcal{K}_{\rho}$ parameterized arithmetic circuit which corresponds to the current case, namely
$$R^{(\rho)}_{._/}(\kappa_{\rho};U_{\rho},Y_{\rho};X^{(\rho)})$$
or
$$R^{(\rho)}_{mult}(\kappa_{\rho};Y_{\rho},Z_{\rho};X^{(\rho)}),$$ 
should be consistent and robust.
\end{enumerate}}
 
Observe that the requirement $(A)$ is automatically satisfied if all the generic computations of our shape are realized by totally division--free ordinary arithmetic circuits. 

Assume now that the routine $\mathcal{A}$ applied to the circuit $\beta$ satisfies the requirement $(A)$ at the node $\rho$ of $\beta$.

Let us first suppose that the node $\rho$ is labelled by a multiplication involving an essential parameter. Recall that in this case we assumed earlier that the length of the vector $F_{\rho_1}$ is one, that $F_{\rho_1}$ is an essential parameter of $\beta$ and that the vectors $F_{\rho_2}$ and $Y_{\rho}$ have the same length. Joining now with the generic computation $R^{(\rho)}_{._/}(W_{\rho};U_{\rho},Y_{\rho};X^{(\rho)})$ at $W_{\rho},U_{\rho}$ and $Y_{\rho}$ the previous computations of $w_{\rho},F_{\rho_1}$ and $F_{\rho_2}$, we obtain a parameterized arithmetic circuit with parameter domain $\mathcal{M}$, whose final results are the entries of a vector which we denote by $F_{\rho}$.

Suppose now that the node $\rho$ is labelled by an essential multiplication. Recall again that in this second case we assumed earlier the vectors $F_{\rho_1}$ and $Y_{\rho}$ and $F_{\rho_2}$ and $Z_{\rho}$ have the same length. Joining with the generic computation 
$$R^{(\rho)}_{mult}(W_{\rho};Y_{\rho},Z_{\rho};X^{(\rho)})$$ 
at $W_{\rho}, Y_{\rho}$ and $Z_{\rho}$ the previous computations of $w_{\rho}, F_{\rho_1}$ and $F_{\rho_2}$ we obtain also a parameterized arithmetic circuit with parameter domain $\mathcal{M}$, whose final results are the entries of a vector which we denote again by $F_{\rho}$. 

One deduces easily from our assumptions on $w_{\rho},F_{\rho_1}$ and $F_{\rho_2}$ and from the requirement $(A)$ in combination with Lemma \ref{lemma intermediate results} and Corollary \ref{proposition 1}, that in both cases the resulting parameterized arithmetic circuit is robust if it is consistent. The other possible labellings of the node $\rho$ by arithmetic operations are treated similarly. In particular, in case that $\rho$ is an input node labelled by the variable $X_i, 1\leq i\leq n$, the requirement $(A)$ implies that the ordinary arithmetic circuit $R^{(\rho)}_{X_i}(w_{\rho};X^{(\rho)})$ is consistent and robust and that all its intermediate results are polynomials in $X^{(\rho)}$ over $\C$ (although $R^{(\rho)}_{X_i}(w_{\rho};X^{(\rho)})$ may contain divisions). 

In view of our comments in Section \ref{sec:Model-discussion-requirements}, we call the recursive routine $\mathcal{A}$ (on input $\beta$) \emph{well behaved under restrictions} if the requirement $(A)$ is satisfied at any node $\rho$ of $\beta$ which depends at least on one input and if joining the corresponding generic computation with $w_{\rho}$, $F_{\rho_1}$ and $F_{\rho_2}$ produces a consistent circuit (observe that this last condition is automatically satisfied when the specialized generic computation of $(A)$ is essentially division--free). If the routine $\mathcal{A}$ is well behaved under restrictions, then $\mathcal{A}$ transforms step by step the input circuit $\beta$ into another consistent robust arithmetic circuit, namely $\mathcal{A}(\beta)$, with parameter domain $\mathcal{M}$. Thus, well behavedness under restrictions guarantees that the recursive routine $\mathcal{A}$ transforms robust parameterized arithmetic circuits in robust ones.

As a consequence of the recursive structure of $\mathcal{A}(\beta)$, each node $\rho$ of $\beta$ generates a subcircuit of $\mathcal{A}(\beta)$ which we call the component of $\mathcal{A}(\beta)$ generated by $\rho$. The output nodes of each component of $\mathcal{A}(\beta)$ form the hypernodes of a hypergraph $\mathcal{H}_{\mathcal{A}(\beta)}$ whose hyperedges are given by the paths connecting the nodes of $\mathcal{A}(\beta)$ contained in distinct hypernodes of $\mathcal{H}_{\mathcal{A}(\beta)}$. The hypergraph $\mathcal{H}_{\mathcal{A}(\beta)}$ may be shrunk to the DAG structure of $\beta$ and therefore we denote the hypernodes of $\mathcal{H}_{\mathcal{A}(\beta)}$ in the same way as the nodes of $\beta$. Notice that well behavedness under restrictions is in fact a property which concerns the hypergraph $\mathcal{H}_{\mathcal{A}(\beta)}$. 

We call $\mathcal{A}$ a (recursive) \emph{parameter routine} if $\mathcal{A}$ does not introduce new standard variables. In the previous recursive construction of the routine $\mathcal{A}$, the parameters at the nodes of $\beta$, used for the realization of the circuit $\mathcal{A}(\beta)$, are supposed to be generated by recursive parameter routines.

\paragraph{The notion of isoparametricity}\ws\ws

We are now going to consider another requirement of our recursive routine $\mathcal{A}$, which will lead us to the notion of \emph{isoparametricity} of $\mathcal{A}$. Isoparametricity will guarantee that the recursive routine $\mathcal{A}$ may be specified (see Section \ref{Isoparametricity and program specification}).

Let us turn back to the previous situation at the node $\rho$ of the input circuit $\beta$. Notations and assumptions will be the same as before. From Lemma \ref{lemma intermediate results} we deduce that the intermediate result of $\beta$ associated with the node $\rho$, say $G_{\rho}$, is a polynomial in $X_1,\dots,X_n$ whose coefficients form the entries of a geometrically robust, constructible map defined on $\mathcal{M}$, say $\theta_{\rho}$. Let $\mathcal{T}_{\rho}$ be the image of this map and observe that $\mathcal{T}_{\rho}$ is a constructible subset of a suitable affine space. The intermediate results of the circuit $\mathcal{A}(\beta)$ at the elements of the hypernode $\rho$ of $\mathcal{H}_{\mathcal{A}(\beta)}$ constitute a polynomial vector which we denote by $F_{\rho}$.

We shall now make another requirement on the routine $\mathcal{A}$ at the node $\rho$ of $\beta$.
\textit{\begin{enumerate}
	\item[(B)] There exists a geometrically robust constructible map $\sigma_{\rho}$ defined on $\mathcal{T}_{\rho}$ such that $\sigma_{\rho}\circ \theta_{\rho}$ constitutes the coefficient vector of $F_{\rho}$.
\end{enumerate}}

In view of the comments made in Section \ref{sec:Model-discussion-requirements} we call the recursive routine $\mathcal{A}$ \emph{isoparametric} (on input $\beta$) if requirements $(A)$ and $(B)$ are satisfied at any node $\rho$ of $\beta$ which depends at least on one input. 

\enter
Let assumptions and notations be as before and consider again the node $\rho$ of the circuit $\beta$. Assume that the recursive routine $\mathcal{A}$ is well behaved under restrictions and denote by $\tau_{\rho}$ the coefficient vector of $F_{\rho}$. Observe that $\tau_{\rho}$ is a geometrically robust constructible map defined on $\M$. Assume, furthermore, that $\mathcal{A}$, applied to the circuit $\beta$, fulfils the requirement $(B)$ at $\rho$. Then the topological robustness (which is a consequence of the geometrical robustness) of $\sigma_{\rho}$ implies that the following condition is satisfied:
\textit{
\begin{enumerate}
	\item[$(B')$] Let $(u_k)_{k\in\N}$ be a (not necessarily convergent) sequence of parameter instances $u_k\in\mathcal{M}$ and let $u\in\mathcal{M}$ such that $(\theta_{\rho}(u_k))_{k\in\N}$ converges to $\theta_{\rho}(u)$. Then the sequence $(\tau_{\rho}(u_k))_{k\in\N}$ is bounded. 
\end{enumerate}
}
Suppose now that the recursive routine $\mathcal{A}$ is well behaved under restrictions and satisfies instead of $(B)$ only condition $(B')$ at the node $\rho$ of $\beta$. Let $u\in\mathcal{M}$ be an arbitrary parameter instance. Then Theorem \ref{th: equiv robus} implies that $\tau_{\rho}$ takes on the set $\{ u'\in\mathcal{M}; \theta_{\rho}(u')=\theta_{\rho}(u) \}$ only finitely many values. In particular, for $\mathfrak{M}_u$ being the vanishing ideal of the $\C$--algebra $\C[\theta_{\rho}]$ at $\theta_{\rho}(u)$, the entries of $\tau_{\rho}$ are integral over the local $\C$--algebra $\C[\theta_{\rho}]_{\mathfrak{M}_u}$ (the argument for that relies on Zariski's Main Theorem and is exhibited in \cite{CaGiHeMaPa03}, Sections 3.2 and 5.1). This algebraic characterization implies that for given $u\in\mathcal{M}$ all the sequences $(\tau_{\rho}(u_k))_{k\in\N}$ of condition $(B')$ have only finitely many distinct accumulation points. This shows that requirement $(B)$ and condition $(B')$ are closely related.

Adopting the terminology of \cite{GHMS09}, we call $\mathcal{A}$ \emph{coalescent} (on input $\beta$), if $\mathcal{A}$ is well behaved under restrictions and satisfies condition $(B')$ for any node $\rho$ of $\beta$. Thus isoparametricity implies coalescence for $\mathcal{A}$, but not vice versa. Nevertheless the notions of isoparametricity and coalescence become quite close for recursive routines which are well behaved under restrictions.

\paragraph{The notion of well behavedness under reductions}\ws\ws

Suppose again that the recursive routine $\mathcal{A}$ is well behaved under restrictions. We call $\mathcal{A}$ \emph{well behaved under reductions} (on input $\beta$) if $\mathcal{A}(\beta)$ satisfies the following requirement:

\begin{quote}
\textit{Let $\rho$ and $\rho'$ be distinct nodes of $\beta$ which compute the same intermediate results. Then the intermediate results at the hypernodes $\rho$ and $\rho'$ of $\mathcal{H}_{\mathcal{A}(\beta)}$ are identical. Mutatis mutandis the same is true for the computation of the parameters of $\mathcal{A}$ at any node of $\beta$.
}	
\end{quote}

Assume that the routine $\mathcal{A}$ is recursive and well behaved under reductions. One verifies then easily that, taking into account the hypergraph structure $\mathcal{H}_{\mathcal{A}(\beta)}$ of $\mathcal{A}(\beta)$, any reduction procedure on $\beta$ may canonically be extended to a reduction procedure of $\mathcal{A}(\beta)$.

In Section \ref{sec:Model-discussion-requirements} we claimed that, cum grano salis, the requirement of well behavedness under reductions implies the requirement of isoparametricity for recursive routines. We are going now to prove this.
 
Let notations and assumptions be as before and let us analyze what happens to the recursive routine $\mathcal{A}$ at the node $\rho$ of $\beta$. For this purpose we shall use the following broadcasting argument.

Recall that $G_{\rho}$ and the entries of $F_{\rho}$ are the intermediate results of $\beta$ and $\mathcal{A}(\beta)$ associated with $\rho$, where $\rho$ is interpreted as a node of the input circuit $\beta$ in the first case and as a hypernode of $\mathcal{H}_{\mathcal{A}(\beta)}$ in the second one. Moreover recall that $G_{\rho}$ is a polynomial in $X_1,\dots,X_n$, that the geometrically robust, constructible map $\theta_{\rho}$, defined on $\mathcal{M}$, represents the coefficient vector of $G_{\rho}$ and that the irreducible constructible set $\mathcal{T}_{\rho}$ is the image of $\theta_{\rho}$. Observe that the entries of $\theta_{\rho}$ may be computed from $\pi_1,\dots,\pi_r$ by a robust arithmetic circuit (e.g., by interpolation of $G_{\rho}$ in sufficiently generic points of $\A^n$). We consider now the robust parameterized arithmetic circuit $\gamma_{\rho}$ which realizes the following trivial evaluation of the polynomial $G_{\rho}$:
\begin{enumerate}
	\item[-] compute simultaneously from $\pi_1,\dots,\pi_r$ all entries of $\theta_{\rho}$ and from $X_1,\dots,X_n$ all monomials occurring in $G_{\rho}$ 
	\item[-] compute $G_{\rho}$ as a linear combination of the monomials of $G_{\rho}$ using as coefficients the entries of $\theta_{\rho}$.
\end{enumerate}

The circuit $\gamma_{\rho}$ has a single output node, say $\rho'$, which computes the polynomial $G_{\rho}$.

Now we paste, as disjointly as possible, the circuit $\gamma_{\rho}$ to the circuit $\beta$ obtaining thus a new robust, parameterized arithmetic circuit $\beta_{\rho}$ with parameter domain $\mathcal{M}$. Observe that $\beta_{\rho}$ contains $\beta$ and $\gamma_{\rho}$ as subcircuits and that $\rho$ and $\rho'$ are distinct nodes of $\beta_{\rho}$ which compute the same intermediate result, namely $G_{\rho}$. The entries of $\theta_{\rho}$ are essential parameters of $\gamma_{\rho}$ and hence also of $\beta_{\rho}$. We suppose now that $\beta_{\rho}$ is, like $\beta$, an admissible input for the recursive routine $\mathcal{A}$. Let $F_{\rho'}$ be a vector whose entries are the intermediate results at the nodes of $\mathcal{A}(\beta_{\rho})$ contained in the hypernode $\rho'$ of $\mathcal{H}_{\mathcal{A}(\beta_{\rho})}$. Analyzing now how $\mathcal{A}$ operates on the structure of the subcircuit $\gamma_{\rho}$ of $\beta_{\rho}$, we see immediately that there exists a geometrically robust constructible map $\sigma_{\rho}$ defined on $\mathcal{T}_{\rho}$ such that the composition map $\sigma_{\rho}\circ\theta_{\rho}$ constitutes the coefficient vector of $F_{\rho'}$. Since by assumption the recursive routine $\mathcal{A}$ is well behaved under reductions and the intermediate results of $\beta_{\rho}$ at the nodes $\rho$ and $\rho'$ consist of the same polynomial $G_{\rho}$, we conclude that the intermediate results at the hypernodes $\rho$ and $\rho'$ of $\mathcal{H}_{\mathcal{A}(\beta_{\rho})}$ are also the same. Therefore we may assume without loss of generality $F_{\rho}=F_{\rho'}$. Hence the geometrically robust, constructible map $\sigma_{\rho}\circ\theta_{\rho}$ constitutes the coefficient vector of $F_{\rho}$.

This proves that the recursive routine $\mathcal{A}$ satisfies, on input $\beta$ and at the node $\rho$, the requirement $(B)$. Since $\beta$ was an arbitrary admissible input circuit for the recursive routine $\mathcal{A}$ and $\rho$ was an arbitrary node of $\beta$ which depends on at least one input, we may conclude that $\mathcal{A}$ is isoparametric. The only assumption we made to draw this conclusion was that the extended circuit $\beta_{\rho}$ is an admissible input for the routine $\mathcal{A}$. This conclusion is however not very restrictive because $\beta$ and $\beta_{\rho}$ compute the same final results.

\paragraph{Isoparametricity and program specification}\ws\ws
\label{Isoparametricity and program specification}

In Section \ref{sec:Model-discussion-requirements}, we mentioned that isoparametric routines are advantageous for program specification and verification. We are now going to explain this.

Let notations and assumptions be as before and let in particular $\mathcal{A}$ be a recursive routine of our computation model which behaves well under restrictions. Assume that $\beta$ is an admissible input for $\mathcal{A}$ and consider the specification language $\mathcal{L}$ introduced in Section \ref{sec: A specification language for circuits}. Suppose that the routine $\mathcal{A}$ is given by an asserted program $\Pi$ formulated in the elementary Hoare Logics of $\mathcal{L}$ (\cite{Apt81}). The standard model of the elementary theory of $\mathcal{L}$ provides us with the states which define the semantics of $\Pi$. The asserted program $\Pi$ represents the routine $\mathcal{A}$ as a loop which transforms node by node the labelled DAG structure of $\beta$ into the labelled DAG structure of $\mathcal{A}(\beta)$. 

At each step of the loop a purely syntactic action, namely a graph manipulation, takes place. This action consists of the join of two or more labelled directed acyclic graphs. Simultaneously, in order to guarantee the correctness of the program $\Pi$, a loop invariant, formulated in our specification language $\mathcal{L}$, has to be satisfied. 

This involves the semantics of $\mathcal{L}$ consisting of the universe of all robust parameterized arithmetic circuits. A loop invariant as above is given by a formula $\bigwedge(\beta_1,\beta_2,\mathcal{M}_1,\rho_1)$ of $\mathcal{L}$ containing the free variables $\beta_1$, $\beta_2$ for circuits over the same parameter domain $\M_1$ and $\rho_1$ for a node of $\beta_1$ and a linked hypernode of $\beta_2$, such that these free variables become instantiated by $\beta$, $\mathcal{A}(\beta)$, $\mathcal{M}$ and the node $\rho$ of $\beta$ or the hypernode $\rho$ of $\mathcal{A}(\beta)$. The variables $U^{(1)},\dots,U^{(m)},\dots$ and the standard input variable vectors $X^{(1)},\dots,X^{(h)},\dots$ occur only bounded in $\bigwedge(\beta_1,\beta_2,\mathcal{M}_1,\rho_1)$ and the variables $\rho_1,\dots,\rho_l,\dots$ occur all bounded except one, namely $\rho_1$.

For $\pi:=(\pi_1,\dots,\pi_r)$ and given variables $X,X'$ and $\rho$ expressing a parameter instantiation, the input variable vectors of $\beta$ and $\mathcal{A}(\beta)$ and a node of $\beta$, we denote by $G_{\rho}(\beta;\pi;X)$ and $F_{\rho}(\mathcal{A}(\beta);\pi;X')$ the function symbols (or vectors of them) which express the intermediate results of $\beta$ or $\mathcal{A}(\beta)$ corresponding to $\rho$.

We require now that any formula of $\mathcal{L}$ built up by $G_{\rho_1},\dots,G_{\rho_l}$ and $F_{\rho_1'},\dots,F_{\rho_{l'}'}$, and containing only $\beta$, $\mathcal{M}$ and $\rho_1$ as free variables is equivalent to a formula built up only by $G_{\rho_1},\dots,G_{\rho_l}$ and $G_{\rho_1'},\dots,G_{\rho_{l'}'}$. This implies that in $\mathcal{L}$ the intermediate result $F_{\rho}$ of $\mathcal{A}(\beta)$ is definable in terms of the intermediate result $G_{\rho}$ of $\beta$. Applied to the node $\rho$ of the concrete circuit $\beta$ with parameter domain $\mathcal{M}$, this means that for $\theta_{\rho}$ and $\tau_{\rho}$ being the coefficient vectors of $G_{\rho}(\beta,\pi,X)$ and $F_{\rho}(\mathcal{A}(\beta),\pi,X')$ and $\mathcal{T}_{\rho}$ being the image of $\theta_{\rho}$, there exists a constructible map $\sigma_{\rho}$ with domain of definition $\mathcal{T}_{\rho}$ such that $\tau_{\rho}=\sigma_{\rho}\circ \theta_{\rho}$ holds. In particular, for $u',u''\in\mathcal{M}$ the assumption $\theta_{\rho}(u')=\theta_{\rho}(u'')$ implies $\tau_{\rho}(u')=\tau_{\rho}(u'')$. 

For the modelling of elimination algorithms this is a reasonable requirement (see Section \ref{sec:Model-Applications}). If we require additionally that the transformation of $G_{\rho}(\beta,\pi,X)$ into $F_{\rho}(\mathcal{A}(\beta),\pi,X')$ is continuous, then the constructible map $\sigma_{\rho}$ has to be geometrically robust (see Section \ref{sec:Model-discussion-requirements}).

In terms of the specification language $\mathcal{L}$, this reasoning may be formulated as follows.

Let $\beta_1,\beta_2,\mathcal{M}_1$ and $\rho_1$ be variables for robust parameterized arithmetic circuits, their parameter domains and their (hyper)nodes. We assume that there exist a formula 
$$\Omega(\beta_1,\beta_2,\mathcal{M}_1,\rho_1)$$
in the free variables $\beta_1,\beta_2,\mathcal{M}_1,\rho_1$ such that for any concrete, for $\mathcal{A}$ admissible circuit $\beta$ with parameter domain $\mathcal{M}$ and basic parameter vector $\pi$ and for any node $\rho$ of $\beta$ the following condition is satisfied:
\begin{itemize} 
	\item[$(*)$] $\Omega(\beta,\mathcal{A}(\beta),\mathcal{M},\rho)$ determines the polynomial $F_{\rho}(\mathcal{A}(\beta),\pi,X')$\\ in terms of $G_{\rho}(\beta,\pi,X)$. 
\end{itemize}

If $\mathcal{L}$ and $\mathcal{A}$ satisfy this assumption we say in the spirit of Hoare Logics that $\mathcal{L}$ is \emph{expressive} for the routine $\mathcal{A}$. 

Observe that condition $(*)$ guarantees that a postcondition for the circuit $\mathcal{A}(\beta)$ can always be translated into an equivalent precondition for the circuit $\beta$. 

\paragraph{Operations with routines}\ws\ws

Let $\mathcal{A}$ and $\mathcal{B}$ be recursive routines as before and suppose that they are well behaved under restrictions and isoparametric or even well behaved under reductions. Assume that $\mathcal{A}(\beta)$ is an admissible input for $\mathcal{B}$. We define the composed routine $\mathcal{B}\circ\mathcal{A}$ in such a way that $(\mathcal{B}\circ\mathcal{A})(\beta)$ becomes the parameterized arithmetic circuit $\mathcal{B}(\mathcal{A}(\beta))$. Since the routines $\mathcal{A}$ and $\mathcal{B}$ are well behaved under restrictions, we see easily that $(\mathcal{B}\circ\mathcal{A})(\beta)$ is a consistent, robust parameterized arithmetic circuit with parameter domain $\mathcal{M}$. From Lemma \ref{lemma intermediate results} and Corollary \ref{proposition 1} we deduce that $\mathcal{B}\circ\mathcal{A}$ is a isoparametric recursive routine if $\mathcal{A}$ and $\mathcal{B}$ are isoparametric. In case that $\mathcal{A}$ and $\mathcal{B}$ are well behaved under reductions, one verifies immediately that $\mathcal{B}\circ\mathcal{A}$ is also well behaved under reductions. Therefore, under these assumptions, we shall consider $\mathcal{B}\circ\mathcal{A}$ also as a routine of our computation model. 

Unfortunately, the composition of two arbitrary coalescent recursive routines need not to be coalescent. Therefore we shall focus in the sequel our attention on isoparametric recursive routines as basic building blocks of the branching--free computation model we are going to introduce. 
  
The identity routine is trivially well behaved under restrictions and reductions and in particular isoparametric.

Let $\mathcal{A}$ and $\mathcal{B}$ be two routines of our computation model and suppose for the sake of simplicity that they are recursive and well behaved under restrictions. Assume that the robust parameterized arithmetic circuit $\beta$ is an admissible input for $\mathcal{A}$ and $\mathcal{B}$ and that there is given a one--to--one correspondence $\lambda$ which identifies the output nodes of $\mathcal{A}(\beta)$ with the input nodes of $\mathcal{B}(\beta)$. Often, for a given input circuit $\beta$, the correspondence $\lambda$ is clear by the context. If we limit ourselves to input circuits $\beta$ where this occurs, we obtain from $\mathcal{A}$ and $\mathcal{B}$ a new routine, called their \emph{join}, which transforms the input circuit $\beta$ into the output circuit $\mathcal{B}(\beta)*_{\lambda}\mathcal{A}(\beta)$ (here we suppose that $\mathcal{B}(\beta)*_{\lambda}\mathcal{A}(\beta)$ is consistent). Analyzing now $\mathcal{B}(\beta)*_{\lambda}\mathcal{A}(\beta)$, we see that the join of $\mathcal{A}$ with $\mathcal{B}$ is well behaved under restrictions in the most obvious sense. Since by assumption the routines $\mathcal{A}$ and $\mathcal{B}$ are recursive, the circuits $\mathcal{A}(\beta)$ and $\mathcal{B}(\beta)$ inherit from $\beta$ a superstructure given by the hypergraphs $\mathcal{H}_{\mathcal{A}(\beta)}$ and $\mathcal{H}_{\mathcal{B}(\beta)}$. Analyzing again this situation, we see that any reduction procedure on $\beta$ can be extended in a canonical way to the circuit $\mathcal{B}(\beta)*_{\lambda}\mathcal{A}(\beta)$. This means that the join of $\mathcal{A}$ with $\mathcal{B}$ is also well behaved under reductions if the same is true for $\mathcal{A}$ and $\mathcal{B}$. More caution is at order with the notions of isoparametricity and coalescence. In a strict sense, the join of two isoparametric or coalescent recursive routines $\mathcal{A}$ and $\mathcal{B}$ is not necessarily isoparametric or coalescent. However, conditions $(B)$ or $(B')$ are still satisfied between the output nodes of $\beta$ and $\mathcal{B}(\beta)*_{\lambda}\mathcal{A}(\beta)$. A routine with one of these two properties is called \emph{output isoparametric} or \emph{output coalescent}, respectively.

The \emph{union} of the routines $\mathcal{A}$ and $\mathcal{B}$ assigns to the input circuit $\beta$ the juxtaposition of $\mathcal{A}(\beta)$ and $\mathcal{B}(\beta)$. Thus, on input $\beta$, the final results of the union of $\mathcal{A}$ and $\mathcal{B}$ are the final results of $\mathcal{A}(\beta)$ and $\mathcal{B}(\beta)$ (taken separately in case of ambiguity). The union of $\mathcal{A}$ and $\mathcal{B}$ behaves well under restrictions and reductions and is isoparametric if the same is true for $\mathcal{A}$ and $\mathcal{B}$. 

Observe also that for a recursive routine $\mathcal{A}$ which behaves well under restrictions and reductions the following holds: let $\beta$ be a robust parameterized arithmetic circuit that broadcasts to a circuit $\beta^*$ and assume that $\beta$ and $\beta^*$ are admissible circuits for $\mathcal{A}$. Then $\mathcal{A}(\beta)$ broadcasts to $\mathcal{A}(\beta^*)$.

\paragraph{End of the description of the branching--free computation model}\ws\ws

From these considerations we conclude that routines, constructed as before by iterated applications of the operations isoparametric recursion, composition, join and union, are still, in a suitable sense, well behaved under restrictions and output isoparametric. If only recursive routines become involved that behave well under reductions, we may also allow broadcastings at the interface of two such operations.

This remains true when we introduce, as we shall do now, in our computational model the following additional type of routine construction.

Let $\beta$ be the robust, parameterized circuit considered before, and let $R(W;Y;X)$ be a generic computation belonging to our shape list. Let $w_{\beta}$ be a precomputed vector of geometrically robust constructible functions with domain of definition $\M$ and suppose that $w_{\beta}$ and $W$ have the same vector length and that the entries of $w_{\beta}$ are the final results of an output isoparametric parameter routine applied to the circuit $\beta$. Moreover suppose that the final results of $\beta$ form a vector of the same length as $Y$.

Let $\mathcal{K}$ be the image of $w_{\beta}$. Observe that $\mathcal{K}$ is a constructible subset of the affine space which has the same dimension as the vector length of $W$. Denote by $\kappa$ the vector of the restrictions to $\mathcal{K}$ of the canonical projections of this affine space. We denote by $R(\kappa;Y;X)$ the ordinary arithmetic circuit over $\C$ obtained by substituting in the generic computation $R(W;Y;X)$ the vector of parameter variables $W$ by $\kappa$. We shall now make the following requirement: 
\textit{
\begin{enumerate}
	\item[(C)] The ordinary arithmetic circuit $R(\kappa;Y;X)$ should be consistent and robust.
\end{enumerate}
}
Observe that requirement $(C)$ is obsolete when $R(W;Y;X)$ is a totally division--free ordinary arithmetic circuit.

Suppose now that requirement $(C)$ is satisfied. A new routine, say $\mathcal{B}$, is obtained in the following way: on input $\beta$ the routine $\mathcal{B}$ joins with the generic computation $R(W;Y;X)$ at $W$ and $Y$ the previous computation of $w_{\beta}$ and the circuit $\beta$. 

From Lemma \ref{lemma intermediate results} and Corollary \ref{proposition 1} we deduce that the resulting parameterized arithmetic circuit $\mathcal{B}(\beta)$ has parameter domain $\mathcal{M}$ and is robust if it is consistent. We shall therefore require that $\mathcal{B}(\beta)$ is consistent (this condition is automatically satisfied if $R(\kappa;Y;X)$ is essentially division--free). One sees immediately that the routine $\mathcal{B}$ is well behaved under restrictions and reductions and is output isoparametric.

From now on we shall always suppose that all our recursive routines are isoparametric and that requirement $(C)$ is satisfied when we apply this last type of routine construction.

An \emph{elementary routine} of our simplified \emph{branching--free computation model} is finally obtained by the iterated application of all these construction patterns, in particular the last one, isoparametric recursion, composition, join and union. As far as only recursion becomes involved that is well behaved under reductions, we allow also broadcastings and reductions at the interface of two constructions. Of course, the identity routine belongs also to our model. The set of all these routines is therefore closed under these constructions and operations. 

We call an elementary routine \emph{essentially division--free} if it admits as input only essentially division--free, robust parameterized arithmetic circuits and all specialized generic computations used to compose it are essentially division--free. The outputs of essentially division--free elementary routines are always essentially division--free robust circuits. The set of all essentially division--free elementary routines is also closed under the mentioned constructions and operations. 

We have seen that elementary routines are, in a suitable sense, well behaved under restrictions. In the following statement we formulate explicitly the property of an elementary routine to be output isoparametric. This will be fundamental in our subsequent complexity considerations.

\begin{proposition}\label{prop: T subset and composition}
Let $\mathcal{A}$ be an elementary routine of our branching--free computation model. Then $\mathcal{A}$ is output isoparametric. More explicitly, let $\beta$ be a robust, parameterized arithmetic circuit with parameter domain $\mathcal{M}$. Suppose that $\beta$ is an admissible input for $\mathcal{A}$. Let $\theta$ be a geometrically robust, constructible map defined on $\mathcal{M}$ such that $\theta$ represents the coefficient vector of the final results of $\beta$ and let $\mathcal{T}$ be the image of $\theta$. Then $\mathcal{T}$ is a constructible subset of a suitable affine space and there exists a geometrically robust, constructible map $\sigma$ defined on $\mathcal{T}$ such that the composition map $\sigma\circ\theta$ represents the coefficient vector of the final results of $\mathcal{A}(\beta)$.
\end{proposition}

A complete proof of this proposition follows with some extra work from our
previous argumentation and will be omitted here. In case that $\mathcal{A}$ is a recursive routine, Proposition \ref{prop: T subset and composition} expresses nothing but the requirement $(B)$ applied to the output nodes of $\beta$.

Let assumptions and notations be as in Proposition \ref{prop: T subset and composition} and suppose that there is given a (not necessarily convergent) sequence $(u_k)_{k\in\N}$ of parameter instances $u_k\in\mathcal{M}$ and that there exists a (possibly unknown) parameter instance $u\in\mathcal{M}$ such that the sequence $(\theta(u_k))_{k\in\N}$ converges to $\theta(u)$. In the spirit of \cite{Alder84}, \cite{Lickteig90}, \S A and \cite{Burgisser97} the sequence of (not necessarily consistent) ordinary arithmetic circuits $(\beta^{(u_k)})_{k\in\N}$ represents an \emph{approximative algorithm} for the instantiation of the final results of $\beta$ at $u$. From Theorem \ref{theorem 2} we conclude that the constructible map $\sigma$ is strongly continuous and therefore the sequence $(\mathcal{A}(\beta)^{(u_k)})_{k\in\N}$ represents also an approximative algorithm for the instantiation of the final results of $\mathcal{A}(\beta)$ at $u$.

One sees easily that this property \emph{characterizes} output parametricity of routines which are well behaved under restrictions.

\enter
Let us observe that Proposition \ref{prop: T subset and composition} implies the following result.

\begin{corollary}\label{cor: accumulation point}
Let assumptions and notations be as in Proposition \ref{prop: T subset and composition}. Then the routine $\mathcal{A}$ is output coalescent and satisfies the following condition:
\begin{enumerate}
	\item[$(*)$] Let $u$ be an arbitrary parameter instance of $\mathcal{M}$ and let $\mathfrak{M}_u$ be the vanishing ideal of the $\C$--algebra $\C[\theta]$ at the point $\theta(u)$. Then the entries of the coefficient vector of the final results of $\mathcal{A}(\beta)$ are integral over the local $\C$--algebra $\C[\theta]_{\mathfrak{M}_u}$.
\end{enumerate}
\end{corollary}

The output coalescence of $\mathcal{A}$ and condition $(*)$ are straight--forward consequences of the output isoparametricity of $\mathcal{A}$. We remark here that condition $(*)$ follows already directly from the output coalescence of $\mathcal{A}$. This highlights again the close connection between isoparametricity and coalescence. The argument requires Zariski's Main Theorem. For details we refer to~\cite{CaGiHeMaPa03}, Sections 3.2 and 5.1. 

\subsubsection{The extended computation model}
\label{sec:Model-discussion-programs and algorithms}

We are now going to extend our simplified branching--free computation model of elementary routines by a new model consisting of \emph{algorithms} and \emph{procedures} which may contain some limited branchings. Our description of this model will be rather informal. An algorithm will be a dynamic DAG of elementary routines which will be interpreted as pipes. At the end points of the pipes, decisions may be taken which depend on testing the validity of suitable universally quantified Boolean combinations of equalities between robust constructible functions defined on the parameter domain under consideration. The output of such an \emph{equality test} is a bit vector which determines the next elementary routine (i.e., pipe) to be applied to the output circuit produced by the preceding elementary routine (pipe). This gives rise to a \emph{extended computation model} which contains branchings. These branchings depend on a limited type of decisions at the level of the underlying abstract data type, namely the mentioned equality tests. We need to include this type of branchings in our extended computation model in order to capture the whole spectrum of known elimination procedures in effective Algebraic Geometry. Because of this limitation of branchings, we shall call the algorithms of our model \emph{branching parsimonious} (compare \cite{GH01} and \cite{CaGiHeMaPa03}). A branching parsimonious algorithm $\mathcal{A}$ which accepts a robust parameterized arithmetic circuit $\beta$ with parameter domain $\mathcal{M}$ as input produces a new robust circuit $\mathcal{A}(\beta)$ with parameter domain $\mathcal{M}$. In particular $\mathcal{A}(\beta)$ \emph{does not contain any branchings}. 

Recall that our two main constructions of elementary routines depend on a previous selection of generic computations from a given shape list. This selection may be handled by calculations with the indexing of the shape list. We shall think that these calculations become realized by deterministic Turing machines. At the beginning, for a given robust parametric input circuit $\beta$ with parameter domain $\mathcal{M}$, a tuple of fixed (i.e., of $\beta$ independent) length of natural numbers is determined. This tuple constitutes an initial configuration of a Turing machine computation which determines the generic computations of our shape list that intervene in the elementary routine under construction. The entries of this tuple of natural numbers are called \emph{invariants} of the circuit $\beta$. These invariants, whose values may also be Boolean (i.e., realized by the natural numbers $0$ or $1$), depend mainly on algebraic or geometric properties of the final results of $\beta$. However, they may also depend on structural properties of the labelled DAG $\beta$.

For example, the invariants of $\beta$ may express that $\beta$ has $r$ parameters, $n$ inputs and outputs, (over $\C$) non--scalar size and depth at most $L$ and $l$, that $\beta$ is totally division--free, that the final results of $\beta$ have degree at most $d\leq 2^{l}$ and that for any parameter instance their specializations form a reduced regular sequence in $\C[X_1,\dots,X_n]$, where $X_1,\dots,X_n$ are the inputs of $\beta$.

Some of these invariants (e.g., the syntactical ones like number of parameters, inputs and outputs and non--scalar size and depth) may simply be read--off from the labelled DAG structure of $\beta$. Others, like the truth value of the statement that the specializations of the final results of $\beta$ at any parameter instance form a reduced regular sequence, have to be precomputed by an elimination algorithm from a previously given software library in effective commutative algebra or algebraic geometry or their value has to be fixed in advance as a precondition for the elementary routine which becomes applied to $\beta$.

In the same vein we may equip any elementary routine $\mathcal{A}$ with a Turing computable function which from the values of the invariants of a given input circuit $\beta$ decides whether $\beta$ is admissible for $\mathcal{A}$, and, if this is the case, determines the generic computations of our shape list which intervene in the application of $\mathcal{A}$ to $\beta$. 

We shall now go a step further letting depend the internal structure of the computation on the circuit $\beta$. In the simplest case this means that we admit that the vector of invariants of $\beta$, denoted by $\text{inv}(\beta)$, determines the architecture of a first elementary routine, say $\mathcal{A}_{\text{inv}(\beta)}$, which admits $\beta$ as input. Observe that the architectures of the elementary routines of our computation model may be characterized by tuples of fixed length of natural numbers. We consider this characterization as an \emph{indexing} of the elementary routines of our computation model. We may now use this indexing in order to combine dynamically elementary routines by composition, join and union. Let us restrict our attention to the case of composition. In this case the output circuit of one elementary routine is the input for the next routine. The elementary routines which compose this display become implemented as pipes which start with a robust input circuit and end with a robust output circuit. Given such a pipe and an input circuit $\gamma$ for the elementary routine $\mathcal{B}$ representing the pipe, we may apply suitable equality tests to the final results of $\mathcal{B}(\gamma)$ in order to determine a bit vector which we use to compute the index of the next elementary routine (seen as a new pipe) which will be applied to $\mathcal{B}(\gamma)$ as input.

A \emph{low level program} of our extended computation model is now a text, namely the transition table of a deterministic Turing machine, which computes a function $\psi$ realizing the following tasks.

Let as before $\beta$ be a robust parameterized arithmetic circuit. Then $\psi$ returns first on input $\text{inv}(\beta)$ a Boolean value, zero or one, where one is interpreted as the informal statement ``$\beta$ is an admissible input''. If this is the case, then $\psi$ returns on $\text{inv}(\beta)$ the index of an elementary routine, say $\mathcal{A}_{\text{inv}(\beta)}$, which admits $\beta$ as input. Then $\psi$ determines the equality tests which have to be realized with the final results of $\mathcal{A}_{\text{inv}(\beta)}(\beta)$. Depending on the outcome of these equality tests $\psi$ determines an index value corresponding to a new elementary routine which admits $\mathcal{A}_{\text{inv}(\beta)}(\beta)$ as input. Continuing in this way one obtains as end result an elementary routine $\mathcal{A}^{(\beta)}$, which applied to $\beta$, produces a final output circuit $\mathcal{A}^{(\beta)}(\beta)$. The function $\psi$ represents all these index computations. We denote by $\psi(\beta)$ the \emph{dynamic} vector of all data computed by $\psi$ on input $\beta$. 

The \emph{algorithm} represented by $\psi$ is the partial map between robust parametric arithmetic circuits that assigns to each admissible input $\beta$ the circuit $\mathcal{A}^{(\beta)}(\beta)$ as output. Observe that elementary routines are particular algorithms. This kind of algorithms constitute our \emph{extended computation model}. We remark that any algorithm of this model is \emph{output isoparametric}. If the pipes of an algorithm are all represented by essentially division--free elementary routines, we call the algorithm itself \emph{essentially division--free}.

One sees easily that the ``Kronecker algorithm'' \cite{GLS01} (compare also \cite{Giusti1}, \cite{Giusti2} and \cite{Giusti3}) for solving non--degenerate polynomial equation systems over the complex numbers may be programmed in our extended computation model. Observe that the Kronecker algorithm requires more than a single elementary routine for its design. In order to understand this, recall that the Kronecker algorithm accepts as input an ordinary division--free arithmetic circuit which represents by its output nodes a reduced regular sequence of polynomials $G_1,\dots,G_n$ belonging to $\C[X_1,\dots,X_n]$. In their turn, the polynomials $G_1,\dots,G_n$ determine a \emph{degree pattern}, say $\Delta:=(\delta_1,\dots,\delta_n)$, with $\delta_i:=\deg \{G_1=0,\dots,G_i=0 \}$ for $1\leq i\leq n$. 

After putting the variables $X_1,\dots,X_n$ in generic position with respect to $G_1,\dots,$ $G_n$, the algorithm performs $n$ recursive steps to eliminate them, one after the other. Finally the Kronecker algorithm produces an ordinary arithmetic circuit which computes the coefficients of $n+1$ univariate polynomials $P,V_1,\dots,V_n$ over $\C$. These polynomials constitute a ``geometric solution'' (see \cite{GLS01}) of the equation system $G_1=0,\dots,G_n=0$ because they represent the zero dimensional algebraic variety $V:=\left\{ G_1=0,\dots,G_n=0 \right\}$ in the following ``parameterized'' form:
$$V:=\left\{ (V_1(t),\dots,V_n(t));t\in\C,P(t)=0 \right\}.$$ 
Let $\beta$ be any robust, parameterized arithmetic circuit with the same number of inputs and outputs, say $X_1,\dots,X_n$ and $G_1(U,X_1,\dots,X_n),\dots,G_n(U,X_1,\dots,X_n)$, respectively. Suppose that the parameter domain of $\beta$, say $\mathcal{M}$, is irreducible and that $\text{inv}(\beta)$ expresses that for each parameter instance $u\in\mathcal{M}$ the polynomials $G_1(u,X_1,\dots,X_n),\dots,G_n(u,X_1,\dots,X_n)$ form a reduced regular sequence in $\C[X_1,\dots,X_n]$ with fixed (i.e., from $u\in\mathcal{M}$ independent) degree pattern. Suppose, furthermore, that the degrees of the individual polynomials $G_1(u,X_1,\dots,X_n)$, $\dots$, $G_n(u,X_1,\dots,X_n)$ are also fixed and that the variables $X_1,\dots,X_n$ are in generic position with respect to the varieties $\{ G_1(u,X)=0,\dots, G_i(u,X)=0 \} , 1\leq i\leq n$. Then, on input $\beta$, the Kronecker algorithm runs a certain number (which depends on $\Delta$) 
of elementary routines of our computation model which finally become combined by consistent iterative joins until the desired output is produced.

Another non--trivial example for an algorithm of our extended computation model, which involves only limited branchings, is the Gaussian elimination procedure of \cite{Edmonds67} (or \cite{Bareiss68}) applied to matrices whose entries are polynomials represented by ordinary arithmetic circuits in combination with a identity--to--zero test for such polynomials. The variables of these polynomials are considered as basic parameters and any admissible input circuit has to satisfy a certain precondition formulated as the non--vanishing of suitable minors of the given polynomial matrix. Details and applications of this type of Gaussian elimination for polynomial matrices can be found in \cite{HeintzTesis83}.

\paragraph{Procedures}\ws\ws

We say that a given algorithm $\mathcal{A}$ of our extended model \emph{computes} (only) \emph{parameters} if $\mathcal{A}$ satisfies the following condition:
\begin{quote}
\emph{for any admissible input $\beta$ the final results of $\mathcal{A}(\beta)$ are all parameters.}	
\end{quote}

Suppose that $\mathcal{A}$ is such an algorithm and $\beta$ is the robust parametric arithmetic circuit with parameter domain $\mathcal{M}$ which we have considered before. Observe that $\mathcal{A}(\beta)$ contains the input variables $X_1,\dots,X_n$ and that possibly new variables, which we call \emph{auxiliary}, become introduced during the execution of the algorithm $\mathcal{A}$ on input $\beta$. Since the algorithm $\mathcal{A}$ computes only parameters, the input and auxiliary variables become finally eliminated by the application of recursive parameter routines and evaluations. We may therefore \emph{collect garbage} in order to reduce $\mathcal{A}(\beta)$ to a \emph{final output circuit} $\mathcal{A}_{\text{final}}(\beta)$ whose intermediate results are only parameters.

If we consider the algorithm $\mathcal{A}$ as a partial map which assigns to each admissible input circuit $\beta$ its final output circuit $\mathcal{A}_{\text{final}}(\beta)$, we call $\mathcal{A}$ a \emph{procedure}.

In this case, if $\psi$ is a low level program defining $\mathcal{A}$, we call $\psi$ a \emph{low level procedure program}. 

A particular feature of our extended computation model is the following:\\
there exists a non--negative integer $f$ (depending on the recursion depth of $\mathcal{A}$) and non--decreasing real valued functions $C_f \geq 0$ ,\dots, $C_0 \geq 0$ depending on one and the same dynamic integer vector, such that with the previous notations and $L_{\beta}$, $L_{\mathcal{A}(\beta)}$ denoting the non--scalar sizes of the circuits $\beta$ and $\mathcal{A}(\beta)$ the condition
$$L_{\mathcal{A}(\beta)} \leq C_f(\psi(\beta))L^f_{\beta} +\dots+ C_0(\psi(\beta))$$
is satisfied.

In the case of the Kronecker algorithm (and most other elimination algorithms of effective Algebraic Geometry) we have $f:=1$, because the recursion depth of the basic routines which intervene is one.

In the sequel we shall need a particular variant of the notion of a procedure which enables us to capture the following situation.

Suppose we have to find a computational solution for a formally specified general algorithmic problem and that the formulation of the problem depends on certain parameter variables, say $U_1,\dots,U_r$, input variables, say $X_1,\dots,X_n$ and output variables, say $Y_1,\dots,Y_s$. Let such a problem formulation be given and suppose that its input is implemented by the robust parameterized arithmetic circuit $\beta$ considered before, interpreting the parameter variables $U_1,\dots,U_r$ as the basic parameters $\pi_1,\dots,\pi_r$. 

Then an algorithm $\mathcal{A}$ of our extended computation model which \emph{solves} the given algorithmic problem should satisfy the architectural requirement we are going to describe now.

The algorithm $\mathcal{A}$ should be the composition of two subalgorithms $\mathcal{A}^{(1)}$ and $\mathcal{A}^{(2)}$ of our computation model which satisfy on input $\beta$ the following conditions:
\textit{
\begin{enumerate}
\item[(i)] The subalgorithm $\mathcal{A}^{(1)}$ computes only parameters, $\beta$ is admissible for $\mathcal{A}^{(1)}$ and none of the indeterminates $Y_1,\dots,Y_s$ is introduced in $\mathcal{A}^{(1)}(\beta)$ as auxiliary variable (all other auxiliary variables become eliminated during the execution of the subalgorithm $\mathcal{A}^{(1)}$ on the input circuit $\beta$).
\item[(ii)] The circuit $\mathcal{A}_{\text{final}}^{(1)}(\beta)$ is an admissible input for the subalgorithm $\mathcal{A}^{(2)}$, the indeterminates $Y_1,\dots,Y_s$ occur as auxiliary variables in $\mathcal{A}^{(2)}(\mathcal{A}_{\text{final}}^{(1)}(\beta))$ and the final results of $\mathcal{A}^{(2)}(\mathcal{A}_{\text{final}}^{(1)}(\beta))$ depend only on $\pi_1,\dots,\pi_r$ and $Y_1,\dots,Y_s$.	 
\end{enumerate}
}

To the circuit $\mathcal{A}^{(2)}(\mathcal{A}_{\text{final}}^{(1)}(\beta))$ we may, as in the case when we compute only parameters, apply garbage collection. In this manner $\mathcal{A}^{(2)}(\mathcal{A}_{\text{final}}^{(1)}(\beta))$ becomes reduced to a final output circuit $\mathcal{A}_{\text{final}}(\beta)$ with parameter domain $\mathcal{M}$ which contains only the inputs $Y_1,\dots,Y_s$.

Observe that the subalgorithm $\mathcal{A}^{(1)}$ is by Proposition \ref{prop: T subset and composition} an output isoparametric procedure of our extended computation model (the same is also true for the subalgorithm $\mathcal{A}^{(2)}$, but this will not be relevant in the sequel).

We consider the algorithm $\mathcal{A}$, as well as the subalgorithms $\mathcal{A}^{(1)}$ and $\mathcal{A}^{(2)}$, as \emph{procedures} of our extended computation model. In case that the \emph{subprocedures} $\mathcal{A}^{(1)}$ and $\mathcal{A}^{(2)}$ are essentially division--free, we call also the procedure $\mathcal{A}$ \emph{essentially division--free}. This will be of importance in Section \ref{sec:Model-Applications}. 

The architectural requirement given by conditions $(i)$ and $(ii)$ may be interpreted as follows:\\
the subprocedure $\mathcal{A}^{(1)}$ is a pipeline which transmits only parameters to the subprocedure $\mathcal{A}^{(2)}$. In particular, no (true) polynomial is transmitted from $\mathcal{A}^{(1)}$ to $\mathcal{A}^{(2)}$. 	

Nevertheless, let us observe that on input $\beta$ the procedure $\mathcal{A}$ establishes by means of the underlying low level program $\psi$ an additional link between $\beta$ and the subprocedure $\mathcal{A}^{(2)}$ applied to the input $\mathcal{A}^{(1)}(\beta)$. The elementary routines which constitute $\mathcal{A}^{(2)}$ on input $\mathcal{A}^{(1)}(\beta)$ become determined by index computations which realizes $\psi$ on $\text{inv}(\beta)$ and certain equality tests between the intermediate results of $\mathcal{A}^{(1)}(\beta)$. In this sense the subprocedure $\mathcal{A}^{(1)}$ transmits not only parameters to the subprocedure but also a limited amount of digital information which stems from the input circuit $\beta$.   	 

The decomposition of the procedure $\mathcal{A}$ into two subprocedures $\mathcal{A}^{(1)}$ and $\mathcal{A}^{(2)}$ satisfying conditions $(i)$ and $(ii)$ represents an architectural restriction which is justified when it makes sense to require that on input $\beta$ the number of essential additions and multiplications contained in $\mathcal{A}_{\text{final}}(\beta)$ is bounded by a function which depends only on $\text{inv}(\beta)$. In Section \ref{ 6.5.1 flat}, we shall make a substantial use of this restriction and give such a justification in the particular case of elimination algorithms.  

Here, we shall only point out the following consequence of this restriction. Let assumptions and notations be as before, let $G,\nu$ and $F$ be vectors composed by the final results of $\beta$, $\mathcal{A}^{(1)}(\beta)$ and $\mathcal{A}_{\text{final}}(\beta)$, respectively, and let $\theta$ and $\varphi$ be the coefficient vectors of $G$ and $F$. Then the images of $\theta$ and $\nu$ are constructible subsets $\mathcal{T}$ and $\mathcal{T}'$ of suitable affine spaces and there exist geometrically robust constructible maps $\sigma$ and $\sigma'$ defined on $\mathcal{T}$ and $\mathcal{T}'$ with $\nu=\sigma\circ\theta$ and $\varphi=\sigma'\circ\nu=\sigma'\circ\sigma\circ\theta$.

Based on \cite{HeiKu04} and \cite{HeiKu07}, we shall develop in future work a high level specification language for algorithms and procedures of our computation model. The idea is to use a generalized variant of the \emph{extended constraint data base model} introduced in \cite{HeiKu04} in order to specify algorithmic problems in symbolic Scientific Computing, especially in effective Algebraic Geometry
. In this sense the procedure $\mathcal{A}$, which solves the algorithmic problem considered before, will turn out to be a \emph{query computation} composed by two subprocedures namely $\mathcal{A}^{(1)}$ and $\mathcal{A}^{(2)}$ which compute each a subquery of the query which specifies the given algorithmic problem. All these queries are called \emph{geometric} because the procedures $\mathcal{A}^{(1)}$, $\mathcal{A}^{(2)}$ and $\mathcal{A}$ are output isoparametric (see \cite{HeiKu07}).

\section[Applications of the extended computation model to complexity issues of effective elimination theory]{\large{Applications of the extended computation model to\\ complexity issues of effective elimination theory}}\label{sec:Model-Applications}
In this section we shall always work with procedures of our extended, branching parsimonious computation model. We shall study representative examples of elimination problems in effective Algebraic Geometry which certify, to a different extent, that \emph{branching parsimonious} elimination procedures based on our computation paradigm \emph{cannot run in polynomial time}. 

\subsection{Flat families of zero--dimensional elimination problems}
\label{ 6.5.1 flat}
We first introduce, in terms of abstract data types, the notion of a flat family of zero--dimensional elimination problems (see also \cite{GH01} and \cite{CaGiHeMaPa03}). Then we fix the classes of (concrete) objects, namely robust parameterized arithmetic circuits with suitable parameter domains, which represent (``implement'') these problems by means of a suitable abstraction function.

Throughout this section, we suppose that there are given indeterminates $U_1,\dots,$ $U_r$, $X_1,\dots,X_n$ and $Y$ over $\C$.

As concrete objects we shall consider robust parameterized arithmetic input and output circuits with parameter domain $\A^r$. The indeterminates $U_1,\dots,U_r$ will play the role of the basic parameters. The input nodes of the input circuits will be labelled by $X_1,\dots,X_n$, whereas the output circuits will have a single input node, labelled by $Y$. 

Let us now define the meaning of the term ``flat family of zero--dimensional elimination problems'' (in the basic parameters $U_1,\dots,U_r$ and the inputs $X_1,\dots,X_n$). Let $U:=(U_1,\dots,U_r)$ and $X:=(X_1,\dots,X_n)$ and let $G_1,\dots,G_n$ and $H$ be polynomials belonging to the $\C$--algebra $\C[U,X]:=\C[U_1,\dots,U_r,X_1,\dots,X_n]$. Suppose that the polynomials $G_1,\dots,G_n$ form a regular sequence in $\C[U,X]$, thus defining an equidimensional subvariety $V:=\left\{ G_1=0,\dots,G_n=0 \right\}$ of the $(n+r)$--dimensional affine space $\A^r\times\A^n$. The algebraic variety $V$ has dimension $r$. Let $\delta$ be the (geometric) degree of $V$ (observe that this degree does not take into account multiplicities or components at infinity). Suppose, furthermore, that the morphism of affine varieties $\pi:V\to\A^r$, induced by the canonical projection of $\A^r\times\A^n$ onto $\A^r$, is finite and generically unramified (this implies that $\pi$ is flat and that the ideal generated by $G_1,\dots,G_n$ in $\C[U,X]$ is radical). Let $\tilde{\pi}:V\to\A^{r+1}$ be the morphism defined by $\tilde{\pi}(v):=(\pi(v),H(v))$ for any point $v$ of the variety $V$. The image of $\tilde{\pi}$ is a hypersurface of $\A^{r+1}$ whose minimal equation is a polynomial of $\C[U,Y]:=\C[U_1,\dots,U_r,Y]$ which we denote by $F$. Let us write $\deg F$ for the total degree of the polynomial $F$ and $\deg_Y F$ for its partial degree in the variable $Y$. Observe that $F$ is monic in $Y$ and that $\deg F\leq\delta \deg H$ holds. Furthermore, for a Zariski dense set of points $u$ of $\A^r$, we have that $\deg_Y F$ is the cardinality of the image of the restriction of $H$ to the finite set $\pi^{-1}(u)$. The polynomial $F(U,H)$ vanishes on the variety $V$.

Let us consider an arbitrary point $u:=(u_1,\dots,u_r)$ of $\A^r$. For given polynomials $A\in\C[U,X]$ and $B\in\C[U,Y]$, we denote by $A^{(u)}$ and $B^{(u)}$ the polynomials $A(u_1,\dots,u_r,X_1,\dots,X_n)$ and $B(u_1,\dots,u_r,Y)$ which belong to $\C[X]:=\C[X_1,\dots,X_n]$ and $\C[Y]$ respectively. Similarly we denote for an arbitrary polynomial $C\in\C[U]$ by $C^{(u)}$ the value $C(u_1,\dots,u_r)$ which belongs to the field $\C$. The polynomials $G_1^{(u)},\dots,G_n^{(u)}$ define the zero--dimensional subvariety
$$V^{(u)}:=\left\{ G_1^{(u)}=0,\dots,G_n^{(u)}=0 \right\} \cong \pi^{-1}(u)$$
of the affine space $\A^n$. The degree (i.e., the cardinality) of $V^{(u)}$ is bounded by $\delta$. Denote by $\tilde{\pi}^{(u)}:V^{(u)}\to\A^1$ the morphism induced by the polynomial $H^{(u)}$ on the variety $V^{(u)}$. Observe that the polynomial $F^{(u)}$ vanishes on the (finite) image of the morphism $\tilde{\pi}^{(u)}$. Observe also that the polynomial $F^{(u)}$ is not necessarily the minimal equation of the image of $\tilde{\pi}^{(u)}$.

We call the equation system $G_1=0,\dots,G_n=0$ and the polynomial $H$ a \emph{flat family of zero--dimensional elimination problems depending on the basic parameters $U_1,\dots,U_r$ and the inputs $X_1,\dots,X_n$} and we call $F$ the associated \emph{elimination polynomial}. A point $u\in\A^r$ is considered as a \emph{parameter instance} which determines a \emph{particular problem instance}, consisting of the equations $G_1^{(u)}=0,\dots,G_n^{(u)}=0$ and the polynomial $H^{(u)}$. A power of the polynomial $F^{(u)}$ is called a \emph{solution} of this particular problem instance.

We suppose now that the given flat family of elimination problems is
implemented by an essentially division-free, robust parameterized
arithmetic circuit $\beta$ with parameter domain $\A^r$ and inputs $X_1,\dots,X_n$, whose final results are the polynomials $G_1,\dots,G_n$ and $H$. The task is to find another essentially division--free, robust parameterized arithmetic circuit $\gamma$ with parameter domain $\A^r$ having a single output node, labelled by $Y$, which computes for a suitable integer $q\in\N$ the power $F^q$ of the associated elimination polynomial $F$. We suppose, furthermore, that this goal is achieved by the application of an essentially division--free procedure $\mathcal{A}$ of our extended computation model to the input circuit $\beta$. Thus we may put $\gamma:=\mathcal{A}_{\text{final}}(\beta)$ and $\gamma$ may be interpreted as an essentially division--free circuit over $\C[U]$ with a single input $Y$ (observe that the parameters computed by the robust circuits $\beta$, $\mathcal{A}(\beta)$ and $\mathcal{A}_{\text{final}}(\beta)$ are geometrically robust constructible functions with domain of definition $\A^r$ which belong by \cite{GHMS09}, Corollary 12 to the $\C$--algebra $\C[U]$). Using the geometric properties of flat families of zero--dimensional problems, we deduce from \cite{Giusti1}, \cite{Giusti2},\cite{Giusti3}, \cite{GLS01} or alternatively from \cite{Caniglia89}, \cite{Dickens} that such essentially division--free procedures always exist and that they compute even the elimination polynomial $F$ (the reader may notice that one needs for this argument the full expressivity of our computation model which includes divisions by parameters).

We say that the essentially division--free procedure $\mathcal{A}$ \emph{solves algorithmically} the \emph{general instance} of the given flat family of zero--dimensional elimination problems if $A$ computes $F$ or a power of it. 

From now on we suppose that there is given a procedure $\mathcal{A}$ of our extended computation model, decomposed in two essentially division--free subprocedures $\mathcal{A}^{(1)}$ and $\mathcal{A}^{(2)}$ as in Section \ref{sec:Model-discussion-programs and algorithms}, such that $\mathcal{A}$ solves algorithmically the general instance of any given flat family of zero--dimensional elimination problems. Our circuit $\beta$ is therefore an admissible input for $\mathcal{A}$ and hence for $\mathcal{A}^{(1)}$. The final results of $\mathcal{A}^{(1)}(\beta)$ constitute a geometrically robust constructible map $\nu$ defined on $\A^r$ which represents by means of $\mathcal{A}^{(1)}_{\text{final}}(\beta)$ an admissible input for the procedure $\mathcal{A}^{(2)}$. Moreover, $\gamma:=\mathcal{A}_{\text{final}}(\beta)$ is an essentially division--free parameterized arithmetic circuit with parameter domain $\A^r$ and input $Y$.

Let $\mathcal{S}$ be the image of the geometrically robust constructible map $\nu$. Then $\mathcal{S}$ is an irreducible constructible subset of a suitable affine space. Analyzing now the internal structure of the essentially division--free, robust parameterized arithmetic circuit $\mathcal{A}^{(2)}(\mathcal{A}^{(1)}(\beta))$, one sees easily that there exists a geometrically robust constructible map $\psi$ defined on $\mathcal{S}$ such that the entries of the geometrically robust composition map $\nu^*:=\psi\circ\nu$ constitute the essential parameters of the circuit $\gamma$. Let $m$ be the number of components of the map $\nu^*$. Since $\nu$ and $\nu^*$ are composed by geometrically robust constructible functions defined on $\A^r$, we deduce from \cite{GHMS09}, Corollary 12 that $\nu$ and $\nu^*$ may be interpreted as vectors of polynomials of $\C[U]$.

The circuit $\gamma$ is essentially division--free. Hence there exists a vector $\omega$ of $m$--variate polynomials over $\C$ such that the polynomials of $\C[U]$, which constitute the entries of $\omega(\nu^*)$, become the coefficients of the elimination polynomial $F$ with respect to the main indeterminate $Y$ (see \cite{Krick96}, Section 2.1). Observe that we may write $\omega(\nu^*)=\omega\circ\nu^*$ interpreting the entries of $\nu^*$ as polynomials of $\C[U]$.

We are now going to see what happens at a particular parameter instance $u\in\A^r$. Since $\beta$, $\mathcal{A}^{(1)}(\beta)$, $\mathcal{A}(\beta)$ and $\gamma=\mathcal{A}_{\text{final}}(\beta)$ are essentially division--free, robust parameterized arithmetic circuits with parameter domain $\A^r$, we may specialize the vector $U$ of basic parameters to the parameter instance $u\in\A^r$, obtaining thus ordinary division--free arithmetic circuits over $\C$ with the same inputs. We denote them by the superscript $u$, namely by $\beta^{(u)}$, $(\mathcal{A}^{(1)}(\beta))^{(u)}$, $(\mathcal{A}(\beta))^{(u)}$ and $\gamma^{(u)}$. One sees immediately that $G_1^{(u)},\dots,G_n^{(u)}$ and $H^{(u)}$ are the final results of $\beta^{(u)}$, that the entries of $\nu(u)$ are the final results of $(\mathcal{A}^{(1)}(\beta))^{(u)}$ and that $(F^{(u)})^q$ is the final result of $\mathcal{A}(\beta)^{(u)}$ and $\gamma^{(u)}$. Observe that the division--free circuit $\gamma^{(u)}$ uses only the entries of $\nu^*(u)$ and fixed rational numbers as scalars. 

In the same spirit as before, we say that the procedure $\mathcal{A}$ solves algorithmically the particular instance, which is determined by $u$, of the given flat family of zero--dimensional elimination problems. 

Let us here clarify how all this is linked to the rest of the terminology used in~\cite{CaGiHeMaPa03}. In this terminology the polynomial map given by $\omega$ defines a ``holomorphic encoding'' of the set of solutions of all particular problem instances and $\nu^*(u)$ is a ``code'' of the particular solution $(F^{(u)})^q$. In the same context the robust constructible map $\nu^*$ is called an ``elimination procedure'' which is ``robust'' since the procedure $\mathcal{A}^{(1)}$ is output isoparametric and since $\nu^*$ is geometrically robust (compare \cite{CaGiHeMaPa03}, Definition 5, taking into account Lemma \ref{lemma intermediate results}, Proposition \ref{prop: T subset and composition} and Corollary \ref{cor: accumulation point} above).

In this sense, we speak about \emph{families} of zero--dimensional elimination problems and their instances and not simply about a single (particular or general) zero--dimensional elimination problem.

Let us now turn back to the discussion of the given essentially division--free procedure $\mathcal{A}$ which solves algorithmically the general instance of any flat family of zero--dimensional elimination problems.

We are now going to show the main result of this paper, namely that the given procedure $\mathcal{A}$ \emph{cannot run in polynomial time}.

\begin{theorem}\label{def: main theorem model}
Let notations and assumptions be as before. For any natural number $n$ there exists an essentially division--free, robust parameterized arithmetic circuit $\beta_n$ with basic parameters $T$, $U_1,\dots,U_n$ and inputs $X_1,\dots,X_n$ which for $U:=(U_1,\dots,U_n)$ and $X:=(X_1,\dots,X_n)$ computes polynomials $G_1^{(n)},\dots,G_n^{(n)} \in\C[X]$ and $H^{(n)}\in\C[T,U,X]$ such that the following conditions are satisfied:
\begin{enumerate}
	\item[(i)] The equation system $G_1^{(n)}=0,\dots,G_n^{(n)}=0$ and the polynomial $H^{(n)}$ constitute a flat family of zero--dimensional elimination problems, depending on the parameters $T$, $U_1,\dots,U_n$ and the inputs $X_1,\dots,X_n$, with associated elimination polynomial $F^{(n)}\in\C[T,U,Y]$.
	\item[(ii)] $\beta_n$ is an ordinary division--free arithmetic circuit of size $O(n)$ over $\C$ with inputs $T$, $U_1,\dots,U_n$, $X_1,\dots,X_n$.
	\item[(iii)] $\gamma_n:=\mathcal{A}_{\text{final}}(\beta_n)$ is an essentially division--free, robust parameterized arithmetic circuit with basic parameters $T,U_1,\dots,U_n$ and input $Y$ such that $\gamma_n$ computes for a suitable integer $q_n \in\N$ the polynomial $(F^{(n)})^{q_n}$. The circuit $\gamma_n$ performs at least $\Omega(2^{\frac{n}{2}})$ essential multiplications and at least $\Omega(2^{n})$ multiplications with parameters. Therefore $\gamma_n$ has, as ordinary arithmetic circuit over $\C$ with inputs $T,U_1,\dots,U_n$ and $Y$, non--scalar size at least $\Omega(2^{n})$. 
\end{enumerate}
\end{theorem}

\begin{proof}
During our argumentation we shall tacitly adapt to the new context the notations introduced before. We shall follow the main technical ideas behind the papers \cite{GH01}, \cite{CaGiHeMaPa03} and~\cite{GHMS09}. We fix now the natural number $n$ and consider the polynomials
$$G_1:=G_1^{(n)}:=X_1^2-X_1,\dots,G_n:=G_n^{(n)}:=X_1^2-X_n$$
and 
$$H:=H^{(n)}:=\sum_{1\leq i\leq n} 2^{i-1}X_i + T\prod_{1\leq i\leq n} (1+(U_i-1)X_i)$$
which belong to $\C[X]$ and to $\C[T,U,X]$, respectively.

Observe that $G_1,\dots,G_n$ and $H$ may be evaluated by a division--free ordinary arithmetic circuit $\beta:=\beta_n$ over $\C$ which has non--scalar size $O(n)$ and inputs $T$, $U_1,\dots,U_n$, $X_1,\dots,X_n$. As parameterized arithmetic circuit $\beta$ is therefore robust. Hence $\beta$ satisfies condition $(ii)$ of the theorem.

One sees easily that $G_1=0,\dots,G_n=0$ and $H$ constitute a flat family of zero--dimensional elimination problems depending on the parameters $T$, $U_1,\dots,U_n$ and the inputs $X_1,\dots,X_n$.

Let us write $H$ as a polynomial in the main indeterminates $X_1,\dots,X_n$ with coefficients $\theta_{\kappa_1,\dots,\kappa_n}\in\C[T,U]$, $\kappa_1,\dots,\kappa_n\in\left\{0,1\right\}$, namely 
$$H=\sum_{\kappa_1,\dots,\kappa_n\in\left\{0,1\right\}} \theta_{\kappa_1,\dots,\kappa_n} X_1^{\kappa_1},\dots,X_n^{\kappa_n}.$$
Observe that for $\kappa_1,\dots,\kappa_n\in\left\{0,1\right\}$ the polynomial $\theta_{\kappa_1,\dots,\kappa_n}(0,U)\in \C[U]$ is of degree at most zero, i.e., a constant complex number, independent of $U_1,\dots,U_n$.

Let $\theta:=(\theta_{\kappa_1,\dots,\kappa_n})_{\kappa_1,\dots,\kappa_n\in\left\{0,1\right\}}$ and observe that the vector $\theta(0,U)$ is a fixed point of the affine space $\A^{2^n}$. We denote by $\mathfrak{M}$ the vanishing ideal of the $\C$--algebra $\C[\theta]$ at this point.

Consider now the polynomial
$$F:=F^{(n)}:= \prod_{0\leq j\leq 2^{n}-1} (Y-(j+T \prod_{1\leq i\leq n}U_i^{\left[j\right]_i}))$$ 
of $\C[T,U,Y]$, where $[j]_i$ denotes the $i$--th digit of the binary representation of the integer $j$, $0\leq j\leq 2^n-1$, $1\leq i\leq n$. Let $q:=q_n$. 

From the identity $ \prod_{\epsilon\in\{ 0,1 \}^n} (Y-H(T,U,\epsilon) )= \prod_{0 \leq j\leq 2^n-1} (Y-(j+T \prod_{1\leq i\leq n}U_i^{\left[j\right]_i}))$ one deduces that $F$ is the elimination polynomial associated with the given flat family of zero--dimensional elimination problems $G_1=0,\dots,G_n=0$ and $H$.

Let us write $F^q$ as a polynomial in the main indeterminate $Y$ with coefficients $\varphi_\kappa\in\C[T,U]$, $1\leq \kappa\leq 2^n q$, namely
$$F^q= Y^{2^n q} + \varphi_1 Y^{2^n q -1} +\dots+ \varphi_{2^n q}.$$
Observe that for $1\leq \kappa\leq 2^n q$ the polynomial $\varphi_\kappa(0,U)\in\C[U]$ is of degree at most zero. Let $\lambda_\kappa:=\varphi_\kappa(0,U)$, $\lambda:=(\lambda_\kappa)_{1\leq \kappa\leq 2^n q}$ and $\varphi:=(\varphi_\kappa)_{1\leq \kappa\leq 2^n q}$. Observe that $\lambda$ is a fixed point of the affine space $\A^{2^n q}$.

Recall that $\beta$ is an admissible input for the procedure $\mathcal{A}$ and hence for $\mathcal{A}^{(1)}$, that the final results of $\mathcal{A}^{(1)}(\beta)$ constitute the entries of the robust constructible map $\nu$ defined on $\A^{n+1}$, that $\nu$ represents (by means of the circuit $\mathcal{A}^{(1)}_{\text{final}}(\beta)$) an admissible input for the procedure $\mathcal{A}^{(2)}$ and that $\gamma=\mathcal{A}_{\text{final}}(\beta)$ is an essentially division--free, parameterized arithmetic circuit with parameter domain $\A^{n+1}$ and input $Y$.

Furthermore, recall that there exists a geometrically robust constructible map $\psi$ defined on the image $\mathcal{S}$ of $\nu$ such that the entries of $\nu^*=\psi\circ\nu$ constitute the essential parameters of the circuit $\gamma$, that the entries of $\nu$ and $\nu^*$ may be interpreted as polynomials of $\C[T,U]$ and that for $m$ being the number of components of the map $\nu^*$, there exists a vector $\omega$ of $m$--variate polynomials over $\C$ such that the polynomials of $\C[T,U]$ which constitute the entries of $\omega(\nu^*)=\omega\circ\nu^*$ become the coefficients of the polynomial $F^q$ with respect to the main indeterminate $Y$. Let $\mathcal{T}$ be the image of the coefficient vector $\theta$ of $H$, and interpret $\theta$ as a geometrically robust constructible map defined on $\A^{n+1}$. Observe that $\mathcal{T}$ is a constructible subset of $\A^{2^n}$. Since $H$ is the unique final result of the circuit $\beta$ which depends on parameters, we deduce from Proposition \ref{prop: T subset and composition} that there exists a geometrically robust constructible map $\sigma$ defined on $\mathcal{T}$ satisfying the condition $\nu=\sigma\circ\theta$. This implies $\nu^*=\psi\circ\sigma\circ\theta$ and, following Definition \ref{def: geometrically robust map} $(i)$ and \cite{GHMS09}, Corollary 12, that the entries of $\nu^*$ are polynomials of $\C[T,U]$ which are integral over the local $\C$--subalgebra $\C[\theta]_{\mathfrak{M}}$ of $\C(T,U)$.

Let $\mu\in\C[T,U]$ be such an entry. Then there exists an integer $s$ and polynomials $a_0,a_1,\dots,a_s \in\C[\theta]$ with $a_0\notin\mathfrak{M}$ such that the algebraic dependence relation
\begin{equation}
\label{(*)}
a_0\mu^s + a_1\mu^{s-1}+\dots+ a_s =0	
\end{equation}
is satisfied in $\C[T,U]$. From \eqref{(*)} we deduce the algebraic dependence relation
\begin{equation}
\label{(**)}
a_0(0,U)\mu(0,U)^s +a_1(0,U)\mu(0,U)^{s-1} +\dots+a_s(0,U) =0	
\end{equation}
in $\C[U]$.

Since the polynomials $a_0,a_1,\dots,a_s$ belong to $\C[\theta]$ and $\theta(0,U)$ is a fixed point of $\A^{2^n}$, we conclude that $\alpha_0:=a_0(0,U),\alpha_1 := a_1(0,U), \dots , \alpha_s:= a_s(0,U)$ are complex numbers. Moreover, $a_0\notin\mathfrak{M}$ implies $\alpha_0\neq 0$. 

Thus \eqref{(**)} may be rewritten into the algebraic dependence relation
\begin{equation}
\label{(***)}
\alpha_0\mu(0,U)^s +\alpha_1 \mu(0,U)^{s-1} +\dots+ \alpha_s =0
\end{equation}
in $\C[U]$ with $\alpha_0\neq 0$.

This implies that the polynomial $\mu(0,U)$ of $\C[U]$ is of degree at most zero. Therefore $w:=\nu^*(0,U)$ is a fixed point of the affine space $\A^m$. 

Since $\gamma$ computes the polynomial $F^q$ and $F^q$ has the form $F^q = Y^{2^n q}+\varphi_1Y^{2^n q-1} +\dots+ \varphi_{2^n q}$ with $\varphi_\kappa\in\C[T,U]$, $1\leq \kappa\leq 2^n q$, we see that $\varphi=(\varphi_\kappa)_{1\leq \kappa\leq 2^n q}$ may be decomposed as follows:
$$\varphi=\omega(\nu^*)=\omega\circ\nu^*.$$
Recall that $\lambda=(\lambda_\kappa)_{1\leq \kappa\leq 2^n q}$ with $\lambda_\kappa:=\varphi_\kappa(0,U)$, $1\leq \kappa\leq 2^n q$, is a fixed point of the affine space $\A^{2^n}$.

For $1\leq \kappa\leq 2^n q$ we may write the polynomial $\varphi_\kappa\in\C[T,U]$ as follows:
\begin{equation}
\label{(****)}
\varphi_\kappa=\lambda_\kappa + \Delta_\kappa T + \text{\ terms\ of\ higher\ degree\ in\ } T	
\end{equation}
with $\Delta_\kappa\in\C[U]$. From \cite{CaGiHeMaPa03}, Lemma 6 we deduce that the elimination polynomial $F$ has the form $F = Y^{2^n}+B_1Y^{2^n-1}+\dots+B_{2^n}$, where for $1\leq l\leq 2^n$ the coefficient $B_l$ is an element of $\C[T,U]$ of the form
$$B_l = (-1)^l \sum_{l\leq j_1<\dots <j_l< 2^n} j_1\dots j_l +TL_l+\text{\ terms\ of\ higher\ degree\ in\ }T,$$
where $L_1,\dots,L_{2^n} \in \C[U]$ are $\C$--linearly independent.

Choose now different complex numbers $\eta_1,\dots,\eta_{2^n}$ from $\C - \{ j\in\Z ; 0\leq j<2^n \}$ and observe that for $1\leq \kappa'\leq 2^n$ the identities
$$\frac{\partial F^q}{\partial T} (0,U,\eta_{\kappa '}) = 
q F^{q-1}(0,U,\eta_{\kappa '}) \frac{\partial F}{\partial T}(0,U,\eta_{\kappa '})=
q \prod_{0 \leq j< 2^n}(\eta_{\kappa '}-j)^{q-1} \sum_{1 \leq l\leq 2^n} L_l \eta_{\kappa '}^{2^n-l}$$
and 
$$\frac{\partial F^q}{\partial T} (0,U,\eta_{\kappa '}) =  \sum_{1\leq \kappa \leq 2^n q} \Delta_\kappa \eta_{\kappa '}^{2^n q - \kappa}$$
hold.

Since $L_1,\dots,L_{2^n}$ are $\C$--linearly independent, we deduce from the non--singularity of the Vandermonde matrix $(\eta_{\kappa '}^{2^n - l})_{1\leq l,\kappa'\leq 2^n}$ that $2^n$ many of the polynomials $\Delta_1,\dots,\allowbreak \Delta_{2^n q}$ of $\C[U]$ are $\C$--linearly independent.

Consider now an arbitrary point $u\in\A^n$ and let $\epsilon_u:\A^1\to\A^m$ and $\delta_u:\A^1\to\A^{2^n q}$ be the polynomial maps defined for $t\in\A^1$ by $\epsilon_u(t):=\nu^*(t,u)$ and $\delta_u(t):=\varphi(t,u)$. Then we have $\epsilon_u(0)=\nu^*(0,u)=w$ and $\delta_u(0)=\varphi(0,u)=\lambda$, independently of $u$. Moreover, from $\varphi=\omega\circ\nu^*$ we deduce $\delta_u=\omega\circ\epsilon_u$.

Thus \eqref{(****)} implies
\begin{equation}
\label{(+)}
(\Delta_1(u),\dots,\Delta_{2^n q}(u))=\frac{\partial\varphi}{\partial t}(0,u)=\delta_u'(0)=(D\omega)_w(\epsilon_u'(0)),	
\end{equation}
where $(D \omega)_w$ denotes the (first) derivative of the $m$--variate polynomial map $\omega$ at the point $w\in\A^m$ and $\delta_u'(0)$ and $\epsilon_u'(0)$ are the derivatives of the parameterized curves $\delta_u$ and $\epsilon_u$ at the point $0\in\A^1$. We rewrite now \eqref{(+)} in matrix form, replacing $(D\omega)_w$ by the corresponding transposed Jacobi matrix $M\in\A^{m\times 2^n q}$ and $\delta_u'(0)$ and $\epsilon_u'(0)$ by the corresponding points of $\A^{2^n q}$ and $\A^m$, respectively. 

Then \eqref{(+)} takes the form
\begin{equation}
\label{(++)}
(\Delta_1(u),\dots,\Delta_{2^n q}(u))= \epsilon_u'(0)M,	
\end{equation}
where the complex $(m\times 2^n q)$--matrix $M$ is independent of $u$.

Since $2^n$ many of the polynomials $\Delta_1,\dots,\Delta_{2^n}\in\C[U]$ are $\C$--linearly independent, we may choose points $u_1,\dots,u_{2^n}\in\A^n$ such that the complex $(2^n \times 2^n q)$--matrix  
$$N:=(\Delta_{\kappa}(u_l))_{\stackrel{\text{\scriptsize{$1\leq l\leq 2^n$}}}{\text{\scriptsize{$1\leq \kappa \leq 2^n q$}}} }$$ 
has rank $2^n$.

Let $K$ be the complex $(2^n\times m)$--matrix whose rows are $\epsilon_{u_1}'(0),\dots,\epsilon_{u_{2^n}}'(0)$. 

Then \eqref{(++)} implies the matrix identity 
$$N=K\cdot M.$$
Since $N$ has rank $2^n$, the rank of the complex $(m\times 2^n)$--matrix $M$ is at least $2^n$. This implies
\begin{equation}
\label{(+++)}
m\geq 2^n.	
\end{equation}
Therefore the circuit $\gamma$ contains $m\geq 2^n$ essential parameters.

Let $L$ be the number of essential multiplications executed by the parameterized arithmetic circuit $\gamma$ and let $L'$ be the total number of multiplications of $\gamma$, excepting those by scalars from $\C$. Then, after a well--known standard rearrangement \cite{PS73} of $\gamma$, we may suppose without loss of generality, that there exists a constant $c>0$ (independent of the input circuit $\gamma$ and the procedure $\mathcal{A}$) such that $L\geq cm^{\frac{1}{2}}$ and $L'\geq cm$ holds.

From the estimation \eqref{(+++)} we deduce now that the circuit $\gamma$ performs at least $\Omega(2^{\frac{n}{2}})$ essential multiplications and at least $\Omega(2^n)$ multiplications, including also multiplications with parameters. This finishes the proof of the theorem.
\end{proof}

\paragraph*{Observations} Let assumptions and notations be as before. In the proof of Theorem \ref{def: main theorem model} we made a substantial use of the output isoparametricity of the procedure $\mathcal{A}^{(1)}$ when we applied Proposition \ref{prop: T subset and composition} in order to guarantee the existence of a geometrically robust constructible map $\sigma$ defined on $\mathcal{T}$ which satisfies the condition $\nu=\sigma\circ\theta$. The conclusion was that the entries of $\nu^*=\psi\circ \nu$ are polynomials of $\C[T,U]$ which are integral over $\C[\theta]_{\mathfrak{M}}$. This implied finally that $\nu^*(0,U)$ is a fixed point of the affine space $\A^m$. Taking into account the results of \cite{CaGiHeMaPa03}, Sections 3.2 and 5.1 it suffices to require that the procedure $\mathcal{A}^{(1)}$ is \emph{output coalescent} in order to arrive to the same conclusion. This means that Theorem \ref{def: main theorem model} remains valid if we require only that the procedure $\mathcal{A}^{(1)}$ is output coalescent.

\enter
In the proof of Theorem \ref{def: main theorem model} we have exhibited an infinite sequence of flat families of zero--dimensional elimination problems represented by robust parameterized arithmetic circuits of small size, such that any implementation of their associated elimination polynomials, obtained by a procedure of our extended computation model which solves the given elimination task for any instance, requires circuits of exponential size.

The statement of Theorem \ref{def: main theorem model} may also be interpreted in terms of a mathematically certified trade--off of quality attributes. Suppose for the moment that we had built our model for branching parsimonious computation in the same way as in Section \ref{sec:Model-discussion}, omitting the requirement of isoparametricity for recursive routines, however. Recall that this requirement implies the output isoparametricity of any algorithm of our extended computation model and recall from Section \ref{sec:Model-discussion-simplified} that well behavedness under reduction is a quality attribute which implies output isoparametricity and therefore also the conclusion of Theorem \ref{def: main theorem model}. 

A complexity class like ``exponential time in worst case'' represents also a quality attribute. Thus we see that the quality attribute ``well behavedness under reduction'' implies the quality attribute ``exponential time in worst case'' for any essentially division--free procedure of our extended computation model which solves algorithmically the general instance of any given flat family of zero--dimensional problems.

\enter
The proof of Theorem \ref{def: main theorem model} depends substantially on the decomposition of the elimination procedure $\mathcal{A}$ into two subprocedures $\mathcal{A}^{(1)}$ and $\mathcal{A}^{(2)}$ satisfying conditions $(i)$ and $(ii)$ of Section \ref{sec:Model-discussion-programs and algorithms}. We are now going to justify this architectural restriction on the procedure $\mathcal{A}$ for the particular case of elimination algorithms. 

As at the beginning of this section, let $U:=(U_1,\dots,U_r)$, $X:=(X_1,\dots,X_n)$, $G_1,\dots,G_n$, $H\in\C[U,X]$ and $F\in\C[U,Y]$ such that $G_1=0,\dots,G_n=0$ and $H$ constitute a flat family of zero--dimensional elimination problems and $F$ its associated elimination polynomial. Suppose that $G_1,\dots,G_n$ and $H$ are implemented by an essentially division--free, robust parameterized arithmetic circuit $\beta$ with parameter domain $\A^r$ and inputs $X_1,\dots,X_n$.

All \emph{known} algorithms which solve the general instance of any flat family of zero--dimensional elimination problems may be interpreted as belonging to our restricted set of procedures. They compute directly the elimination polynomial $F$ (and not an arbitrary power of it). Thus let $\mathcal{A}$ be such a known algorithm and let $\mathcal{A}^{(1)}$ and $\mathcal{A}^{(2)}$ be the subalgorithms which compose $\mathcal{A}$ in the same way as before. Then $\mathcal{A}^{(1)}$ computes the coefficients of $F$, where $F$ is considered as a polynomial over $\C[U]$ in the indeterminate $Y$. The subalgorithm $\mathcal{A}^{(2)}$ may be interpreted as the Horner scheme which evaluates $F$ from its precomputed coefficients and $Y$.

Observe that $F$, and hence $\deg_Y F$, depends only on the polynomials $G_1,\dots,G_n$ and $H$, but not on the particular circuit $\beta$. Therefore $\deg_Y F$ is determined by $\psi(\beta)$, where $\psi$ is the low level program of the algorithm $\mathcal{A}$.\\
For any parameter instance $u\in\A^r$ we may think $(\mathcal{A}^{(1)}(\beta))^{(u)}$ as a constraint database (in the sense of \cite{HeiKu04} and \cite{HeiKu07}) which allows to evaluate the univariate polynomial $F^{(u)}\in\C[Y]$ as often as desired for arbitrary inputs $y\in\A^1$, using each time a number of arithmetic operations in $\C$, namely $\deg_Y F$, which does not depend on the non--scalar size of $\beta$. 

Moreover $\mathcal{A}$ satisfies the following condition:
\textit{
\begin{itemize}
	\item[$(D)$]  There exist non--decreasing real valued functions $C_1 \geq 0$ and $C_2 \geq 0$ depending on dynamic integer vectors, such that for $L_{\beta}$ and $L_{\mathcal{A}(\beta)}$, being the non--scalar sizes of the circuits $\beta$ and $\mathcal{A}(\beta)$, the inequality
$$L_{\mathcal{A}(\beta)}\leq C_1 (\psi(\beta))L_{\beta} + C_2(\psi(\beta))$$
holds. 
\end{itemize}
}

Let now $\mathcal{A}$ be an \emph{arbitrary}, essentially division--free algorithm of our extended computation model which solves the general instance of any flat family of zero--dimensional elimination problems and let $\beta$ be an input circuit for $\mathcal{A}$ which represents a particular family of such problems. Let $F$ be the associated elimination polynomial.

Then the complexity of the algorithm $\mathcal{A}$ is only competitive with known elimination algorithms if we require that the number of \emph{essential} additions and multiplications of $\mathcal{A}_{\text{final}}(\beta)$ is bounded by $2\cdot \deg_Y F$. This leads us to the requirement that $\mathcal{A}$ must be decomposable in two subalgorithms $\mathcal{A}^{(1)}$ and $\mathcal{A}^{(2)}$ as above.  

Therefore any elimination algorithm of our extended computation model which is claimed to improve upon known algorithms for \emph{all} admissible input circuits $\beta$, must have this architectural structure. In particular, such an algorithm cannot call the input circuit $\beta$ when the output variable $Y$ became already involved. This justifies the architectural restriction we made in the statement and proof of Theorem \ref{def: main theorem model}.

Observe that the final results of the circuit $\mathcal{A}^{(1)}(\beta)$ form a geometrically robust constructible map defined on the parameter domain of the circuit $\beta$. For a given parameter instance, the value of this map allows to compute the value of the coefficient vector of the elimination polynomial $F$ on this instance.

\enter
Moreover, the competitivity of $\mathcal{A}$ with known elimination algorithms requires that $\mathcal{A}$ must satisfy condition $(D)$.

From Theorem \ref{def: main theorem model} and its proof we deduce now the lower bound
$$\text{max}\{ C_1(\psi(\beta_n)), C_2(\psi(\beta_n))\} = \Omega (\frac{\delta_n}{L_{\beta_n}}),$$
where $\delta_n$ is the geometric degree of the subvariety of $\A^r\times\A^{n+1}$ defined by the polynomials $G_1^{(n)},\dots,G_n^{(n)},Y-H^{(n)}$ (observe $\delta_n = 2^n$). Adding to $\beta_n$ a suitable addition node we obtain a totally division--free new circuit $\beta_n^*$ which represents $G_1^{(n)},\dots,G_n^{(n)}$ and $Y-H^{(n)}$. Observe that for each $(s,u)\in\A^1\times\A^n$ the degree pattern of the polynomials $G_1^{(n)},\dots,G_n^{(n)},Y-H(s,u,X)$ is constant and the system degree is $\delta_n$. The polynomial $F^{(n)}$ is the output of the Kronecker algorithm applied to $\beta_n^*$ and the variable $Y$. Therefore the algorithm $\mathcal{A}$ produces on $\beta_n$ the same output as the Kronecker algorithm applied to $\beta_n^*$ and the variable $Y$. We conclude now from $L_{\beta_n^*}=O(n)$ that the Kronecker algorithm is nearly optimal in our extended computation model.

\enter

In our computation model, algorithms are transformations of parameterized arithmetic circuits over one and the same parameter domain. This represents a substantial ingredient for the proof of Theorem \ref{def: main theorem model}. If we allow branchings which lead to subdivisions of the parameter domain of the input circuit, the conclusion of Theorem \ref{def: main theorem model} may become uncertain (see \cite{GHK11}).

Our computation model is also restrictive in another sense:\\
suppose that there is given an essentially division--free, robust parameterized arithmetic circuit $\beta_n$ evaluating the polynomial $H^{(n)}$ as in the proof of Theorem \ref{def: main theorem model} and an essentially division--free procedure $\mathcal{B}$ of our extended computation model which recomputes $H^{(n)}$ from the input $\beta_n$. Then, the output of $\mathcal{B}$ on $\beta_n$ is an essentially division--free robust parameterized arithmetic circuit of size $\Omega(2^n)$, although the size of $\beta_n$ is $O(n)$.    
 
\subsection{The elimination of a block of existential quantifiers}
\label{sec: The elimination of a block of existential quantifiers}

Let notations be the same as in the proof of Theorem \ref{def: main theorem model} in Section \ref{ 6.5.1 flat}. Let $n\in\N$, $S_1,\dots,S_n$ new indeterminates, $S:=(S_1,\dots,S_n),$ $\hat{G}_1^{(n)}:=X_1^2-X_1-S_1, \dots,\hat{G}_n^{(n)}:=X_n^2-X_n-S_n$ and again $H^{(n)}:=\sum_{1\leq i\leq n} 2^{i-1}X_i + T \prod_{1\leq i\leq n}(1 + (U_i-1)X_i).$

Observe that the polynomials $\hat{G}_1^{(n)},\dots,\hat{G}_n^{(n)}$ form a reduced regular sequence in $\C[S,T,U,X]$ and that they define a subvariety $\hat{V}_n$ of the affine space $\A^n\times\A^1\times\A^n\times\A^n$ which is isomorphic to $\A^n\times\A^1\times\A^n$ and hence irreducible and of dimension $2n+1$. Moreover, the morphism $\hat{V}_n \to \A^n\times\A^1\times\A^n $ which associates to any point $(s,t,u,x)\in\hat{V}_n$ the point $(s,t,u)$, is finite and generically unramified. Therefore the morphism $\hat{\pi}_n:\hat{V}_n \to \A^n\times\A^1\times\A^n\times\A^1$ which associates to any $(s,t,u,x)\in \hat{V}_n$ the point $(s,t,u,H^{(n)}(t,u,x))\in \A^n\times\A^1\times\A^n\times\A^1$ is finite and its image $\hat{\pi}_n(\hat{V}_n)$ is a hypersurface of $\A^n\times\A^1\times\A^n\times\A^1$ with irreducible minimal equation $\hat{F}^{(n)}\in \C[S,T,U,Y]$.

Hence $\hat{G}_1^{(n)}=0,\dots,\hat{G}_n^{(n)}=0$ and $H^{(n)}$ represent a flat family of zero--dimensional elimination problems whose associated elimination polynomial is just $\hat{F}^{(n)}$.

Observe that $\deg \hat{F}^{(n)}= \deg_Y \hat{F}^{(n)}=2^n$ and that for $0\in\A^n$ the identity $$\hat{F}^{(n)}(0,T,U,Y)=F^{(n)}(T,U,Y) \text{\ws holds,}$$
where $F^{(n)}$ is the elimination polynomial associated with the flat family of zero dimensional elimination problems given by $X_1^2-X_1=0$, $\dots$, $X_n^2-X_n=0$ and $H^{(n)}$. Since $\hat{F}^{(n)}$ is irreducible, any equation of $\C[S,T,U,Y]$ which defines $\hat{\pi}_n(\hat{V}_n)$ in $\A^n\times\A^1\times\A^n\times\A^1$ is without loss of generality a power of $\hat{F}^{(n)}$.

We consider now $S_1,\dots,S_n,T,U_1,\dots,U_n$ as basic parameters, $X_1,\dots,X_n$ as input and $Y$ as output variables.

Let $\mathcal{A}'$ be an essentially division--free procedure of our extended computation model satisfying the following condition: \\
$\mathcal{A}'$ accepts as input any robust parameterized arithmetic circuit $\beta$ which represents the general instance of a flat family of zero--dimensional elimination problems with associated elimination polynomial $F$ and $\mathcal{A}'_{\text{final}}(\beta)$ has a single input $Y$ and a single final result which defines the same hypersurface as $F$.

With this notions and notations we have the following result.

\begin{proposition}
\label{proposition A}
There exist an ordinary division--free arithmetic circuit $\hat{\beta}_n$ of size $O(n)$ over $\C$ with inputs $S_1,\dots,S_n$, $T$, $U_1,\dots,U_n$, $X_1,\dots,X_n$ and final results $\hat{G}_1^{(n)},\dots,\hat{G}_n^{(n)},H^{(n)}$. The essentially division--free, robust parameterized arithmetic circuit $\hat{\gamma}_n:=\mathcal{A}'_{\text{final}}(\hat{\beta}_n)$ depends on the basic parameters $S_1,\dots,S_n$, $T$, $U_1,\dots,U_n$ and the input $Y$ and its single final result is a power of $\hat{F}^{(n)}$. The circuit $\hat{\gamma}_n$ performs at least $\Omega(2^{\frac{n}{2}})$ essential multiplications and at least $\Omega(2^n)$ multiplications with parameters. As ordinary arithmetic circuit over $\C$ with inputs $S_1,\dots,S_n$, $T$, $U_1,\dots,U_n$ and $Y$, the circuit $\hat{\gamma}_n$ has non--scalar size at least $\Omega(2^n)$.
\end{proposition}
\begin{proof}
The existence of an ordinary division--free arithmetic circuit as in the statement of Proposition \ref{proposition A} is evident. The rest follows immediately from the proof of Theorem \ref{def: main theorem model} in Section \ref{ 6.5.1 flat} by restricting the parameter domain $\A^n\times\A^1\times\A^n$ of $\hat{\beta}_n$ and $\hat{\gamma}_n$ to $\A^1\times\A^n$, i.e., by specializing $S$ to $0\in\A^n$. Observe that this restriction of $\hat{\gamma}_n$ may become an inconsistent circuit, but this does not affect the argumentation which is based on the consideration of suitable geometrically robust constructible functions.
\end{proof}
\enter

Suppose now that there is given another essentially division--free procedure $\mathcal{A}''$ of our extended computation model satisfying the following condition:\\
$\mathcal{A}''$ accepts as input any robust arithmetic circuit $\beta$ which represents the general instance of a flat family of zero--dimensional elimination problems with associated elimination polynomial $F$ and there exists a Boolean circuit $b$ in as many variables as the number of final results of $\mathcal{A}_{\text{final}}''(\beta)$ such that the algebraic variety defined by $F$ coincides with the constructible set which can be described by plugging into $b$ the final results of $\mathcal{A}_{\text{final}}''(\beta)$ as polynomial equations. 	

Observe that this represents the most general architecture we can employ to adapt in the spirit of Section \ref{sec:Model-discussion-programs and algorithms} our extended computation model for \emph{functions} to \emph{parametric decision problems}. 

Let $s\in\N$ and $A_1,\dots,A_s$ new indeterminates with $A:=(A_1,\dots,A_s)$. We suppose that there is given an essentially division--free procedure $\mathcal{B}$ of our extended computation model which accepts as input any essentially division--free, robust parameterized arithmetic circuit $\gamma$ with the basic parameters $A_1,\dots,A_s$ and the input variable $Y$, such that $\mathcal{B}_{\text{final}}(\gamma)$ represents, by its output nodes, in $\C[A,Y]$ the multiplicative decomposition of the final results of $\gamma$ by their greatest common divisor and complementary factors. 

In this sense, we call the procedure $\mathcal{B}$ a \emph{GCD algorithm}.

\enter
Let $\psi_{\mathcal{A}''}$ and $\psi_{\mathcal{B}}$ be the given low level programs of the procedures $\mathcal{A}''$ and $\mathcal{B}$. We require that $\mathcal{A}''$ and $\mathcal{B}$ are competitive with known algorithms which solve the same tasks. Following our argumentation in Section \ref{ 6.5.1 flat} we may therefore suppose that there exist four non--decreasing real valued functions $C_1\geq 0$, $C_2\geq 0$ and $D_1\geq 0$, $D_2\geq 0$ which depend on dynamic integer vectors and which satisfy the estimates
$$L_{\mathcal{A}''(\beta)}\leq C_1(\psi_{\mathcal{A}''}(\beta))L_{\beta} + C_2(\psi_{\mathcal{A}''}(\beta))$$ 
and
$$L_{\mathcal{B}(\gamma)}\leq D_1(\psi_{\mathcal{B}}(\gamma))L_{\gamma} + D_2(\psi_{\mathcal{B}}(\gamma)).$$

We consider again the ordinary division--free arithmetic circuit $\hat{\beta}_n$ of Proposition \ref{prop: T subset and composition} which represents the polynomials $\hat{G}_1^{(n)},\dots,\hat{G}_n^{(n)}$ and $H^{(n)}$.

With these notions and notations we may now formulate the following statement.

\begin{theorem}
\label{proposition C D}
Let assumptions and notations be as before. Then we have 
	\[\text{max}\{ C_i(\psi_{\mathcal{A}''}(\hat{\beta}_n)),
	               D_i(\psi_{\mathcal{B}}(\mathcal{A}_{\text{final}}''(\hat{\beta}_n)));
	               i=1,2\} = \Omega(\frac{2^{\frac{n}{2}}}{n})  
\]
\end{theorem}
\begin{proof}
If we plug into the Boolean circuit $b$ the final results of $\mathcal{A}_{\text{final}}''(\hat{\beta}_n)$ as polynomial equations, we obtain by assumption a description of the hypersurface $\hat{\pi}(\hat{V}_n)$ of the affine space $\A^n \times \A^1 \times \A^n \times \A^1$. This implies that between the final results of $\mathcal{A}_{\text{final}}''(\hat{\beta}_n)$ there exists a selection, say the polynomials $P_1,\dots,P_m$ and $R_1,\dots,R_t$ of $\C[S,T,U,Y]$ such that the formula
$$P_1=0 \wedge \dots \wedge P_m=0 \wedge R_1\neq 0 \wedge \dots \wedge R_t\neq 0$$
defines a nonempty Zariski open (and dense) subset of the irreducible surface $\hat{\pi}(\hat{V}_n)$ of $\A^n\times\A^1\times\A^n\times\A^1$. 

Let $R:=R_1 \dots R_t$ and observe that the greatest common divisor of $P_1,\dots,P_m$ has the form $(\hat{F}^{(n)})^q \cdot Q$, where $q$ belongs to $\N$ and $Q$ is the greatest common divisor of $P_1,\dots,P_m,R$. Therefore we may compute $(F^{(n)})^q$ in the following way: erasing suitable nodes from the circuit $\mathcal{A}_{\text{final}}''(\hat{\beta}_n)$ and adding $t - 1$ multiplication nodes we obtain two robust parameterized arithmetic circuits $\gamma_1^{(n)}$ and $\gamma_2^{(n)}$ with basic parameters $S_1,\dots,S_n$, $T$, $U_1,\dots,U_n$ and input $Y$ whose final results are $P_1,\dots,P_m$ and $P_1,\dots,P_m,R$ respectively.

Between the final results of $\mathcal{B}_{\text{final}}(\gamma_1^{(n)})$ and $\mathcal{B}_{\text{final}}(\gamma_2^{(n)})$ are the polynomials\\ 
$(\hat{F}^{(n)})^q\cdot Q$ and $Q$. Applying the procedure $\mathcal{B}$ to the union of $\mathcal{B}_{\text{final}}(\gamma_1^{(n)})$ and $\mathcal{B}_{\text{final}}(\gamma_2^{(n)})$ we obtain finally an essentially division--free, robust parameterized arithmetic circuit with basic parameters $S_1,\dots,S_n$, $T$, $U_1,\dots,U_n$ and input $Y$ whose single final result is $(\hat{F}^n)^q$.  

Joining the circuits $\mathcal{A}''(\hat{\beta}_n)$, $\mathcal{B}_{\text{final}}(\gamma_1^{(n)})$, $\mathcal{B}_{\text{final}}(\gamma_2^{(n)})$ and the final division node we obtain an ordinary arithmetic circuit of non--scalar size at most
$$1+3L_{\mathcal{B}(\mathcal{A}_{\text{final}}''(\hat{\beta}_n))}\leq$$
$$1+3(D_1(\psi_{\mathcal{B}}(\mathcal{A}_{\text{final}}''(\hat{\beta}_n)))L_{\mathcal{A}_{\text{final}}''(\hat{\beta}_n)}+
	D_2(\psi_{\mathcal{B}}(\mathcal{A}_{\text{final}}''(\hat{\beta}_n))))\leq$$ 
$$1+3C_1(\psi_{\mathcal{A}}(\hat{\beta}_n)) D_1(\psi_{\mathcal{B}}(\mathcal{A}_{\text{final}}''(\hat{\beta}_n)))L_{\hat{\beta}_n}+$$
$$3C_2(\psi_{\mathcal{A}}(\hat{\beta_n}))D_1(\psi_{\mathcal{B}}(\mathcal{A}_{\text{final}}''(\hat{\beta}_n)))+D_2(\psi_{\mathcal{B}}(\mathcal{A}_{\text{final}}''(\hat{\beta}_n))).$$
	
On the other hand we deduce from Theorem \ref{def: main theorem model} 
$$L_{\hat{\beta}_n}=O(n) \text{\ and\ } 1+3L_{\mathcal{B}(\mathcal{A}_{\text{final}}''(\hat{\beta}_n))} = \Omega(2^n).$$

This implies the estimate of Theorem \ref{proposition C D}.
\end{proof}
\enter

In a simple minded understanding, Theorem \ref{proposition C D} says that in our extended computation model either the elimination of a single existential quantifier block in a prenex first--order formula of the elementary language of $\C$ or the computation of a greatest common divisor of a finite set of circuit represented polynomials requires \emph{exponential time}. Complexity results in this spirit were already obtained in \cite{GH01} and \cite{CaGiHeMaPa03} (compare also Proposition \ref{proposition A} and Observation in Section \ref{ 6.5.1 flat}).
 
\subsection{Arithmetization techniques for Boolean circuits}
\label{Arithmetization techniques for Boolean circuits}
Let $m\in\N$ and let $0,1$ and $Z_1,\dots,Z_m$ be given constants and variables. Let $Z:=(Z_1,\dots,Z_m)$. Following the context we shall interpret $0,1$ as Boolean values or the corresponding complex numbers and $Z_1,\dots,Z_m$ as Boolean variables or indeterminates over $\C$. With $\wedge,\vee,\bar{\ws}$ we denote the Boolean operations ``and'', ``or'' and ``not''. A Boolean circuit $b$ with inputs $Z_1,\dots,Z_m$ is a DAG whose indegree zero nodes are labelled by $0,1$ and $Z_1,\dots,Z_m$ and whose inner nodes have indegree two or one. In the first case an inner node is labelled by $\wedge$ or $\vee$ and in the second by $\bar{\ws}$. Some inner nodes of $b$ become labelled as outputs. We associate with $b$ a semantics as follows:
\begin{enumerate}
	\item[-] indegree zero nodes which are labelled by $0,1$ become interpreted by the corresponding constant functions $\{ 0,1 \}^m \to \{ 0,1 \}$,
	\item[-] indegree zero nodes which are labelled by $Z_1,\dots,Z_m$ become interpreted by the corresponding projection function $\{ 0,1 \}^m \to \{ 0,1 \}$,
	\item[-] let $\rho$ be an inner node of $b$ of indegree two whose parent nodes $\rho_1$ and $\rho_2$ are already interpreted by Boolean functions $g_{\rho_1},g_{\rho_2}:\{ 0,1 \}^m \to \{ 0,1 \}$. If $\rho$ is labelled by $\wedge$, we interpret $\rho$ by the Boolean function $g_{\rho}:=g_{\rho_1} \wedge g_{\rho_2}$ and if $\rho$ is labelled by $\vee$, we interpret $\rho$ by the Boolean function $g_{\rho}:=g_{\rho_1} \vee g_{\rho_2}$,	
	\item[-] let $\rho$ be an inner node of $b$ of indegree one whose parent node $\rho'$ became already interpreted by a Boolean function $g_{\rho'}:\{ 0,1 \}^m \to \{ 0,1 \}$. Then we interpret $\rho$ by the Boolean function $g_{\rho}:=\ol{g_{\rho'}}$.   
\end{enumerate}

For a node $\rho$ of $b$ we call $g_{\rho}$ the \emph{intermediate result} of $b$ at $\rho$. If $\rho$ is an output node, we call $g_{\rho}$ a \emph{final result} of $b$. 

An arithmetization $\beta$ of the Boolean circuit $b$ consists of the same DAG as $b$ with a different labelling as follows.

Let $U,V$ be new indeterminates over $\C$. The constants $0,1$ become interpreted by the correspondent complex numbers and $Z_1,\dots,Z_m$ as indeterminates over $\C$. Let $\rho$ be an inner node of $\beta$. If $\rho$ has indegree two, then $\rho$ becomes labelled by a polynomial $R_{\rho}\in\Z[U,V]$ and if $\rho$ has indegree one by a polynomial $R_{\rho}\in\Z[U]$. The output nodes of $\beta$ and $b$ are the same.

Representing for each inner node $\rho$ of $\beta$ the polynomial $G_{\rho}$ by a division--free ordinary arithmetic circuit over $\Z$ in the inputs $U,V$ or $U$, we obtain an ordinary division--free arithmetic circuit over $\Z$ in the inputs $Z_1,\dots,Z_m$.

Just as we did in Section \ref{sec:Model-discussion-simplified} we may associate with $\beta$ a semantics which determines for each node $\rho$ of $\beta$ a polynomial $G_{\rho}\in\Z[Z]$. We say that $\beta$ is an \emph{arithmetization} of the Boolean circuit $b$ if the following condition is satisfied:\\
for any node $\rho$ of $b$ and any argument $z\in \{ 0,1 \}^m$ the Boolean value $g_{\rho}(z)$ coincides with the arithmetic value $G_{\rho}(z)$ (in a somewhat imprecise notation: $g_{\rho}(z)=G_{\rho}(z)$).	

An example of an arithmetization procedure is given by the map which associates to each node $\rho$ of $b$ a polynomial $[g_{\rho}]$ of $\Z[Z]$ satisfying the following conditions:
\begin{enumerate}
	\item[-] $[0]=0$, $[1]=1$, $[Z_1]=Z_1,\dots,[Z_m]=Z_m$ 
	\item[-] for $\rho$ an inner node of indegree two of $b$ with parents $\rho_1$ and $\rho_2$:
	$$ [g_{\rho}]=[g_{\rho_1}]\cdot[g_{\rho_2}] \text{\ if\ the\ label\ of\ $\rho$\ is\ $\wedge$}$$	
	and
	$$ [g_{\rho}]=[g_{\rho_1}]+[g_{\rho_2}]-[g_{\rho_1}]\cdot[g_{\rho_2}] \text{\ if\ the\ label\ of\ $\rho$\ is\ $\vee$}$$
  \item[-] for $\rho$ an inner node of indegree one of $b$ with parent $\rho'$:
  $$[g_{\rho}]=1-[g_{\rho'}].$$   
\end{enumerate}

The resulting arithmetic circuit is called the \emph{standard arithmetization} of $b$ (see, e.g., \cite{shamir92} and \cite{BF91}).


Let $n,r\in\N$ and $U_1,\dots,U_r,X_1,\dots,X_n$ be new variables. For $m:=n+r$ we replace now $Z$ by $U$ and $X$, where $U:=(U_1,\dots,U_r)$ and $X:=(X_1,\dots,X_n)$. We shall interpret $U_1,\dots,U_r$ as parameters and $X_1,\dots,X_n$ as input variables.

Let $b$ be a Boolean circuit with the inputs $U_1,\dots,U_r,X_1,\dots,X_n$ and just one final result $h:\{ 0,1 \}^r\times\{ 0,1 \}^n \to \{ 0,1 \}$.

We wish to describe the set of instances $u\in\{ 0,1 \}^r$ where $h(u,X_1,\dots,X_n)$ is a satisfiable Boolean function.

For this purpose let us choose an arithmetization $\beta$ of $b$. We interpret $\beta$ as an ordinary arithmetic circuit over $\Z$ with the parameters $U_1,\dots,U_r$ and the inputs $X_1,\dots,X_n$. The single final result of $\beta$ is a polynomial $H\in\Z[U,X]$ which satisfies for any $u\in \{ 0,1 \}^r$, $x\in \{ 0,1 \}^n$ the following condition:
$$h(u,x)=H(u,x).$$
Without loss of generality we may suppose that the polynomials $X_1^2-X_1$, $\dots$, $X_n^2-X_n$ are intermediate results of $\beta$. We relabel now $\beta$ such that these polynomials and $H$ become the final results of $\beta$. Observe that $X_1^2-X_1=0,\dots,X_n^2-X_n=0$ and $H$ represent a flat family of zero--dimensional elimination problems.

Let $Y$ be a new indeterminate and let $F\in\Z[U,Y]$ the associated elimination polynomial. One verifies easily the identity 
$$F(U,Y)=\prod_{x\in \{ 0,1 \}^n}(Y-H(U,x)).$$ 

Let $\mathcal{A}$ be an essentially division--free procedure of our extended computation model which solves algorithmically the general instance of any flat family of zero--dimensional elimination problems. Then $\beta$ is an admissible input for $\mathcal{A}$ and there exists an integer $q\in\N$ such that $F^q$ is the final result of $\mathcal{A}_{\text{final}}(\beta)$.

We consider now the task to count for any $u\in\{ 0,1 \}^r$ the number $k$ of instances $x\in\{ 0,1 \}^n$ with $h(u,x)=1$.

The polynomial $F^q$ encodes two possible solutions of this task. 

The first solution is the following: let $l$ be the order of the univariate polynomial $F^q(u,Y)$ at zero. Then $q$ divides $l$ and we have $k=2^n-\frac{l}{q}$.

The second and more interesting solution is the following: write $F^q = Y^{2^n q}+ \varphi_1 Y^{2^n q -1} +\dots +\varphi_{2^n q}$ with $\varphi_1,\dots,\varphi_{2^n q} \in \Z[U]$. Then $\varphi_1(u)$ is an integer which is divisible by $q$ and we have $k=-\frac{\varphi_1(u)}{q}$.

Observe also $\deg \varphi_1 \leq \deg_{U} H$.

These considerations show the relevance of an \emph{efficient} evaluation of the polynomial $F^q$ (e.g., by the circuit $\mathcal{A}_{\text{final}}(\beta)$). 

We ask therefore whether $\mathcal{A}_{\text{final}}(\beta)$ can be polynomial in the size of the Boolean circuit $b$. The following example illustrates that the answer may become negative.

\enter
In the sequel we are going to exhibit for $r:=2n+1$ a Boolean circuit $b$ of size $O(n)$ which evaluates a function $h:\{ 0,1 \}^r\times \{ 0,1 \}^n \longrightarrow \{ 0,1 \}$ such that the standard arithmetization $\beta$ of $b$ represents a flat family of zero--dimensional elimination problems with associated elimination polynomial $F$ and such that any essentially division--free procedure $\mathcal{A}$ of our extended computation model that accepts the input $\beta$ and computes by means of $\mathcal{A}_{\text{final}}(\beta)$ a power of $F$, requires time $\Omega(2^n)$ for this task. This means that it is unlikely that algorithms designed following the paradigm of object--oriented programming are able to evaluate the polynomial $\varphi_1$ efficiently.

On the other hand, since the degree of $\varphi_1$ is bounded by $\deg_U H$ and therefore ``small'', there exists a polynomial time interactive protocol which checks for any $u\in\{ 0,1 \}^r$ and any $c\in\Z$ the equation $\varphi_1(u)=c$. Thus this problem belongs to the complexity class $IP$ (see \cite{Lund92} for details).

\enter
We are now going to exhibit an example of a Boolean circuit which highlights the unfeasibility of our computation task. 

For this purpose let $r:=2n+1$ and $S_1,\dots,S_n,T,U_1,\dots,U_n$ parameters and $X_1,\dots,X_n$ input variables and let $S:=(S_1,\dots,S_n)$ and $U:=(U_1,\dots,U_n)$.

We consider the Boolean function $h:\{ 0,1 \}^{2n+1} \times \{ 0,1 \}^{n} \to \{ 0,1 \} $ defined by the Boolean formula 
$$\phi:= \bigwedge_{1\leq i\leq n} (\ol{X_i} \vee (S_i \wedge X_i)) \vee (T \wedge \bigwedge_{1\leq i\leq n}(\ol{X_i} \vee (U_i\wedge X_i))).$$
From the structure of the formula $\phi$ we infer that $h$ can be evaluated by a Boolean circuit $b$ of size $O(n)$ in the inputs $S_1,\dots,S_n,T,U_1,\dots,U_n$.

Let $\beta$ be the standard arithmetization of the Boolean circuit $b$ and let $H$ be the final result of $\beta$. Observe that the total, and hence the non--scalar size of $\beta$ is $O(n)$. Then we have
$$H=\prod_{1\leq i\leq n}(1+(S_i-1)X_i)+(1-\prod_{1\leq i\leq n}(1+(S_i-1)X_i))T \prod_{1\leq i\leq n}(1+(U_i-1)X_i).$$
Observe that the equations $X_1^2-X_1=0,\dots,X_n^2-X_n=0$ and the polynomial $H$ represent a flat family of zero--dimensional elimination problems. Let $F$ be the associated elimination polynomial. Then $F$ can be written as 
$$F=Y^{2^n}+B_1Y^{2^n-1}+\dots+B_{2^n}= \prod_{0\leq j< 2^n}(Y-(\prod_{1\leq i\leq n}S^{[j]_i}+(1-\prod_{1\leq j\leq n} S^{[j]_i})T \prod_{1\leq i\leq n}U_i^{[j]_i} ))$$ 
with 
$$B_k=(-1)^k \sum_{0\leq j_1<\dots <j_k<2^n} 
\prod_{1\leq h\leq k} 
(\prod_{1\leq i\leq n}S_i^{[j_h]_i}+
(1-\prod_{1\leq i\leq n}S_i^{[j_h]_i})T \prod_{1\leq i\leq n}U_i^{[j_h]_i})$$
for $1\leq k\leq 2^n$. 

Let
$$L_k:=(-1)^{k}\sum_{0\leq j_1<\dots <j_k<2^n} \sum_{1\leq h\leq k} \prod_{1\leq i\leq n} S_i^{[j_1]_i}\dots(1-\prod_{1\leq i\leq n}S_i^{[j_h]_i})\dots \prod_{1\leq i\leq n}S_i^{[j_k]_i} \prod_{1\leq i\leq n}U_i^{[j_h]_i},$$
where $1\leq k\leq 2^n$.

Then we have
$$B_k=(-1)^k \sum_{0\leq j_1<\dots <j_k<2^n} 
\prod_{1\leq i\leq n} S_i^{[j_1]_i} \dots \prod_{1\leq i\leq n} S_i^{[j_k]_i} + L_k.T + \text{\small{\ terms\ of\ higher\ degree\ in\ $T$}}$$ 

Let $\epsilon:\A^{2^n}\to\A^{2^n}$ be the morphism of affine spaces which assigns to each point $z\in\A^{2^n}$ the values of the elementary symmetric functions in $2^n$ variables at $z$. Observe that the Jacobian of $\epsilon$ at $(\prod_{1\leq i\leq n} S^{[j]_i})_{0\leq j< 2^n}$ is a non--singular $(2^n \times 2^n)$--matrix $N(S)$. The polynomials $L_k,1\leq k\leq 2^n$ are obtained by applying $N(S)$ to $((1-\prod_{1\leq i\leq n}S^{[j]_i})\prod_{1\leq i\leq n} U_i^{[j]_i})_{0\leq j< 2^n}$. Since the monomials $\prod_{1\leq i\leq n}U_i^{[j]_i},0\leq j< 2^n,$ are linearly independent over $\C(S)$ we conclude that the polynomials $L_k,1\leq k\leq 2^n$ have the same property.

With this preparation we are now able to repeat textually the arguments in the proof of Theorem \ref{def: main theorem model} of Section \ref{ 6.5.1 flat} in order to show the following statement.

\begin{theorem}
\label{theorem 3}
Let assumptions and notations be as before and let $\mathcal{A}$ be an essentially division free procedure of our extended computation model which accepts the arithmetic circuit $\beta$ as input. Suppose that $\mathcal{A}_{\text{final}}(\beta)$ has a unique final result and that it is a power of the elimination polynomial $F$. Then the non--scalar size of $\mathcal{A}_{\text{final}}(\beta)$ is at least $\Omega(2^n)$.
\end{theorem}


\subsection{The multivariate resultant}


Let $U_1,\dots,U_m$ be basic parameters and let $X_1,\dots,X_n$ or $X_0,\dots,X_n$ be input variables. Let $G_1,\dots,G_n,H\in\C[U,X]$ be a flat family of zero--dimensional elimination problems such that for any $u\in\A^m$ the homogenizations of $G_1(u,X),\dots,G_n(u,X)$ (by $X_0$) have no common zero at infinity. Let $F\in\C[U,Y]$ be the corresponding elimination polynomial and let $R\in\C[U]$ be the (multivariate) resultant of the homogenizations of $G_1,\dots,G_n,H$. Then we have $R=F(U,0)$.

On the other hand, $F$ is the resultant of the homogenizations of $G_1,\dots,G_n$ and $H-Y$. Thus, multihomogeneous resultants and elimination polynomials of flat families of zero--dimensional elimination problems are closely related from the algebraic point of view. 

From the computational point of view this relation is more intricate.

To the degree pattern of the homogenizations of $G_1,\dots,G_n,H$ there corresponds a generic resultant. We may take a computation of this resultant and specialise its inputs to the coefficients of $G_1,\dots,G_n,H-Y$ with respect to $X_1,\dots,X_n$. If this specialized computation can be simplified by means of reductions, we may expect to gain something. Proceeding in this way we obtain an algorithm which may be interpreted as an elementary routine of our computation model. Observe that this elementary routine does not use joins of two subroutines that contain each a recursion.

We are going to show that reductions do not produce a general improvement of traditional resultant computations in the sense described above.

\begin{theorem}
Consider the following, with respect to $X_0,\dots,X_n$ homogeneous, polynomials
\[
X^2_1-X_0X_1,\dots,X_n^2-X_0X_n, YX_0^n-\sum_{1\leq i\leq n} 2^{i-1}X_0^{n-1}X_i-T \prod_{1\leq i\leq n} (X_0+(U_i-1)X_i),	
\]
which are supposed to be given by an essentially division--free arithmetic circuit $\beta$ in the basic parameters $T,U_1,\dots,U_n,Y$ and the input variables $X_0,\dots,X_n$. Suppose furthermore that $Y$ is the last basic parameter introduced by $\beta$ and observe that such a circuit $\beta$ of size $O(n)$ exists. Let $\mathcal{A}$ be an essentially division--free elementary routine of our computation model which on input $\beta$ evaluates the resultant of the polynomials above with respect to the variables $X_0,\dots,X_n$. Suppose that $\mathcal{A}$ does not use joins of two subroutines that contain each a recursion. Then the output circuit $\mathcal{A}_{\text{final}}(\beta)$ has size at least $\Omega(2^n)$.
\end{theorem}

\begin{proof}
Taking into account that $Y$ is the last basic parameter introduced by $\beta$, that all recursive subroutines of $\mathcal{A}$ are isoparametric and that $\mathcal{A}$ does not use joins of two subroutines that contain each a recursion, we may decompose $\mathcal{A}$ in two subroutines $\mathcal{A}^{(1)}$ and $\mathcal{A}^{(2)}$ such that $\mathcal{A}^{(1)}$ admits $\beta$ as input and produces a vector of parameters in $T,U_1,\dots,U_n$ which themselves constitute the inputs of $\mathcal{A}^{(2)}$ (the argument is somewhat tedious and is not given here). Only the routine $\mathcal{A}^{(2)}$ introduces the parameter $Y$. Therefore $\mathcal{A}=(\mathcal{A}^{(1)},\mathcal{A}^{(2)})$ constitutes a procedure of our computation model with basic parameters $T,U_1,\dots,U_n$, input variables $X_0,\dots,X_n$ and output variable $Y$. Observing that the resultant of the given polynomials equals the elimination polynomial of the flat family of zero--dimensional elimination problems given by $X^2_1-X_1,\dots,X^2_n-X_n, \sum_{1\leq i\leq n} 2^{i-1}X_i + T\prod_{1\leq i\leq n} (1+(U_i-1)X_i)$, we conclude from Theorem \ref{def: main theorem model} that the size of $\mathcal{A}_{\text{final}}(\beta)$ is at least $\Omega(2^n)$.
\end{proof}


\subsection{A family of hard elimination polynomials} 
\label{independent of the model}

As a major result of this paper, we are now going to exhibit an infinite family of parameter dependent elimination polynomials which require essentially division--free, robust parameterized arithmetic circuits of exponential size for their evaluation, whereas the circuit size of the corresponding input problems grows only polynomially. This result is valid \emph{without any architectural assumption} on the algorithm that computes these elimination polynomials.

Let notations be as before and consider again for given $n\in\N$ the polynomial $H^{(n)}:= \sum_{1 \leq i \leq n} 2^{i-1}X_i + T \prod_{1 \leq i \leq n} (1+ (U_i-1)X_i) $ of Section \ref{ 6.5.1 flat}. Observe that $H^{(n)}$ can be evaluated using $n-1$ non--scalar multiplications involving $X_1,\dots,X_n$.  

The set $\mathcal{O}:=\{ \sum_{1 \leq i \leq n} 2^{i-1}X_i + t \prod_{1 \leq i \leq n} (1+ (u_i-1)X_i);(t,u_1,\dots,u_n)\in\A^{n+1} \}$ is contained in a finite--dimensional $\C$--linear subspace of $\C[X]$ and therefore $\mathcal{O}$ and its closure $\ol{\mathcal{O}}$ are constructible sets.

From \cite{GHMS09}, Section 3.3.3 we deduce the following facts:\\
there exist $K:=16n^2+2$ integer points $\xi_1,\dots,\xi_K \in\Z^n$ of bit length at most $4n$ such that for any two polynomials $f,g\in \ol{\mathcal{O}}$ the equalities $f(\xi_k)=g(\xi_k), 1\leq k \leq K$, imply $f=g$. Thus the polynomial map $ \Xi:\ol{\mathcal{O}}\to\A^K$ defined for $f\in \ol{\mathcal{O}}$ by $\Xi(f):= (f(\xi_1),\dots,f(\xi_K))$ is injective. Moreover $\mathcal{M}:=\Xi(\mathcal{O})$ is an irreducible constructible subset of $\A^K$ and we have $\ol{\mathcal{M}}=\Xi(\ol{\mathcal{O}})$. Finally, the constructible map $\phi:=\Xi^{-1}$, which maps $\mathcal{M}$ onto $\mathcal{O}$ and $\ol{\mathcal{M}}$ onto $\ol{\mathcal{O}}$, is a restriction of a geometrically robust map and therefore by Corollary \ref{proposition 1} itself geometrically robust.

For $\epsilon\in\{ 0,1 \}^n$ we denote by $\phi_\epsilon$ the map $\ol{\mathcal{M}}\to\A^1$ which assigns to each point $v\in\ol{\mathcal{M}}$ the value $\phi(v)(\epsilon)$. From Corollary \ref{proposition 1} we conclude that $\phi_{\epsilon}$ is a geometrically robust constructible function which belongs to the function field $\C(\ol{\mathcal{M}})$ of the irreducible algebraic variety $\ol{\mathcal{M}}$.

Observe that for $t\in\A^1$ and $u\in\A^n$ the identities $\phi_{\epsilon}(\Xi(H^{(n)}(t,u,X)))=\phi(\Xi(H^{(n)}(t,u,X)))(\epsilon)= ((\Xi^{-1}\circ\Xi)(H^{(n)}(t,u,X)))(\epsilon)= H^{(n)}(t,u,\epsilon)$ hold.

Let $\tilde{F}^{(n)}:= \prod_{\epsilon\in \{ 0,1\}^n } (Y-\phi_\epsilon)$. Then $\tilde{F}^{(n)}$ is a geometrically robust constructible function which maps $\ol{\mathcal{M}}\times\A^1$ (and hence $\mathcal{M}\times\A^1$) into $\A^1$. Considering again the elimination polynomial $F^{(n)}= \prod_{0 \leq j \leq 2^n-1}(Y-(j + T\prod_{1 \leq i\leq n} U_i^{[j]_i})) = \prod_{\epsilon\in \{ 0,1\}^n } (Y- H^{(n)}(T,U,\epsilon))$ of Section \ref{ 6.5.1 flat}, we have for $t\in\A^1$ and $u\in\A^n$ the identities 
\begin{equation}
\begin{array}{l}
\displaystyle \tilde{F}^{(n)}(\Xi(H^{(n)}(t,u,X)),Y) = \prod_{\epsilon\in \{ 0,1 \}^n} (Y - \phi_\epsilon(\Xi(H^{(n)}(t,u,X))))=\\
\displaystyle  \prod_{\epsilon\in \{ 0,1 \}^n} (Y-H^{(n)}(t,u,\epsilon)) = F^{(n)}(t,u,Y)
\end{array} 
\label{equ (*)}
\end{equation}

Let $S_1,\dots,S_K$ be new indeterminates and observe that the existential first order formula of the elementary theory of $\C$, namely
\begin{equation}
\begin{array}{l}
\displaystyle (\exists X_1)\dots(\exists X_n)(\exists T)(\exists U_1)\dots(\exists U_n) (X_1^2-X_1=0 \wedge\dots \wedge X_n^2-X_n=0 \wedge\\
\displaystyle  \bigwedge_{1 \leq j \leq K} S_j=H^{(n)}(T,U,\xi_j) \wedge Y=H^{(n)}(T,U,X))
\end{array} 
\label{equ (**)}
\end{equation}

describes the constructible subset $\{ (s,y)\in\A^{K+1}; s\in\mathcal{M},y\in\A^1,\tilde{F}^{(n)}(s,y)=0 \}$ of $\A^{K+1}$. Moreover, $\tilde{F}^{(n)}$ is the greatest common divisor in $\C(\ol{\mathcal{M}})[Y]$ of all polynomials of $\C[\ol{\mathcal{M}}][Y]$ which vanish identically on the constructible subset of $\A^{K+1}$ defined by the formula (\ref{equ (**)}). Hence $\tilde{F}^{(n)}\in\C(\ol{\mathcal{M}})[Y]$ is a (parameterized) \emph{elimination polynomial}.

Observe that the polynomials contained in the formula (\ref{equ (**)}) can be represented by a totally division--free arithmetic circuit $\tilde{\beta}_n$ of size $O(n^3)$. Therefore, the formula (\ref{equ (**)}) is also of size $O(n^3)$. 

\begin{theorem}
\label{theorem model independent}
Let notations and assumptions be as before and let $\tilde{\gamma}$ be an essentially division--free, robust parameterized arithmetic circuit with domain of definition $\mathcal{M}$ such that $\tilde{\gamma}$ evaluates the elimination polynomial $\tilde{F}^{(n)}$.

Then $\tilde{\gamma}$ performs at least $\Omega(2^{\frac{n}{2}})$ essential multiplications and at least $\Omega(2^n)$ multiplications with parameters.
\end{theorem}

\begin{proof}
Let $\tilde{\gamma}$ be as in the statement of the theorem. Without loss of generality we may assume that $\tilde{\gamma}$ has a single output node which evaluates the polynomial $\tilde{F}^{(n)}$. There exists a totally division--free arithmetic circuit $\tilde{\gamma}_n$ of size $O(n^3)$ which computes at its output nodes the polynomials $H^{(n)}(T,U,\xi_k),1\leq k \leq K$. 

From (\ref{equ (*)}) we deduce that the join $\tilde{\gamma}*\tilde{\gamma}_n$ of the circuit $\tilde{\gamma}_n$ with the circuit $\tilde{\gamma}$ at the basic parameter nodes of $\tilde{\gamma}$ is an essentially division--free, robust parameterized arithmetic circuit which evaluates the elimination polynomial $F^{(n)}$. Observe that the outputs of $\tilde{\gamma}_n$ are only parameters and that only the circuit $\tilde{\gamma}$ introduces the variable $Y$. Moreover, there exists an isoparametry between $H^{(n)}$ and the outputs of $\tilde{\gamma}_n$. We may therefore think that the circuit $\tilde{\gamma}* \tilde{\gamma}_n $ is produced by an essentially division--free procedure of our extended computation model which becomes applied to the circuit $\tilde{\beta}_n$. From Theorem \ref{def: main theorem model} and its proof we deduce now that $\tilde{\gamma}*\tilde{\gamma}_n$ contains at least $\Omega(2^n)$ essential multiplications and at least $\Omega(2^n)$ multiplications with parameters. Since the size of $\gamma_n$ is $O(n^3)$, we draw the same conclusion for $\tilde{\gamma}$.
\end{proof}

Theorem \ref{theorem model independent} is essentially contained in the arguments of the proof of \cite{GH01}, Theorem 5 and \cite{CaGiHeMaPa03}, Theorem 4.

Observe that a quantifier--free description of $\mathcal{M}$ by means of circuit represented polynomials, together with an essentially division--free, robust parameterized arithmetic circuit $\tilde{\gamma}$ with domain of definition $\mathcal{M}$, which evaluates the elimination polynomial $\tilde{F}^{(n)}$, captures the intuitive meaning of an algorithmic solution of the elimination problem described by formula (11), when we restrict our attention to solutions of this kind and minimize the number of equations and branchings. In particular, the circuit $\tilde{\gamma}$ can be evaluated for any input point $(s,y)$ with $s\in\mathcal{M}$ and $y\in\A^1$ and the intermediate results of $\tilde{\gamma}$ are polynomials of $\C(\ol{\mathcal{M}})[Y]$ whose coefficients are geometrically robust constructible functions defined on $\mathcal{M}$.

\subsection{Divisions and blow ups} 
\label{divisions and blow ups}
 
We are now going to analyze the main argument of the proof of Theorem \ref{def: main theorem model} from a geometric point of view.

We recall first some notations and assumptions we made during this proof.

With respect to the indeterminates $X_1,\dots,X_n$, we considered the vector $\theta$ of coefficients of the expression
$$H=\sum_{1\leq i\leq n} 2^{i-1} X_i + T \prod_{1\leq i\leq n} (1 + (U_i -1) X_i)$$
as a polynomial map $\A^{n+1}\to\A^{2^n}$ with image $\mathcal{T}$. Recall that $\mathcal{T}$ is an irreducible constructible subset of $\A^{2^n}$.

Further, with respect to the indeterminate $Y$, we considered the vector $\varphi$ of nontrivial coefficients of the monic polynomial
$$F=\prod_{1\leq j\leq 2^{n}-1} (Y - (j+ T \prod_{1\leq i\leq n} U_i^{[j]_i}))$$     
also as a polynomial map $\A^{n+1}\to\A^{2^n}$.

One sees immediately that there exists a unique polynomial map $\eta:\mathcal{T}\to\A^{2^n}$ such that $\varphi=\eta\circ\theta$ holds. Using particular properties of $\theta$ and $\varphi$ we showed implicitly in the proof of Theorem \ref{def: main theorem model} that $\eta$ satisfies the following condition:

\begin{quote}
\textit{Let $m$ be a natural number, $\zeta:\mathcal{T}\to\A^m$ a geometrically robust constructible and $\pi:\A^m\to\A^{2^n}$ a polynomial map such that $\eta = \pi\circ\zeta$ holds. Then the condition 
$$m\geq 2^n$$
is satisfied.
}	
\end{quote}

This means that the following computational task cannot be solved efficiently:

\enter
Allowing certain restricted divisions, reduce the datum $\theta$ consisting of $2^n$ entries to a datum $\zeta$ consisting of only $m\leq 2^n$ entries such that the vector $\eta$ still may be recovered from $\zeta$ without using any division, i.e., by an ordinary division--free arithmetic circuit over $\C$.

Here the allowed divisions involve only quotients which are geometrically robust functions defined on $\mathcal{T}$. 

In order to simplify the following discussion we shall assume without loss of generality that all our constructible maps have geometrically robust extensions to $\ol{\mathcal{T}}$.

Let $f$ and $g$ be two elements of the coordinate ring $\C[\ol{\mathcal{T}}]$ of the affine variety $\ol{\mathcal{T}}$ and suppose that $g\neq 0$ holds and that the element $\frac{f}{g}$ of the rational function field $\C(\ol{\mathcal{T}})$ may be extended to a robust constructible function defined on $\ol{\mathcal{T}}$, which we denote also by $\frac{f}{g}$, since this extension is unique.

Then we have two cases: the coordinate function $g$ divides $f$ in $\C[\ol{\mathcal{T}}]$ or not. In the first case we may compute $\frac{f}{g}$, by means of an ordinary division--free arithmetic circuit over $\C$, from the restrictions to $\ol{\mathcal{T}}$ of the canonical projections $\A^{2^n}\to\A^1$. Thus $\frac{f}{g}$ belongs to the coordinate ring $\C[\ol{\mathcal{T}}]$. In the second case this is not any more true and $\C[\ol{\mathcal{T}}][\frac{f}{g}]$ is a proper extension of $\C[\ol{\mathcal{T}}]$ in $\C(\ol{\mathcal{T}})$. In both cases $\C[\ol{\mathcal{T}}][\frac{f}{g}]$ is the coordinate ring of an affine chart of the blow up of $\C[\ol{\mathcal{T}}]$ at the ideal generated by $f$ and $g$. We refer to this situation as a \emph{division blow up} which we call \emph{essential} if $\frac{f}{g}$ does not belong to $\C[\ol{\mathcal{T}}]$.

Therefore we have shown in the proof of Theorem \ref{def: main theorem model} that essential division blow ups do not help to solve efficiently the reduction task formulated before.

A similar situation arises in multivariate polynomial interpolation (see \cite{GHMS09}, Theorem 23).

Following \cite{Harris92}, Theorem 7.2.1 any rational map may be decomposed into a finite sequence of successive blow ups followed by a regular morphism of algebraic varieties. Our method indicates the interest to find lower bounds for the number of blow ups (and their embedding dimensions) necessary for an effective variant of this result.

Problem adapted methods for proving lower bounds for the number of blow ups necessary to resolve singularities would also give indications which order of complexity can be expected for efficient desingularization algorithms (see \cite{EncinasVilla00}). At this moment only upper bound estimations are known \cite{RocioBlanco09}. 

\subsection{Comments on complexity models for geometric elimination}
\label{Comments on complexity models for geometric elimination}

\subsubsection{Relation to other complexity models}

The question, whether $P\neq NP$ or $P_{\C}\neq NP_{\C}$ holds in the classical or the BSS Turing Machine setting, concerns only computational \emph{decision} problems. These, on their turn, are closely related to the task to construct \emph{efficiently}, for a given prenex existential formula, an equivalent, quantifier free one (compare \cite{BSS89}, \cite{HeintzMorgenstern93}, \cite{SS95} and \cite{BSSlibro98}). In the sequel we shall refer to this and to similar, geometrically motivated computational tasks as ``\emph{effective elimination}''.

Theorem \ref{def: main theorem model} in Section \ref{ 6.5.1 flat} does not establish a fact concerning decision problems like the $P_{\C}\neq NP_{\C}$ question. It deals with the \emph{evaluation of a function} which assigns to suitable prenex existential formulas over $\C$ \emph{canonical}, equivalent and quantifier--free formulas of the same elementary language.

Theorem \ref{def: main theorem model} says that in our computation model this function cannot be evaluated efficiently. If we admit also \emph{non--canonical} quantifier--free formulas as function values (i.e., as outputs of our algorithms), then this conclusion remains true, provided that the calculation of parameterized greatest common divisors is feasible and efficient in our model (see \cite{CaGiHeMaPa03}, Section 5).

It is not clear what this implies for the $P_{\C}\neq NP_{\C}$ question.

Intuitively speaking, our exponential lower complexity bound for effective geometric elimination is only meaningful and true for computer programs designed in a professional way by software engineers. Hacker programs are excluded from our considerations.

This constitutes an enormous difference between our approach and that of Turing machine based complexity models, which always include the hacker aspect. Therefore the proof of a striking lower bound for effective elimination becomes difficult in these models.

Our argumentation is based on the requirement of output parametricity which on its turn is the consequence of two other requirements, a functional and a non--functional one, that we may employ alternatively. More explicitly, we require that algorithms (and their specifications) are described by branching parsimonious asserted programs or, alternatively, that they behave well under reductions (see Sections \ref{sec:Model-discussion-simplified} and \ref{sec:Model-discussion-programs and algorithms}).

Let us observe that the complexity statement of Theorem \ref{def: main theorem model} refers to the relationship between input and output and not to a particular property of the output alone. In particular, Theorem \ref{def: main theorem model} does not imply that certain polynomials, discussed below, like the permanent or the Pochhammer polynomials, are hard to evaluate. 

Let notations and assumptions be as in Section \ref{ 6.5.1 flat}. There we considered for arbitrary $n\in\N$ the flat family of zero dimensional elimination problems
$$G_1^{(n)}=0,\dots,G_n^{(n)}=0, H^{(n)}$$
given by
$$G_1^{(n)}:=X_1^2-X_1,\dots,G_n^{(n)}:=X_n^2-X_n$$    
and
$$H^{(n)}:=\sum_{1\leq i\leq n}2^{i-1}X_i \ws+\ws T\prod_{1\leq i\leq n}(1+(U_i-1)X_i).$$
Let $X_{n+1},\dots,X_{3n-1}$ be new indeterminates and let us consider the following polynomials 
$$G^{(n)}_{n+1}:=X_{n+1}-2X_2-X_1,\dots,G^{(n)}_j:=X_j-X_{j-1}-2^{j-n}X_{j-n+1},\ws n+2\leq j\leq 2n-1,$$
$$G^{(n)}_{2n}:=X_{2n}-U_1X_1+X_1-1,$$
$$G_k^{(n)}:=X_k-U_{k-2n+1}X_{k-1}X_{k-2n+1}+X_{k-1}X_{k-2n+1}-X_{k-1},\ws 2n+1\leq k\leq 3n-1$$
and
$$L^{(n)}:=X_{2n-1}+TX_{3n-1}.$$
One verifies easily that $G^{(n)}_1=0,\dots,G^{(n)}_{3n-1}=0,L^{(n)}$ is another flat family of zero dimensional elimination problems and that both families have the same associated elimination polynomial, namely
$$F^{(n)}:=\prod (Y-(j+T\prod_{1\leq i\leq n}U_i^{\left[\rho \right]_i}))$$
Suppose now that there is given an essential division--free procedure $\mathcal{A}$ of our extended computation model which solves algorithmically the general instance of any given flat family of zero--dimensional elimination problems.

Let $\beta_n$ and $\beta_n^*$ be two essentially division--free, robust parameterized arithmetic circuits which implement the first and the second flat family of zero dimensional elimination problems we are considering. 

Then $\beta_n$ and $\beta_n^*$ are necessarily distinct circuits. Therefore $\mathcal{A}_{\text{final}}(\beta_n)$ and $\mathcal{A}_{\text{final}}(\beta_n^*)$ represent two implementations of the elimination polynomial $F^{(u)}$ by essentially division--free, robust parameterized arithmetic circuits.

From Theorem \ref{def: main theorem model} and its proof we are only able to deduce that the circuit $\mathcal{A}_{\text{final}}(\beta_{n})$ has non--scalar size at least $\Omega(2^n)$, but we know nothing about the non--scalar size of $\mathcal{A}_{\text{final}}(\beta_{n}^*)$. 

In the past, many attempts to show the non--polynomial character of the elimination of just one existential quantifier block in the arithmetic circuit based elementary language over $\C$, employed the reduction to the claim that an appropriate candidate family of specific polynomials was hard to evaluate (this approach was introduced in \cite{HeintzMorgenstern93} and became adapted to the BSS model in \cite{SS95}). 

The Pochhammer polynomials and the generic permanents discussed below form such candidate families.

In Section \ref{independent of the model} we exhibited a \emph{certified} infinite family of parameter dependent elimination polynomials which require essentially division--free, \emph{robust} parameterized arithmetic circuits of exponential size for their evaluation, whereas the circuit size of the corresponding input problems grows polynomially.  

Here the requirement of robustness modelizes the intuitive meaning of an algorithmic solution with few equations and branchings of the underlying elimination problem.

\subsubsection{The hacker aspect}

Let us finish this section with a word about hacking and interactive (zero--knowledge) proofs.

Hackers work in an ad hoc manner and quality attributes are irrelevant for them. We may simulate a hacker and his environment by an \emph{interactive proof system} where the prover, identified with the hacker, communicates with the verifier, i.e., the user of the hacker's program. Thus, in our view, a hacker disposes over unlimited computational power, but his behaviour is deterministic. Only his communication with the user underlies some quantitative restrictions: communication channels are tight. Hacker and user become linked by a protocol of zero--knowledge type which we are going to explain now.

Inputs are natural numbers in \emph{unary} representation. Each natural number represents a mathematical object belonging to a previously fixed abstract data type of polynomials. For example $n\in\N$ may represent the $2^n$--th Pochhammer polynomial
$$T^{\underline{2^n}} :=\prod_{0\leq j< 2^n}(T-j)$$
or the $n$--th generic permanent 
$$\text{Perm}_n:=\sum_{\tau\in \text{Sym}(n)}X_{1\tau(1)},\dots,X_{n\tau(n)},$$
where $T$ and $X_{11},\dots,X_{nn}$ are new indeterminates and $\text{Sym}(n)$ denotes the symmetric group operating on $n$ elements.

On input $n\in\N$ the hacker sends to the user a division--free labelled directed acyclic graph $\Gamma_n$ (i.e., a division--free ordinary arithmetic circuit over $\Z$) of size $n^{O(1)}$ and claims that $\Gamma_n$ evaluates the polynomial represented by $n$. 

The task of the user is now to check this claim in uniform, bounded probabilistic or non--uniform polynomial time, namely in time $n^{O(1)}$.

In the case of the Pochhammer polynomial and the permanent a suitable protocol exists. This can be formulated as follows.

\begin{proposition}
The languages
\[
\begin{array}{rcl}
\mathcal{L}_{\text{Poch}}&:=&\{(n,(\Gamma_j)_{0\leq j\leq n}); \ws n\in\N, \ws\Gamma_j \text{\ is\ for\ $0\leq j\leq n$\ }\\
                  &  &\ws \ws \text{a\ division--free\ labelled\ directed\ acyclic\ graph\ evaluating\ $T^{\underline{2^j}}$}\} \\
\end{array}
\]
and 
\[
\begin{array}{rcl}
\mathcal{L}_{\text{Perm}}&:=&\{(n,\Gamma);\ws n\in\N \ws, \Gamma \text{\ is\ a\ labelled\ directed\ acyclic\ graph\ evaluating\ $\text{Perm}_n$\ } \}\\
\end{array}
\]
belong to the complexity class BPP and hence to $P/poly$ (here $n\in\N$ is given in unary representation). 
\end{proposition}   

\begin{proof}
We show only that $\mathcal{L}_{\text{Poch}}$ belongs to the complexity class $P/poly$. The proof that $\mathcal{L}_{\text{Poch}}$ belongs to BPP follows the same kind of argumentation and will be omitted here. The case of the language $\mathcal{L}_{\text{Perm}}$ can be treated analogously and we shall not do it here (compare \cite{KI2004}, Section 3).

Let $n\in\N$ and let $\Gamma$ be a division--free labelled directed acyclic graph with input $T$ and a single output node. Let $\Gamma'$ be the division--free labelled directed acyclic graph which is given by the following construction:
\begin{enumerate}
	\item[-] choose a labelled acyclic graph $\mu_n$ of size $n+O(1)$ with input $T$ and with $T-2^{2^{n-1}}$ as single final result
	\item[-] take the union $\ol{\Gamma}$ of the circuits $\Gamma$ and $\Gamma*\mu_n$ and connect the two output nodes of $\ol{\Gamma}$ by a multiplication node which becomes then the single output node of the resulting circuit $\Gamma'$. 
\end{enumerate}
From the polynomial identity $T^{\underline{2^{n}}}=T^{\underline{2^{n-1}}}(T)\ws\cdot\ws T^{\underline{2^{n-1}}}(T-2^{2^{n-1}})$ one deduces easily that $\Gamma'$ computes the polynomial $T^{\underline{2^{n}}}$ if and only if $\Gamma$ computes the polynomial $T^{\underline{2^{n-1}}}$.

For $0\leq j\leq n$ let $\Gamma_j$ be a division--free labelled directed acyclic graph with input $T$ and a single output node.

Suppose that in the previous construction the circuit $\Gamma$ is realized by the labelled directed acyclic graph $\Gamma_{n-1}$. Then one sees easily that $(n,(\Gamma_j)_{0\leq j\leq n})$ belongs to $\mathcal{L}_{\text{Poch}}$ if and only if the following conditions are satisfied:
\begin{enumerate}
	\item[$(i)$] the circuit $\Gamma_0$ computes the polynomial $T$,
	\item[$(ii)$] the circuits $\Gamma'$ and $\Gamma_n$ compute the same polynomial, 
  \item[$(iii)$] $(n-1,(\Gamma_j)_{0\leq j\leq n-1})$ belongs to $\mathcal{L}_{\text{Poch}}$. 
\end{enumerate}
Therefore, if condition $(ii)$ can be checked in non--uniform polynomial time, the claimed statement follows.

For $0\leq j\leq n$ let $L_j$ and $L$ be the sizes of the labelled directed acyclic graphs $\Gamma_j$ and $\Gamma'$ and observe that $L=2L_{n-1}+n+O(1)$ holds.

Let $P_{n-1}$ and $P$ be the final results of the circuits $\Gamma_{n-1}$ and $\Gamma'$. From \cite{CaGiHeMaPa03}, Corollary 2 we deduce that there exist $m:=4(L+2)^2+2$ integers $\gamma_1,\dots,\gamma_m\in\Z$ of bit length at most $2(L+1)$ such that the condition $(ii)$ above is satisfied if and only if
\begin{enumerate}
	\item[$(iv)$] $P_{n-1}(\gamma_1)=P(\gamma_1),\dots,P_{n-1}(\gamma_m)=P(\gamma_m)$	 
\end{enumerate}
holds.

From \cite{HM97} we infer that condition $(iv)$ can be checked by a nondeterministic Turing machine with advise in (non--uniform) time $O(L^3)=O((L_{n-1}+n)^3)$.

Applying this argument recursively, we conclude that membership of $(n,(\Gamma_j)_{0\leq j\leq n})$ to $\mathcal{L}_{\text{Poch}}$ may be decided in non--uniform time $O(\sum_{0\leq j\leq n} (L_j+j)^3)$ and therefore in polynomial time in the input size. Hence the language $\mathcal{L}_{\text{Poch}}$ belongs to the complexity class $P/poly$. The proof of the stronger result, namely $\mathcal{L}_{\text{Poch}}\in \text{BPP}$, is similar.       
\end{proof}

\enter
Finally we observe that for $n\in\N$ the Pochhammer polynomial $T^{\ol{2^n}}$ is the associated elimination polynomial of the particular problem instance, given by $T:=0$, of the flat family of zero--dimensional elimination problems $G_1^{(n)}=0,\dots,G_n^{(n)}=0,H^{(n)}$, which we considered in Section \ref{ 6.5.1 flat}. 

From the point of view of effective elimination, the sequence of Pochhammer polynomials becomes discussed in \cite{HeintzMorgenstern93} (see also \cite{SS95}). From the point of view of factoring integers, Pochhammer polynomials are treated in \cite{Lipton94}.

\subsubsection{Final comment}

Let us mention that our approach to effective elimination theory, which led to Theorem \ref{def: main theorem model} and preliminary forms of it, was introduced in \cite{HMPW98} and extended in \cite{GH01} and \cite{CaGiHeMaPa03}.

The final outcome of our considerations in Sections \ref{ 6.5.1 flat} and \ref{Comments on complexity models for geometric elimination} is the following: neither mathematicians nor software engineers, nor a combination of them will ever produce a practically satisfactory, \emph{generalistic} software for elimination tasks in Algebraic Geometry. This is a job for \emph{hackers} which may find for \emph{particular} elimination problems \emph{specific} efficient solutions.




\newcommand{\etalchar}[1]{$^{#1}$}


\end{document}